\documentclass[11pt,lot,lof]{duthesis}%lot - list of tables lof- list of figures los- list of symbols

\usepackage{multirow}
\usepackage{amscd,amsmath,amssymb,amsthm}

\usepackage{algorithm} %format of the algorithm
\usepackage{algorithmic} %format of the algorithm
\newtheorem{prop}{Proposition}
\usepackage{float}
\usepackage{graphicx}
\usepackage{booktabs}

\usepackage{cite}

{}
{}
{}

\title{Safe and Efficient Intersection Control of Connected and Autonomous Intersection Traffic}
\submitted{Januray, 2018}
\author{Qiang Lu}
\advisor{Mohammad H. Mahoor}
\coadvisor{Kyoung-Dae Kim}

\abstract
{
In this dissertation, we address a problem of safe and efficient intersection crossing traffic management of autonomous and connected ground traffic. Toward this objective, an algorithm that is called the Discrete-time occupancies trajectory based Intersection traffic Coordination Algorithm (DICA) is proposed. All vehicles in the system are Connected and Autonomous Vehicles (CAVs) and capable of wireless Vehicle-to-Intersection communication. The main advantage of the proposed DTOT-based intersection management is that it enables us to utilize the space within an intersection more efficiently resulting in less delay for vehicles to cross the intersection. In the proposed framework, an intersection coordinates the motions of CAVs based on their proposed DTOTs to let them cross the intersection efficiently while avoiding collisions. In case when there is a collision between vehicles' DTOTs, the intersection modifies conflicting DTOTs to avoid the collision and requests CAVs to approach and cross the intersection according to the modified DTOTs. 
We then prove that the basic DICA is deadlock free and also starvation free. We also show that the basic DICA has a computational complexity of $\mathcal{O}(n^2 L_m^3)$ where $n$ is the number of vehicles granted to cross an intersection and $L_m$ is the maximum length of intersection crossing routes. 
To improve the overall computational efficiency of the algorithm, the basic DICA is enhanced by several computational approaches. The enhanced algorithm has the computational complexity of $\mathcal{O}(n^2 L_m \log_2 L_m)$.

Next we addressed the problem of evacuating emergency vehicles as quickly as possible through autonomous and connected intersection traffic in this dissertation. The proposed intersection control algorithm Reactive DICA aims to determine an efficient vehicle-passing sequence which allows the emergency vehicle to cross an intersection as soon as possible while the travel times of other vehicles are minimally affected. When there are no emergency vehicles within the intersection area, the vehicles are controlled by DICA. When there are emergency vehicles entering communication range, we prioritize emergency vehicles through optimal ordering of vehicles. Since the number of possible vehicle-passing sequences increases rapidly with the number of vehicles, finding an efficient sequence of vehicles in a short time is the main challenge of the study. A genetic algorithm is proposed to solve the optimization problem which finds the optimal vehicle sequence that gives the emergency vehicles the highest priority.

We verify the efficiency of the proposed approaches through simulations using an open-source traffic simulator, called the Simulation of Urban Mobility (SUMO). 
The simulation results show that DICA performs better than another existing intersection management scheme: Concurrent Algorithm in \cite{kim2014mpc}. 
The overall throughput as well as the computational efficiency of the computationally enhanced DICA are also compared with those of an optimized traffic light control.
The efficiency of the proposed Reactive DICA is validated through comparisons with DICA and a reactive traffic light algorithm. The results show that Reactive DICA is able to decrease the travel times of emergency vehicles significantly in light and medium traffic volumes without causing any noticeable performance degradation of normal vehicles.
}

\acknowledgements
{
First of all, I would like to extend my sincere gratitude to my advisor Dr. Mohammad H. Mahoor, and co-advisor Dr. Kyoung-Dae Kim for their instructive advices and useful suggestions on my thesis. I am deeply grateful of their help in the completion of this thesis. I learned a lot knowledge from their excellent performance in teaching and research capabilities as well as high motivation and responsibility to students.

I am also deeply indebted to all the other professors and staff like Dr. Jun Jason Zhang, Dr. Kimon P. Valavanis, Dr. Amin Khodaei and Molly Dunn, etc. for their direct and indirect help to me. 

Special thanks should go to my friends Huaiguang Jiang, Qiao Li, etc. who have helped me a lot during my PhD. 

Finally, I am indebted to my parents, my sister and my girlfriend for their continuous support and encouragement.
}

\dedication{Dedicated to my parents}

\begin{document}

\chapter{Introduction} \label{chapter:intro}% introduction
\section{Introduction} \label{sec:intro}
Recent years, due to the rapidly increasing demand for transportation from the larger and larger population, roads have become more and more congested. Careful city planning can usually alleviate such transportation problems, but the unexpected growth of population and vehicle usage leads to persistent congestions. 
As efficient ways to address congestion problems, self-driving vehicles and autonomous transportation systems have attracted a lot of research and development efforts from academia, industry, and government. 
For example, during the mid-1990s the California PATH (Partners for Advanced Transportation Technology) launched the Automated Highway System (AHS) program \cite{horowitz2000control} and the US DARPA (Defense Advanced Research Projects Agency) held a series of autonomous vehicle challenges during the 2000s \cite{darpa} to take advantages in sensing and computation technologies. Also, many companies have already made decisions to hugely invest in developing their own self-driving vehicles or vehicles with advanced driving assistance systems \cite{bengler2014three}. 
However, despite many recent successful road testing results of several self-driving cars such as Google driverless car \cite{markoff2010google}, it is hard to argue that the overall system-wide traffic safety, as well as throughput, will be improved substantially when we have a few Connected and Autonomous Vehicles (CAVs) among all other conventional vehicles. 
In fact, the potential of autonomous vehicles in terms of the traffic efficiency and safety will be unleashed when most cars on roads are autonomous and connected \cite{Tientrakool2011p}. 
Thus, in addition to many efforts to make today's traffic more efficient by improving utilization of traditional traffic infrastructure such as the work presented in \cite{chen}, we believe that it is also very important to develop traffic control algorithms that take advantages of connectivity and autonomy of autonomous vehicles to prepare for the next generation transportation system. 
However, while there have been many efforts toward this direction, the development of safe and efficient autonomous transportation systems is still at its early stage. In this paper, among many research problems like vehicle path planning \cite{konar}, autonomous parking control \cite{li2003autonomous}, collision avoidance \cite{mammeri,colombo}, relation between occupant experience and intersection capacity \cite{le2015autonomous}, intersection management of mixed traffic \cite{onieva2015multi}, traffic coordination with power \cite{jiang2017traffic}, etc. that should be addressed toward this objective, we are particularly interested in addressing a problem of safe and efficient intersection crossing traffic management of autonomous connected traffic. Compared with highways, road intersections are more complicated places where accidents are more likely to happen and become the bottleneck of traffic performance improvement.

In literature, there are a number of notable results for autonomous intersection crossing traffic management. 
In reference \cite{kowshik2011provable}, a scheme that consists of a time-slot allocation intersection crossing algorithm, and an algorithm for updating failsafe maneuvers of each vehicle so as to avoid collisions while crossing an intersection was proposed. 
In \cite{lee2012development}, Lee et al. proposed an algorithm, called the Cooperative Vehicle Intersection Control (CVIC), which manipulates every individual vehicle's driving motion by providing them proper acceleration or deceleration rate so that vehicles can cross the intersection safely. 
Wu et al. \cite{wu2012cooperative} introduced a new intersection traffic management framework that is formulated as a combinatorial optimization problem and solved the problem approximately using the ant colony system algorithm \cite{dorigo1996ant}. 
Fei Yan et al. \cite{yan2013autonomous} combined vehicles whose routes are compatible with each other into mini groups and obtained an efficient vehicle passing sequence by their proposed genetic algorithm. A nonlinear programming formulation for autonomous intersection control was developed in \cite{zhu2015} where the nonlinear constraints were relaxed by a set of linear inequalities. 
Kyoung-Dae Kim and P.R. Kumar \cite{kim2014mpc} developed a Model Predictive Control (MPC) framework which dynamically generates a sequence of collision-free motions. Their paper considered two algorithms, a simple First In First Out (FIFO) Crossing algorithm and a Concurrent Algorithm, and shown the corresponding system-wide safety and crossing traffic liveness.
A polling-systems-based algorithm was proposed in \cite{miculescu2016} with provable guarantees on safety and performance. In the algorithm, a rigorous upper bound is provided for the expected wait time. 
Most of them are centralized approaches in which control decisions are made typically by a central agent. Decentralized intersection control approaches have also been proposed in literature. For example, \cite{malikopoulos2016decentralized} formulated a decentralized framework whereby each autonomous vehicle minimizes its energy consumption under the throughput-maximizing timing constraints and hard safety constraints to avoid rear-end and lateral collisions. A complete analytical solution of the decentralized problems was presented in the paper. These approaches are similar in that they all ensure safety within an intersection by preventing vehicles with conflicting intersection crossing routes from being inside the intersection at the same time.

To further improve the overall intersection control performance, some researchers studied algorithms that allow conflicting vehicles exist inside an intersection simultaneously. To achieve this objective, most approaches discretized the intersection space into grid cells so that vehicles of conflicting routes can exist at the same time within an intersection but not within a same cell.
Kurt Dresner and Peter Stone \cite{dresner2008multiagent} presented a reservation-based approach Autonomous Intersection Management (AIM) which treats vehicles and intersections as autonomous agents in a multiagent system. In AIM, vehicles request and receive time slots from the intersection during which they may pass. However, no global coordination is made for crossing vehicles to obtain optimal traffic flow. Following AIM, similar researches and improved approaches were proposed like expediting the crossing of emergency vehicles \cite{dresner2006human}, determining the priority of requests using auctions \cite{carlino2013auction, vasirani2012market}, etc. To improve previous AIM models, Levin et al. \cite{levin2017optimizing} studied the approach to choose the optimal subset of vehicles to move at every time step by formulating an Integer Program (IP). Arbitrary objective functions can be admitted by the IP so a more general class of policies can be applied to optimize the order of vehicles to cross the intersection. 
Representative centralized approaches also include auction-based intersection managements proposed in \cite{carlino2013auction, vasirani2012market}. 
A series of decentralized approaches only based on vehicle-to-vehicle (V2V) communications were proposed in \cite{Azimi2011p, Azimi2012p, Azimi2013p, Azimi2014} by Reza Azimi et al. In their papers, the intersection area is considered as a grid which is divided into small cells. Each cell in the intersection grid is associated with a unique identifier. One of their advanced intersection protocols named AMP-IP (Advanced Maximum Progression Intersection Protocol) allows the lower-priority vehicle to go ahead and cross the conflicting cell before the arrival of the higher-priority vehicle if there is enough time for lower-priority vehicle to clear the cell.
Roughly speaking, these approaches are all based on the grid cell partitioning of an intersection space. 
In \cite{dresner2008multiagent}, the effect of the grid cell granularity on the computational efficiency of an intersection traffic management framework such as AIM was studied. Clearly, higher granularity gives more flexibility for better traffic throughput. However, the computational complexity increases proportionally to the square of the granularity. On the other hand, when the cell size becomes large for better computational efficiency, one can see that the intersection space is not utilized efficiently resulting in lower traffic throughput. 
Therefore, to overcome this trade-off issue between the granularity and computational efficiency of an algorithm, it might be a good alternative approach to utilize each vehicle's actual occupancy instead of grid cells to improve the overall traffic throughput. And this has motivated our research on this topic.

As an approach to address the above-mentioned granularity issue, we propose a novel intersection traffic management scheme based on the idea of the \textit{Discrete-Time Occupancies Trajectory (DTOT)}. Conceptually, a DTOT is a discrete-time sequence of a vehicle's actual occupancy within an intersection space. The proposed DTOT-based Intersection Control Algorithm (DICA) allows the flexibility that each vehicle can choose its path as well as motions along the path that a CAV wants to take to cross an intersection. A CAV who is approaching an intersection will check whether it is the \textit{Head Vehicle} on its lane. If it is, the CAV will propose its request to the intersection. From the request, a vehicle's enter time and exit time to the intersection, route and other detailed information of passing the intersection can be found. Then the intersection responds with a DTOT to the vehicle to avoid collisions and improve the overall intersection throughput. The DTOT corresponds to a trajectory with achievable speed and acceleration for the CAV. If the CAV is not the Head Vehicle, it will just follow other cars. Thus, the management scheme is only dealing with head vehicles which reduces largely the communication needs for vehicles and the computational complexity of the central control agent. It is assumed that a CAV always want to go through an intersection as quickly as possible. So CAVs in our algorithm always try their best to reach maximum allowed speed within the intersection.

Then, we provide an in-depth analysis of the original DTOT-based Intersection Coordination Algorithm (DICA) to show that it satisfies the liveness property in terms of deadlock as well as starvation issues and also to derive the overall computational complexity of the algorithm. Another contribution is that we propose several computational approaches to improve the overall computational efficiency of the DICA and also enhance the algorithm accordingly so that it can be operated in real-time for autonomous and connected intersection crossing traffic management. We also present simulation results that show the improved computational efficiency of the enhanced algorithm and the overall throughput performance in comparison with that of an optimized traffic light control.

Human lives and the amount of financial loss highly depend on the response time (from the time the emergency service is called to the time help is offered) of emergency vehicles. The travel time of emergency vehicles to the accident scene is critical to the response time. So it is very useful and helpful to reduce travel time of emergency vehicles on roads, especially on intersections where congestions are more likely to happen. The survival chance of injured people in an accident falls sharply if they reach the operating table later than 60 minutes after the accident \cite{martinez2010emergency}. Hence, shortening the travel times of crossing intersections for emergency vehicles will help to save lives. In reality, the current way to handle emergency vehicles is similar to using Vehicle-to-Vehicle (V2V) communication (siren and lights) to warn non-emergency vehicles on roads to yield to the emergency vehicle. Some drivers cannot respond quickly to the warnings which may result in additional time delay for emergency vehicles and even serious accidents. In most of existing approches \cite{oliveira2005making, viriyasitavat2012priority, dresner2006human} for emergency vehicles, the travel times of non-emergency vehicles will be affected significantly and the advantages of autonomous vehicles are not fully taken. 

Following computationally enhancing DICA, we extend DICA approach to include emergency vehicles in the traffic to be controlled. Our goal of intersection control is to let emergency vehicles cross intersections as fast as possible while maintaining adequate traffic performance. In this dissertation, we assume that emergency vehicles are taking normal routes which means that they will not travel in a wrong lane. A genetic algorithm is proposed to find the optimal passing sequence of vehicles whose trajectories can be rearranged. This optimal sequence aims to make the emergency vehicles cross the intersection in the fastest way. When there is no emergency vehicle inside the communication region, vehicles are controlled by DICA. Thus, the proposed algorithm is called Reactive DICA.
Among many sequence forming approaches \cite{carlino2013auction,qian2014priority,Wu2012p207,yan2013autonomous} in literature, \cite{yan2013autonomous} is the most similar approach with ours which also proposed a genetic algorithm to form vehicle sequences. 
However, unlike the approach proposed in this dissertation, they are essentially not allowing vehicles with conflicting routes to be inside the intersection at the same time.

\section{Organization of the Dissertation}

The remainder of this dissertation is organized into five chapters. 
In Chapter \ref{chapter:lr}, we review literature on autonomous transportation especially on autonomous intersection control.
In Chapter \ref{chapter:dica}, we introduce the assumptions and interactions between vehicle and intersection in DICA. Then we explain the functions in DICA in detail and analyze the liveness of DICA. Finally, DICA is simulated and compared with Concurrent Algorithm \cite{kim2014mpc} to validate its performance.
In Chapter \ref{chapter:e-dica}, we analyze the computational complexity of DICA in detail first and then proposed several computational techniques to improve overall computational efficiency. The improved performance is validated through comparisons with original DICA and an optimized traffic light control algorithm.
Chapter \ref{chapter:r-dica} describes Reactive DICA especially the proposed genetic algorithm in detail. Simulation results of Reactive DICA are compared with a reactive traffic light algorithm to show its performance.
%In Chapter \ref{chapter:dica-mip}, we formulate the trajectory optimization problem with each constraint explained detailedly. Then we propose DICA-MIP algorithm incorporating the MIP optimization approach. The efficiency of DICA-MIP is evaluated against DICA through extensive simulations.
%
Finally, Chapter \ref{chapter:conclusion} draws the conclusion of the dissertation and provides potential future work.

\chapter{Literature Review} \label{chapter:lr}% literature review
%Autonomous transportation system

\section{Control of Highway Systems}
Highway congestion has brought intolerable burdens to many urban residents. Governments have made many attempts including building more roads and raising tolls or other related taxes, promoting public transportation or better vehicle occupany (carpooling) and developing a high-speed communications network that in many ways to reduce travel demand \cite{varaiya1993smart}. 
In the meantime, researches of Intelligent Vehicle/Highway System (IVHS) provide new ways to ease the congestion on highways. There are mainly two types of methods to increase highway capacity while ensure safety, \textsl{1}. vehicle platooning, represented by the AHS which had been deeply developed at the University of California Partners for PATH program, in cooperation with the State of California Department of Transportation (Caltrans) and the United States Federal Highway Administration (FHWA) since 1989 \cite{horowitz2000control, varaiya1993smart}. \textsl{2}. algorithms and protocols for each individual vehicle to drive fast and safely, represented by references \cite{falcone2007predictive, reghelin2012using}.

\subsection{Vehicle Platooning}
Organizing traffic into platoons which are groups of up to $20$ tightly spaced cars can increase highway capactiy greatly. $8000$ vehicles per hour per lane could be achieved while today's highways with manually controlled vehicles only have $2000$ capacity \cite{horowitz2000control}. Platooning also has a benefit of reducing aerodynamic drag because vehicles are tightly spaced. People cannot react quickly enough to drive safely with such small headways, so every vehicle in a platoon should be automated.   
A car joins a platoon when it enters the AHS and exits as a one-car platoon or part of an exiting platoon when it approaches its destination. Inside AHS, the car is controlled automatically by computers and a series of maneuvers \cite{varaiya1993smart} are executed by the car including splitting from and joining platoons and lane changing to navigate through the highway network.

Reference \cite{varaiya1993smart} introduces the finite state machine method which is used for the execution of maneuvers. Join control law, platoon leader control law and other laws are given in \cite{horowitz2000control} to ensure collision avoidance of AHS. What is noteworthy is that the on-board vehicle control system is a
hybrid control system which consists of a discrete event dynamical system (the coordination layer) and a continuous-time dynamical system (the regulation and physical layers). 

The platoon concept is useful but does not take full advantage of the microscope centralized possibilities because it does not consider the interaction of all vehicles \cite{reghelin2012using}. In a platoon, only leaders (and free agents) can initiate maneuvers, while followers maintain platoon formation at all times \cite{horowitz2000control}.

\subsection{Algorithms and Protocols for Individual Vehicles}
For an individual autonomous vehicle, many algorithms and protocols have been proposed for its longitudinal and lateral motion on a highway system while the vehicle's safety is ensured.

\subsubsection{Longitudinal Control}
The longitudinal control of a vehicle is to control the vehicle's speed and its adaptation to road features by adjusting the throttle \cite{ruan2005study} and the brake pedal as needed \cite{Naranjo2003throttle}. 
A longitudinal control system will handle issues like nonlinear vehicle dynamics, operations of vehicles from high-speed cruising to a complete full-stop, and other operations under the communication constraints \cite{rajamani2000demonstration}. Following we describe several different controllers that are designed to ensure safety of longitudinal motion of each vehicle.

Reference \cite{kowshik2011provable} gives the theorem for perpetual safety of two cars on a lane that if each car makes worst-case assumption about the car immediately in front of it, then the safety of multiple cars on a lane can be guaranteed.
Reference \cite{kim2013collision} uses a simple discrete-time kinematic model to represent the longitudinal motion of a vehicle, $x_{t+h} = f(x_t,v_t) = x_t+v_th$, where $x_t, v_t, h$ correspond to $x$-axis (forward direction) position, linear velocity of a vehicle at time t, and sampling period respectively. The paper formulated an MPC problem for the vehicle to ensure the generation of safety-guaranteed motions.

\subsubsection{Lateral Control}
Typically lane changing provides a maneuver for a fast vehicle to pass a slow vehicle, which can be observed everywhere on the highway \cite{enache2009driver}.  
A framework of lane change decision-making under typical urban driving conditions was proposed by Gipps \cite{gipps1986model}, which includes the effects of traffic lights, obstacles and types of surrounding vehicles. Taking into account the potential conflict objectives and assuming logical driver behavior, the model focuses on the decision-making process.

Reference \cite{kim2013collision} constructs the MPC problem for Lane Change of a vehicle on a multi-lane traffic and introduces the Lane Change Protocol that a vehicle can initiate its lane change action only when its state satisfies certain relations with neighboring vehicles. The paper also proposes the Yield Protocol which is
incorporated into the MPC motion planning framework to ensure the liveness of Lane Change.
A novel MPC-based approach for active steering control design is presented in \cite{falcone2007predictive}. Experimental results showed that a vehicle under MPC feedback policy is capable of stabilizing with a speed up to $21$ m/s on slippery surfaces such as snow covered or icy roads. The paper designed and experimentally tested three types of MPC controllers including a nonlinear MPC, an LTV MPC, and an LTV MPC controller of low order.

\section{Autonomous Intersection Control}
Intersection traffic control is more of importance since it is more likely to have accidents at intersections compared with highways. Besides accidents, the trip delays from the impact of intersections also lead to waste of human and natural resources. To tackle intersection control problems, many solutions have been proposed in the literature. The main problem is to provide an intersection control algorithm with guaranteed safety and improved efficiency.

% heuristic
\subsection{Centralized Control}
Kyoung-Dae Kim \cite{kim2013collision} proposes a framework for intersection management that integrates decisions made by the intersection in the discrete domain for vehicle ordering with decisions made by each vehicle in the continuous domain for its safety-guaranteed motion. The proposed intersection control algorithms (FIFO and Concurrent Algorithm) and protocols ensure that the intersection is safe and live. 

Reference \cite{kowshik2011provable} designs smart intersections where vehicles negotiate the intersection crossings through an interaction of centralized and distributed decision making. The route that was taken by a car $i$ is described by an ordered pair R($i$)=(O($i$),D($i$)), of origin and destination, respectively. A scheme is proposed that consists of a time-slot allocation intersection crossing algorithm, and an algorithm for updating failsafe maneuvers of each vehicle so as to avoid collisions while crossing an intersection.

An interesting result about collision avoidance at intersection is shown in \cite{colombo}, that it is an NP-hard problem to check the membership in the maximal controlled invariant set, which is the largest set of states for which there exists a control that avoids collisions. An algorithm with provable error bounds is proposed to solve such a problem approximately in polynomial running time. The paper provides and proves a tight bound on the approximated solution. 

The most existing evasion plans are not dealt with the vehicles already in the intersection which suddenly stops because of tire blowout or other mechanical failures. \cite{ZHU2013p10013} provides appropriate plans for such situations. According to the width of the intersection, two ways to avoid collisions are offered for the vehicles in intersection. And the paper proves the validity of the methods by empirical study.

\subsection{Decentralized Control}

Since 2011, researcher Reza Azimi has proposed a series of reliable intersection protocols that use only vehicle-to-vehicle (V2V) communications. The installation of a centralized infrastructure at each intersection is impractical due to the high overall system cost. Azimi advocates the use of V2V communications and distributed intersection algorithms that runs in each vehicle. He designed the vehicular network protocols that integrate mobile wireless communications standards such as Dedicated Short Range Communications (DSRC) and Wireless Access in a Vehicular Environment (WAVE). 
Collision Detection Algorithm for Intersections (CDAI) which uses a priority-based policy, Stop-Sign Protocol (SSP), Throughput Enhancement Protocol (TEP) and Throughput Enhancement Protocol with Agreement (TEPA) are proposed in \cite{Azimi2011p}. SSP is similar to actual stop-sign situation that every vehicle must stop before an intersection when there is a stop-sign. Using TEP, vehicles stop at the intersection only if the CDAI predicts a collision and assigns a lower priority to them based on the message it receives from all vehicles at the intersection. TEPA is built on TEP and is explicitly designed to handle lost V2V messages.

In \cite{Azimi2012p}, Azimi et al. define an intersection as a perfect square box that pre-defines the entry and exit points for each lane to which it is connected. The intersection area is discretized as grid cells and each cell is associated with a unique identifier. 
More advanced protocols CC-IP (Concurrent Crossing-Intersection Protocol) and MP-IP (Maximum Progression-Intersection Protocol) which are improved from \cite{Azimi2011p} are proposed. In CC-IP, CROSS message is broadcasted when a vehicle enters the intersection area, other vehicles can simultaneously pass through the intersection if they detect no potential collision with the vehicle already crossing the intersection, otherwise the vehicle stop before the intersection area. However, in MP-IP, the vehicle which has the lower priority uses CDAI and finds out the first common cell (trajectory intersecting cell, abbreviated as TIC) with the crossing higher-priority vehicle. Then, instead of not entering the intersection as in CC-IP, the vehicle stops at the cell just before entering the TIC.

Reference \cite{Azimi2013p} illustrates more advanced protocols about intersection using vehicular networks based on the work of \cite{Azimi2011p,Azimi2012p}. After detailed describing of MP-IP and AMP-IP (Advanced Maximum Progression Intersection Protocol), the paper proves that these protocols avoid deadlock situations inside the intersection area. The improvement of AMP-IP is that, the protocol allows lower priority vehicles to advance and cross the conflicting cell before the higher priority vehicle arrives.
A Safety Time Interval of $2$s is used to help lower-priority vehicles decide whether they can go through the conflicting cell safely. Simulation results show that the latest V2V intersection protocol AMP-IP yields over 85\% overall performance improvement over the common traffic light models.
In \cite{Azimip}, Azimi also investigates the use of their proposed V2V-intersection protocols for autonomous driving at roundabouts. The improvement in safety and throughput when the intersection protocols are used to traverse roundabouts is also quantified.

\subsection{Optimization Approaches}
Machine learning is an emerging aspect in traffic management and control, which also show its capability in different areas such as power system \cite{gu2014sta123tistical, jiang2014sync123hrophasor}, medical science \cite{jiang2013time}, renewable energy \cite{jiang2012fa123ult,jiang2014fau123lt,jiang2014sp123atial,jiang2016short,jiang2016sh123ort,
gu2016knowledge,ding2016automa123tic,yang2017sh123ort,jiang2017b123ig,jiang2017sh234ort,zhang2017consum234ption}. Similarly, considering the optimization, the convex optimization also illustrate its capability in different areas \cite{gu2017lowefad,gu2017cha32432nce,jiang2015sync123hrophasor,banavar2015ov123erview,jiang2015pmu}.

In recent years, many researches have been done to improve intersection control performance by using optimization approaches. Reference \cite{zhu2015} developed a novel linear programming formulation for autonomous intersection control in which the nonlinear constraints were relaxed by a set of linear inequalities. While the objective function of the optimization problem in \cite{zhu2015} involves the travel time, other studies \cite{dai2016, kamal2015} are trying to solve similar control problem using an objective function with multiple criteria like safe speeds and accelerations while avoiding collisions. 
Compared with centralized control approaches, infrastructure support is not needed in decentralized control. In the approach \cite{wu2015} proposed by Wu et al., the estimated arrival time is shared wirelessly among vehicles to obtain the best passing sequence. The problem of coordinating online a continuous flow of connected and autonomous vehicles crossing two adjacent intersections was formulated as a decentralized optimal control problem in \cite{zhang2016}. The solution gives the optimal acceleration/deceleration for each vehicle at any time to minimize fuel consumption.
Some researchers also studied the control mechanism when there is only part of the traffic are autonomous vehicles \cite{guler2014}.

\subsection{Emergency Vehicle Handling}
Many studies have been done to allow emergency vehicles to have a faster travel across intersections. Based on MAS (Multi-Agent System), \cite{oliveira2005making} introduced a state machine for the intersection controller to change traffic signal status according to lane occupation when an emergency vehicle is approaching. Some researchers have explored the priority evacuation of emergency vehicles under an autonomous and connected traffic environment. Wantanee Viriyasitavat and Ozan K. Tonguz proposed an intersection control system that only uses Vehicle-to-Vehicle (V2V) communication to give emergency vehicles priority of crossing \cite{viriyasitavat2012priority}. The paper proposed that at an intersection, a leader should be elected from all approaching vehicles to serve as the temporary traffic light infrastructure and stop at the intersection to coordinate the traffic. The green signal is always given to the lane of detected emergency vehicles and through coordination ``green-wave'' signals are displayed for the emergency vehicles to let them move at a faster speed. Kurt Dresner and Peter Stone proposed a simple way to deal with emergency vehicles under their intersection control framework AIM (Autonomous Intersection Management) \cite{dresner2006human}. Their algorithm only grants reservations to vehicles in the lanes that have approaching emergency vehicles which allow the emergency vehicle to continue on its way relatively unhindered. 
%
%
%\subsection{Mixed Traffic}
%An interesting problem of semiautonomous multivehicle safety has been studied by Rajeev Verma and Domitilla Del Vechhhio. They mainly made research about the problem of collision avoidance between autonomous vehicles and human-driven vehicles. And a formal hybrid control approach to design semiautonomous multivehicle systems that are guaranteed to be safe is given \cite{Verma2011p44}. 
%
%\subsection{Multiple Intersections}
%\subsection{Simulation Platforms}
 %input the file 

\chapter{DTOT-based Intersection Traffic Management} \label{chapter:dica}% DICA
As introduced in Chapter 1, DTOT is the abbreviation of Discrete-Time Occupancy Trajectory which is a sequence of a vehicle's actual occupancy within an intersection space. Based on the concept of DTOT, we propose a novel intersection control algorithm named DICA (DTOT-based Intersection traffic Coordination Algorithm) in this chapter.

\section{Assumptions} \label{sec:dtot}

\begin{figure} [!t]
	\centering
	\includegraphics[width=4.in]{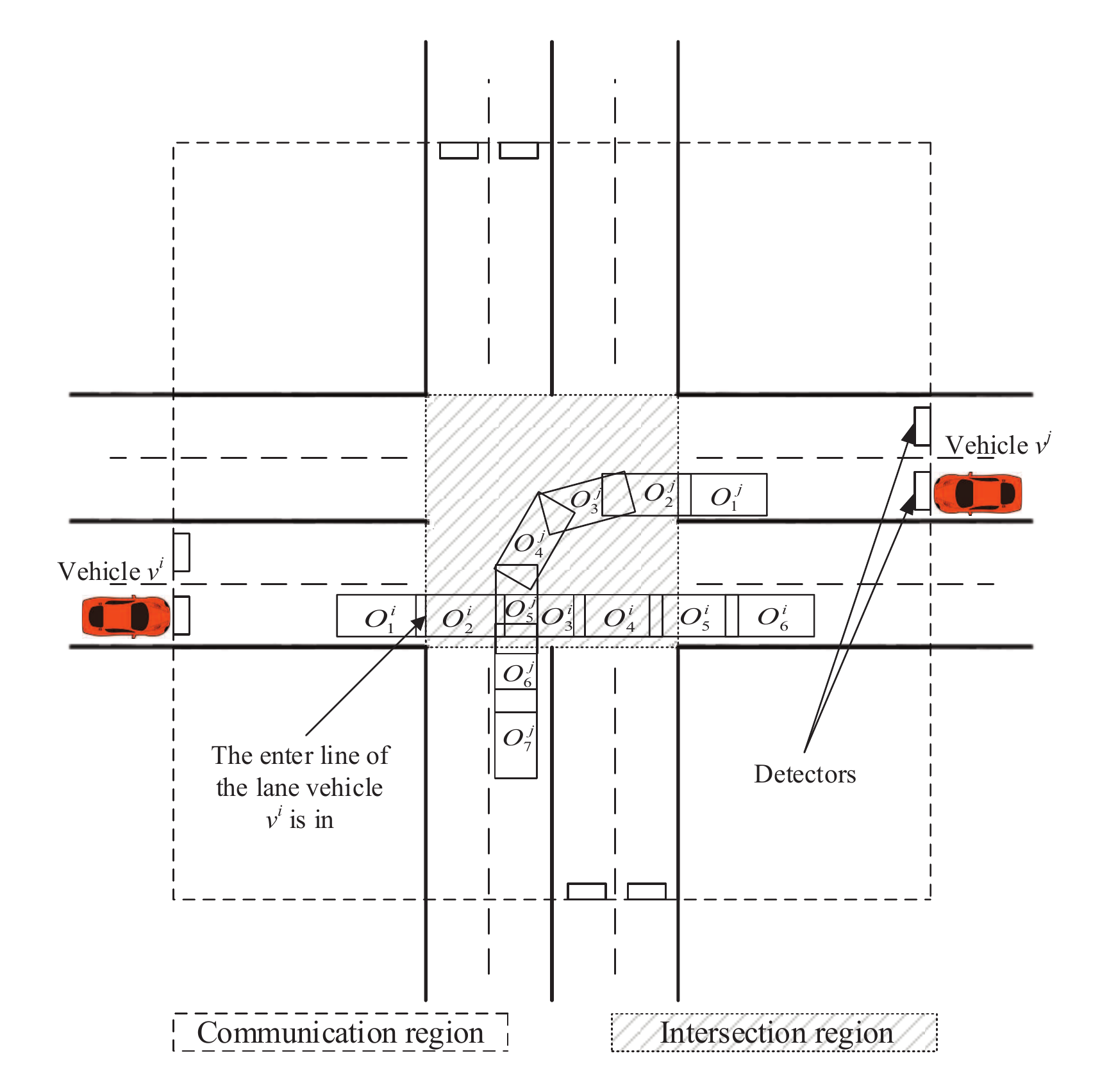}
	%\resizebox{3in}{!}{\includegraphics{figures/DTOT}}
	\caption[DTOTs of two conflicting vehicles]{DTOTs of two conflicting vehicles. ($O^p_q$ represents the $q$-th occupancy in a vehicle $v^p$'s DTOT. Note that occupancies in this figure are intentionally made very sparse for clear illustration purpose. DTOT starts with the occupancy in which the vehicle's front bumper first contacts the enter line of its lane of an intersection, and ends with the occupancy that the vehicle is completely out of the intersection region.)}
	\label{fig:Example}
\end{figure}

In this section, we introduce the basic idea and algorithm of DICA that is developed for autonomous and connected intersection crossing traffic in which all vehicles are Connected and Autonomous Vehicles (CAVs) and capable of wireless vehicular communication. 
We assume that an intersection has the wireless communication capability as well as a computation unit so that it can exchange information with vehicles and perform necessary computations to coordinate vehicles to cross the intersection safely. 
In DICA, there is no traffic light that controls the intersection crossing traffic. Instead, each vehicle communicates with the \emph{Intersection Control Agent} (ICA), to get permission to access the intersection. As shown in Figure \ref{fig:Example}, an intersection consists of two regions. The bigger region in the figure, which we call the \emph{communication region}, is defined by the wireless vehicular communication range. The smaller region in the figure, which we call the \emph{intersection region}, is the area within an intersection that is shared by all roads connected to the intersection. 
We also assume that each vehicle is equipped with an RFID (Radio Frequency IDentification) chip and there are detectors installed at the entrance of the communication region so that ICA can detect each vehicle's VIN (Vehicle Identification Number), the lane on which a vehicle is approaching an intersection, and the time when a vehicle enters the communication region. 
Since all vehicles are autonomous, we assume that each vehicle can obtain its position, speed, and the relative distance to an intersection precisely and also can avoid collisions with other vehicles autonomously when it is approaching an intersection. With regard to wireless vehicular connectivity, we only require information exchange between CAVs and ICA. Thus, there is no V2V communication.

\section{Interaction between ICA and a CAV}
A CAV is considered a \emph{head vehicle} in its lane if there are no vehicles in front of it or the vehicle which is immediately in front of it has begun to enter the intersection region. 
As shown in Figure \ref{fig:architecture}, the interaction between a CAV and ICA is initiated from the CAV, when it becomes a head vehicle, by sending a REQUEST message to reserve a sequence of spaces and times to cross the intersection. 
ICA knows whether a vehicle is a head vehicle or not according to the list of vehicles for each lane. Thus, a REQUEST message not from a head vehicle will be neglected by ICA. The list can be constructed in ICA since, as explained earlier, ICA knows each vehicle's VIN, the lane on which the vehicle is approaching, and the time when a vehicle passes a detector installed at the boundary of the communication region of an intersection. 
Each REQUEST message contains information that is necessary to reserve the space and time within the intersection region to cross the intersection such as (i) the VIN, (ii) the Vehicle Size (VS), and (iii) the Timed State Sequence (TSS).  The VS is simply the length and width of the vehicle and the TSS is the discrete time state trajectory of the vehicle starting from the entering moment of an intersection region to the moment when the vehicle crosses the intersection region completely.
Note that it is implicitly assumed that each discrete time state of a vehicle in TSS is also timed. This means that if a vehicle state $\mathbf{x}_t$ is given, then we can say that a vehicle possesses the state $\mathbf{x}$ at time $t$. For simplicity of our discussion, we assume that the state $\mathbf{x}$ of a vehicle consists of the $(x, y)$ coordinate of the vehicle's location and the orientation $\theta$. We also assume that, while it is possible that each vehicle can have different sampling period to generate its TSS, all vehicles use the same sampling period which is small enough to generate a close approximation of the vehicle's actual continuous motion within an intersection.

\begin{figure}[!htb]
	\centering
	\includegraphics[width=3.8in]{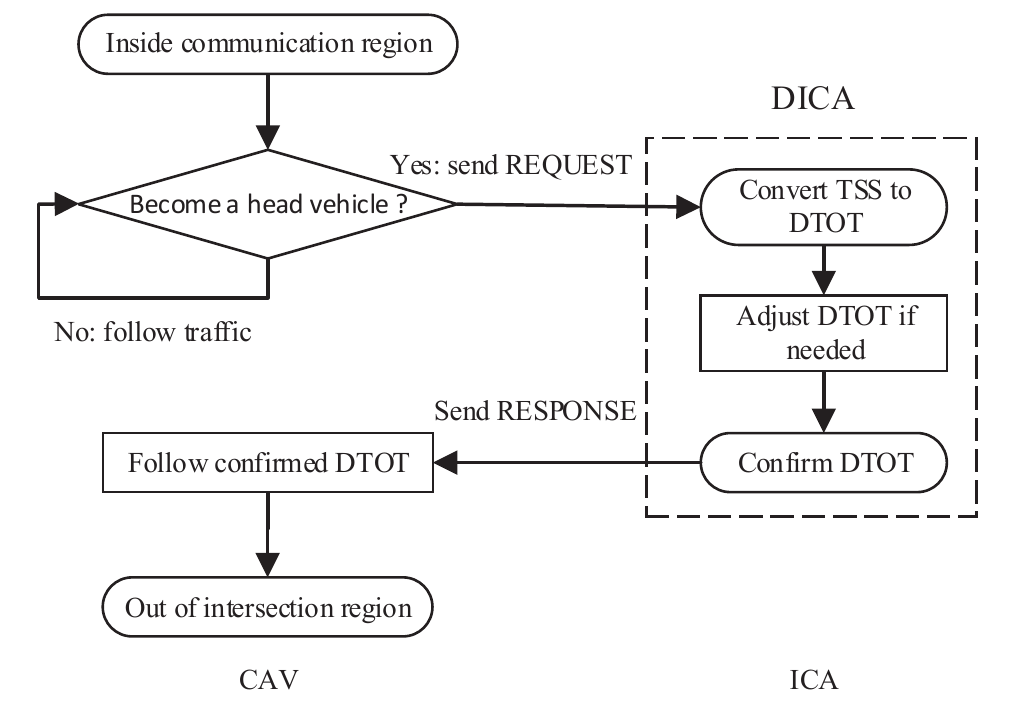}
	\caption{Interaction between a CAV and ICA.}
	\label{fig:architecture}
\end{figure}

ICA converts the TSS to the corresponding DTOT using the VS information which is also contained in the received REQUEST message. The DTOT is simply a sequence of timed rectangular spaces that a vehicle needs to occupy within an intersection region to cross the intersection.
Now, ICA uses all confirmed DTOTs to adjust the requested DTOT to avoid collisions if needed. ICA then converts the collision-free DTOT to TSS and sends it back to the vehicle using a RESPONSE message which contains (i) the VIN and (ii) TSS so that the vehicle can follow the confirmed DTOT to cross the intersection. More detailed explanation on how to process the requested TSS to generate a confirmed DTOT is presented in the following section. In the sequel, we say that a vehicle is a \emph{confirmed vehicle} if it has received a confirmed DTOT from ICA. And we assume that every vehicle is able to follow the confirmed DTOT precisely. In practice vehicles will have tracking errors to follow a given DTOT. To avoid potential collisions with other vehicles, we can increase the size of every occupancy in the DTOT by the upper bound of tracking errors. 
Since the focus of this chapter is to develop an algorithm for ICA for safer and higher throughput intersection crossing traffic, we simply assume that we have an ideal wireless vehicular communication performance such that all REQUEST and RESPONSE messages are exchanged correctly and timely. However, it is important to note that, despite such an ideal communication assumption, our DTOT-based algorithm can still be applicable in practice with small modifications of the algorithm to take into account the communication unreliability.
For instance, typically we may face two problems (i.e. package delay and lost) to handle the imperfect communications existing between CAVs and ICA in real situations. We could use the upper bound of the package delay to extend every occupancy in a DTOT which is safe for vehicles but a little bit conservative. For package lost problem, an ACK message can be added to confirm the delivery of REQUEST and RESPONSE messages whose details can be found in next Section. A CAV will send REQUEST again if it does not receive the ACK message from ICA. The same strategy could be applied to ICA and RESPONSE message.

\begin{algorithm}[!thb]   
	\caption{DICA (\textbf{D}TOT-based \textbf{I}ntersection traffic \textbf{C}oordination \textbf{A}lgorithm)}
	\begin{algorithmic}[1]
		%    	\caption{DTOT Intersection Control Algorithm}
		\STATE Let $\mathcal{S}$ be the set of confirmed vehicles and $n = \vert \mathcal{S} \vert$.
		\STATE Let $v^i$ be the vehicle to be considered for confirmation.
		%		\FOR {vehicle $i \in$ unconfirmed Head Vehicle set}
		%		\STATE
		\STATE Convert $TSS(v^i)$ to $DTOT(v^i)$ \label{convert}
		%		\STATE
		%		\STATE Based on $n$ confirmed vehicles:
		\STATE Call checkFV$(\mathcal{S}, DTOT(v^i)) \rightarrow DTOT(v^i)$ \label{FV}
		%		\STATE
		\STATE Call getCV$(\mathcal{S}, DTOT(v^i)) \rightarrow \mathcal{C}$ \label{getCV}
		%		\STATE
		\WHILE {$\mathcal{C} \neq \emptyset$} \label{while}
		\STATE Pop the first vehicle in $\mathcal{C} \rightarrow v^j$
		\STATE Call updateDTOT$(DTOT(v^i), DTOT(v^j)) \rightarrow DTOT(v^i)$ \label{update}
		\STATE Call getCV$(\mathcal{S}, DTOT(v^i)) \rightarrow \mathcal{C}$
		\ENDWHILE \label{endwhile}
		
		%		\STATE
		\STATE Store $DTOT(v^i)$ for vehicle $v^i$ \label{store}
		\STATE Convert $DTOT(v^i)$ to $TSS(v^i)$
		\STATE Send $TSS(v^i)$ to vehicle $v^i$ \label{send}
		%		\STATE
		%		\ENDFOR
		
	\end{algorithmic}
	\label{code:dica}
\end{algorithm}

\section{DTOT-based Intersection Traffic Coordination}
ICA processes a REQUEST message from a head vehicle according to the procedures shown in Algorithm \ref{code:dica} which we call the DTOT-based Intersection traffic Coordination Algorithm (DICA). As shown in the algorithm, we use TSS($v$) and DTOT($v$) to denote the TSS and DTOT for a vehicle $v$ respectively.  We also use $\mathcal{S}$ to denote the set of vehicles which have already been confirmed at the time when a REQUEST message is being processed. We say that two vehicles are \emph{space-time conflicting} if their trajectories are conflicting not only in space but also in time. More precisely, two vehicles are considered to be in space-time conflict in our algorithm when their DTOTs have at least one pair of occupancies that conflict in both space and time. We use another set $\mathcal{C}$ in Algorithm \ref{code:dica} to represent the subset of $\mathcal{S}$ which contains the set of vehicles whose confirmed DTOTs have space-time conflict with the DTOT of the vehicle that is currently being processed for confirmation. Vehicles in $\mathcal{C}$ are ordered in ascending order of a certain attribute of their confirmed DTOTs. To explain this attribute more clearly, let us consider a situation when DICA processes a vehicle $v^i$'s DTOT and there are two vehicles $v^j$ and $v^k$ in the set $\mathcal{C}$. Now let us suppose that DTOT($v^j$) starts to space-time conflict with DTOT($v^i$) from its $n$-th occupancy and DTOT($v^k$) starts to space-time conflict with DTOT($v^i$) from its $m$-th occupancy. If we use $O^p_{q}$ to denote the $q$-th occupancy within DTOT($v^p$) and $\tau(O^p_{q})$ be the time when the vehicle $v^p$ occupies $O^p_q$, then we say that, in this particular situation, $\tau(O^j_n)$ is the first time at which $v^j$ starts to collide with $v^i$. Similarly, $\tau(O^k_m)$ is the time at which $v^k$ starts to collide with $v^i$. In the sequel, this specific time instant for each vehicle in $\mathcal{C}$ is represented by the variable `\emph{firstTimeAtCollision}'.  
In this particular situation, $\tau(O^j_n)$ and $\tau(O^k_m)$ are denoted by $v^j.firstTimeAtCollision$ and $v^k.firstTimeAtCollision$, respectively. Vehicles in the set $\mathcal{C}$ are ordered according to this variable. Specifically, if $v^j.firstTimeAtCollision$ is earlier than $v^k.firstTimeAtCollision$, then $v^j$ gets higher priority than $v^k$ and vice versa. 
To see more clearly how the `\emph{firstTimeAtCollision}' is determined, we can consider an illustrative example shown in Figure \ref{fig:Example}. In the figure, DTOT($v^i$) and DTOT($v^j$) have space conflicts in \{$O^i_2, O^i_3$\} and \{$O^j_5, O^j_6,$\}. If we assume that these occupancies are also conflicting in time, then $v^j.firstTimeAtCollision$ with respect to the vehicle $v^i$ is $\tau(O^j_5)$. 

As shown in Algorithm \ref{code:dica}, when ICA receives a REQUEST message from a head vehicle $v^i$, it first converts the TSS($v^i$) into the corresponding DTOT($v^i$) using the vehicle's VS. Then ICA calls the function \verb+checkFV()+ to determine if there exist \emph{front vehicles} (See Section \ref{subsec:front} for more details about front vehicles.) that affect the vehicle $v^i$'s motion and also to adjust $v^i$'s DTOT if needed. Then the function \verb+getCV()+ is called to determine the set $\mathcal{C}$ which is the set of vehicles whose DTOTs are space-time conflicting with DTOT($v^i$). The \verb+updateDTOT()+ function adjusts DTOT($v^i$) appropriately so that DTOT($v^i$) avoids space-time conflict with other vehicle's DTOT. These two functions are iteratively called within the \verb+while+ loop until the set $\mathcal{C}$ becomes empty, which indicates that no vehicles in the set $\mathcal{C}$ will collide with the vehicle $v^i$. After DTOT($v^i$) is appropriately adjusted and confirmed that there is no space-time conflict with all other confirmed vehicles, then the confirmed DTOT($v^i$) is converted into TSS($v^i$). Finally, ICA sends the confirmed TSS($v^i$) back to the vehicle $v^i$ so that the vehicle can cross the intersection safely by following the confirmed DTOT.
In the following sections, we provide more detailed explanation on the subfunctions called within DICA.

\subsection{Collision Avoidance with Front Vehicles} \label{subsec:front}

\begin{figure}[!tb]
	\centering
	\includegraphics[width=2.8in]{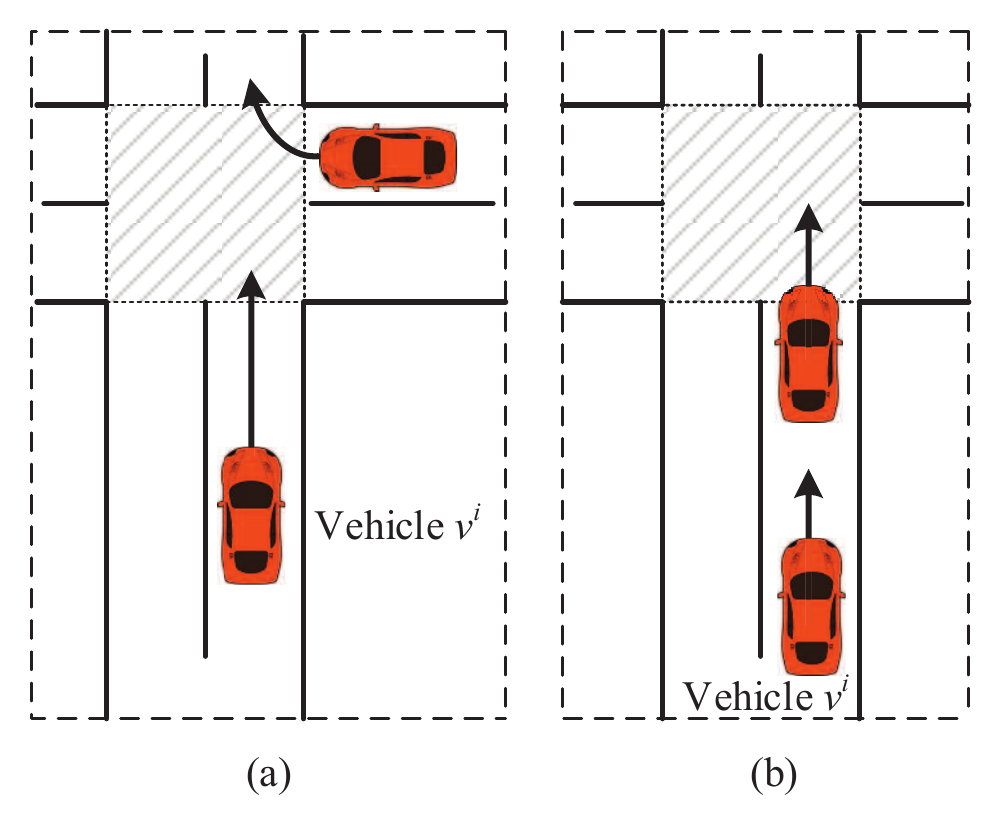}
	\caption[Example situations of front vehicles]{Example situations of front vehicles: (a) vehicles with different routes but same exit lane, and (b) vehicles with same intersection crossing routes.}
	\label{fig:Front}
\end{figure}

As shown in Figure \ref{fig:Front}, there are two types of front vehicles when a vehicle $v^i$ is approaching and crossing an intersection. 
In DICA, a vehicle is considered as a front vehicle of $v^i$ if the vehicle comes from another lane but has the same exit lane as vehicle $v^i$ or the vehicle is immediately in front of $v^i$ and has the exact same intersection crossing route as that of $v^i$. For a vehicle $v^i$, if there is another confirmed vehicle whose exit lane is the same as that of vehicle $v^i$ and will exit the intersection earlier, then they may collide immediately after crossing the intersection if the speed of vehicle $v^i$ is higher than that of the other confirmed vehicle. To address this problem, AIM \cite{dresner2008multiagent} adopted a simple strategy which gives one second separation time between these two vehicles. However, it is important to note that the separation time should depend on the speeds of the two vehicles. Hence, instead of using a fixed separation time approach, we use an approach that restricts the maximum speed of a following vehicle by the speed of the front vehicle. In the example situation (a) shown in Figure \ref{fig:Front}, the vehicle $v^i$'s maximum allowed speed within an intersection is restricted by the front vehicle's exit speed. If there is another confirmed vehicle that has the same intersection crossing route as vehicle $v^i$, we adjust $v^i$'s speed to leave adequate distance between them. In Algorithm \ref{code:dica}, the function \verb+checkFV()+ looks for the existence of above mentioned front vehicles from all confirmed vehicles and delay the new head vehicle to avoid potential collisions if needed.

%\textbf{Decide whether to include the three types of front vehicles in DICA-MIP}

\subsection{Vehicles for Collision Avoidance} \label{subsec:oti}
\begin{algorithm}[!htb]    
	\caption{getCV$(\mathcal{S}, DTOT(v^i))$}
	\begin{algorithmic}[1]
		
		%    	\caption{DTOT Intersection Control Algorithm}
		\STATE $\mathcal{C} = \emptyset$
		\FOR {$v^j$ in $\mathcal{S}$}
		%			\STATE 
		%			\STATE
		\FOR {$O^j_{k_j}$ in $DTOT(v^j)$} \label{Ov}
		%		\STATE
		\IF {$v^j$ not in $\mathcal{C}$} \label{CVi}
		%				\STATE
		\FOR {$O^i_{k_i}$ in $DTOT(v^i)$}
		%		\STATE
		\IF {$O^j_{k_j} \cap O^i_{k_i} \neq \emptyset$} \label{if}
		%		\STATE
		\STATE Call getOTI($O^j_{k_j})$ $\rightarrow I(O^j_{k_j}) := [\tau_{lb}({O}^j_{k_j}), \tau_{ub}({O}^j_{k_j})]$
		\STATE Call getOTI($O^i_{k_i}$) $\rightarrow I(O^i_{k_i}) := [\tau_{lb}({O}^i_{k_i}), \tau_{ub}({O}^i_{k_i})]$
		\IF {$I({O}^j_{k_j}) \cap I({O}^i_{k_i}) \neq \emptyset$}
		\STATE Assign $\tau_{lb}({O}^j_{k_j}) \rightarrow v^j.firstTimeAtCollision$
		\STATE Push $v^j$ into $\mathcal{C}$
		\STATE Sort $\mathcal{C}$ in ascending order of \emph{firstTimeAtCollision}
		\ENDIF 
		%		\STATE
		\ENDIF \label{endif}
		%		\STATE
		\ENDFOR
		%				\STATE
		\ENDIF
		%		\STATE
		\ENDFOR
		%			\STATE			
		%			\STATE
		\ENDFOR
		
	\end{algorithmic}
	\label{code:GetCV}
\end{algorithm}

The function \verb+getCV()+ returns the set $\mathcal{C}$ that contains vehicles which will cause potential collisions inside the intersection with vehicle $v^i$. To better understand the operation of function \verb+getCV()+, it is necessary to introduce the way we check the space-time conflict between two occupancies from DTOTs of two vehicles. 
For every individual occupancy in a DTOT of a vehicle, we define the entrance time ($\tau_{lb}$) and the exit time ($\tau_{ub}$) of the occupancy as the times when the vehicle first contacts and is totally out of the occupancy. These two times can be estimated by taking the times of the previous and next occupancies which are the closest to the occupancy while having no overlapping area. As an example, for the occupancy $O^j_4$ of the vehicle $v^j$ in Figure \ref{fig:Example}, the entrance time $\tau_{lb}(O^j_4)$ and the exit time $\tau_{ub}(O^j_4)$ of that occupancy can be determined by $\tau(O^j_2)$ and $\tau(O^j_6)$, respectively. Note that a DTOT for a vehicle consists of many more numbers of occupancies in practice. Hence, the entrance times and exit times determined in this way can be very close to the actual entrance and exit times of the occupancy. For the first several occupancies in a DTOT, there may not be a previous occupancy that has no overlapping area with themselves. For these occupancies, we simply take the first occupancy's time in the DTOT as these occupancies' entrance time. As an example shown in Figure \ref{fig:Example}, we use $\tau(O^j_1)$ as the entrance time $\tau_{lb}(O^j_2)$ for the occupancy $O^j_2$. Similarly, we take the last occupancy's time as the exit time $\tau_{ub}$ for the last several occupancies in a DTOT.

As shown in Algorithm \ref{code:GetCV}, the function \verb+getCV()+ determines the set $\mathcal{C}$ by checking space-time conflict for every pair of occupancies $(O^i_n, O^j_m)$ for all $n, m$, and $j$ in the set $\mathcal{S}$. Since an occupancy in a DTOT is represented as a rectangle, it is relatively straightforward to do a space conflict checking. For this, Algorithm \ref{code:GetCV} simply checks if two rectangles have a non-empty intersection or not. 
If a pair of occupancies $(O^i_n, O^j_m)$ are space-conflicting, then the function continues to investigate these occupancies to determine if they are in time-conflict as well. The above-explained entrance and exit times of an occupancy are used for this purpose. For a given occupancy $O$, the function \verb+getOTI()+ calculates these entrance $\tau_{lb}(O)$ and exit $\tau_{ub}(O)$ times for that occupancy and returns a corresponding time interval $I(O) := [\tau_{lb}(O), \tau_{ub}(O)]$ which we call the \emph{occupancy time interval} in the sequel. 
Then the two occupancy time intervals for the pair of space-conflicting occupancies are compared to determine if these occupancies are also occupied around the same time. 
If a pair of occupancies $(O^i_n, O^j_m)$ are conflicting in both space and time, then the vehicle $v^j$ is included in the set $\mathcal{C}$ and the corresponding \emph{firstTimeAtCollision} is determined so that the vehicle $v^j$ is appropriately ordered within the set $\mathcal{C}$.

\subsection{DTOT Update} \label{subsec:update}
The first vehicle $v$ in the set $\mathcal{C}$ is the earliest vehicle that is space-time conflicting with vehicle $v^i$. Then, in line \ref{update} of Algorithm \ref{code:dica}, the function \verb+updateDTOT()+ modifies vehicle $v^i$'s DTOT to avoid collision with vehicle $v$ based on space-time conflicting occupancies between vehicles $v^i$ and $v$. However, it is still uncertain whether $\mathcal{C}$ will be empty or not after this update of avoiding collision with vehicle $v$. 
In fact, it is still possible that the modified DTOT of vehicle $v^i$ will be in space-time conflict with DTOTs of other confirmed vehicles.
Hence, to ensure that vehicle $v^i$ avoids collision with all other confirmed vehicles, it is necessary to construct $\mathcal{C}$ based on the updated vehicle $v^i$'s DTOT and update the DTOT again to avoid collision with the first vehicle in the set.
This process is repeated in the \verb+while+ loop in Algorithm \ref{code:dica} until the set $\mathcal{C}$ becomes empty which means that vehicle $v^i$ is not conflicting with any confirmed vehicles.
When a vehicle proposes its DTOT to ICA, we assume that it prefers to select the fastest way to pass the intersection which means the vehicle will try to use the maximum allowed speed to cross. Our current strategy for updating a vehicle's DTOT is to delay the vehicle until other confirmed vehicles cross an intersection safely. While it is an interesting future research problem to develop more sophisticated approaches to improve the overall performance, the current simple delay strategy is still very effective to ensure collision-free intersection traffic. Note that, since the times of occupancies in a vehicle's DTOT are always delayed whenever the vehicle's DTOT is updated, it is guaranteed that the vehicle can always meet the updated DTOT by simply decelerating to experience a long time before entering the intersection. The worst case is that a vehicle may need to stop and wait for some time before an intersection to meet the given confirmed TSS from ICA.

\section{Liveness Analysis} \label{sec:anal}
A \emph{deadlock} is a situation where two or more processes are unable to proceed and each process is waiting for another one to finish because they are competing for shared resources. In an intersection crossing traffic, a deadlock could happen when several vehicles are trying to cross the intersection at the same time. For example, if the coordination between vehicles who want to cross an intersection is not done appropriately, then a deadlock may occur between two vehicles on a same lane. As discussed in \cite{dresner2008multiagent}, it is possible that even when the vehicle in front cannot get confirmed due to the conflict of its intersection crossing route with those of other vehicles which are already confirmed to enter and cross an intersection, the vehicle in the back may get confirmed because its intersection crossing route is not conflicting with other confirmed vehicles' crossing routes. And the vehicle successfully reserves the space for its intersection crossing route within an intersection. In this situation, the front vehicle cannot get confirmed since some part of the intersection crossing route of it conflicts with that of the behind vehicle which is already confirmed and also the behind vehicle cannot proceed to cross the intersection due to the unconfirmed front vehicle.
A deadlock situation may also occur when several vehicles from different directions want to cross an intersection at the same time. This type of deadlock situation is discussed in detail in \cite{azimi2013reliable} for the case of four vehicles in which none of the vehicles can progress inside the intersection because each of the vehicles' next occupancies are already occupied by other vehicles. Now we show that DICA shown in Algorithm \ref{code:dica} are free from these deadlock situations.

\begin{prop} \label{prop:deadlock}
DICA is deadlock free.
\end{prop}
\begin{proof}
Let $\mathcal{S}_k$ denote the set of confirmed vehicles at the $k$-th time step of DICA. Then, we show that the set $\mathcal{S}_k$ is deadlock free for all $k = 0, 1, 2, \cdots$ by induction. First, at time step $k = 0$, it is easy to see that there is no deadlock in $\mathcal{S}_0$ since no vehicle is confirmed yet, i.e., $\vert \mathcal{S}_0 \vert = 0$ where $\vert \cdot \vert$ denotes the cardinality of a set. Then, at time step $k > 0$, let us suppose that $\mathcal{S}_k$ is deadlock free and a new head vehicle $v^i$ is under consideration for confirmation. Note that, as discussed in Section \ref{sec:dtot}, a vehicle is considered by DICA for confirmation only if it is the head vehicle on its lane. Hence, it is trivial to see that there won't be a deadlock situation between the vehicle $v^i$ and other vehicle $v^{i'}$ which is behind $v^i$ since $v^{i'} \not\in \mathcal{S}_k$.
Next, let us note that once a vehicle $v^j$ is in $\mathcal{S}_k$, then the vehicle's DTOT will not be changed while and after a new vehicle $v^i$ is processed to be confirmed by DICA. 
Hence, it is easy to see that any vehicle which is in $\mathcal{S}_k$ at time step $k$ remains deadlock free at the next time step $(k+1)$. Now suppose that the new vehicle $v^i$ has been confirmed by DICA at time step $k$ and included in the set of confirmed vehicle at time step $(k+1)$, i.e., $v^i \in \mathcal{S}_{k+1}= \mathcal{S}_k \cup \{ v^i \}$. Since all vehicles in $\mathcal{S}_k \subset \mathcal{S}_{k+1}$ are deadlock free, if the new vehicle $v^i$ is deadlock free, then we know that $\mathcal{S}_{k+1}$ is deadlock free and this proves the deadlock free property of DICA. In fact, it is straightforward to see that $v^i$ is also deadlock free after its DTOT is updated and confirmed by DICA. First, note that modification of the vehicle $v^i$'s DTOT is not affected by any vehicle $v \not\in \mathcal{S}_k$. Instead, it is affected only by vehicles which are already in the $\mathcal{S}_k$. Since all vehicles in $\mathcal{S}_k$ are deadlock free and eventually proceed to cross and exit the intersection, the vehicle $v^i$'s DTOT is also updated so that the vehicle $v^i$ will eventually enter and cross the intersection while all vehicles in $\mathcal{S}_k$ cross the intersection safely. Thus, the vehicle $v^i$ is also deadlock free at time step $(k+1)$ and this concludes the proof of this proposition.
\end{proof}

In an intersection crossing traffic, a \emph{starvation} situation may occur when vehicles from a certain direction are waiting for a very long time or even indefinitely to be allowed to enter and cross an intersection while vehicles from other directions are continuously allowed to cross the intersection. Now we show that a starvation situation will not occur in an intersection crossing traffic that is coordinated by DICA.

\begin{prop} \label{prop:starvation}
DICA is starvation free.
\end{prop}
\begin{proof}
First, let us recall that, as discussed in Section \ref{sec:dtot}, DICA considers a vehicle for confirmation only when the vehicle becomes the head vehicle on its lane. Now let $\sigma(v)$ be the vehicle $v$'s entrance time to the communication region of an intersection, $\mathcal{H}$ be the set of head vehicles which is ordered by $\sigma(v)$ for all $v \in \mathcal{H}$, and $\mathcal{H}^-$ be the set of vehicles which are approaching to cross an intersection but not included in the set $\mathcal{H}$. Clearly, $\vert \mathcal{H} \vert$ is bounded by the number of all lanes from which vehicles are approaching an intersection to cross and $\vert \mathcal{H}^- \vert$ is also bounded by both the number of lanes and the length of lanes within the communication region of an intersection.
Note that DICA processes vehicles in $\mathcal{H}$ for confirmation according to the order of vehicles in $\mathcal{H}$. Once the first vehicle in $\mathcal{H}$ is processed and gets confirmed, then the vehicle is removed from $\mathcal{H}$. Note that if DICA is not starvation free, then there must exist at least one vehicle $v \in \mathcal{H}$ such that the vehicle $v$ will never (or at least take an unnecessarily very long time to) become the first element in the ordered set $\mathcal{H}$. Thus, to prove the starvation free property of DICA, it suffices to show that, for any vehicle $v \in \mathcal{H}$, the vehicle $v$ will be removed from $\mathcal{H}$ in finite time.
To show this, we can consider the last vehicle $v_{last}$ in the ordered set $\mathcal{H}$. If $\sigma(v_{last}) \le \sigma(v)$ for all $v \in \mathcal{H}^-$, then the vehicle $v_{last}$ will be cleared right after all other vehicles in $\mathcal{H}$ are confirmed and this is the earliest time for $v_{last}$ to be removed from $\mathcal{H}$. On the other hand, if $\sigma(v_{last}) > \sigma(v)$ for all $v \in \mathcal{H}^-$ as the worst situation for $v_{last}$, then the vehicle $v_{last}$ might need to wait until all ($\vert \mathcal{H} \vert + \vert \mathcal{H}^- \vert$) vehicles get confirmed to be considered for confirmation. Thus, it is clear that the vehicle $v_{last}$ will be cleared from $\mathcal{H}$ in finite time. 
\end{proof}

\section{Simulation}
In this section, we present simulation results of DICA and the Concurrent Algorithm \cite{kim2014mpc} under same configurations. Compared with Concurrent Algorithm, DICA provides a more efficient way of coordinating vehicles to avoid unnecessary delay. 

\subsection{Simulation Setup}
Traffic simulation is performed by the microscopic road traffic simulation package SUMO (Simulation of Urban MObility) \cite{krajzewicz2012recent}. This simulator is widely used in the research community, which makes it easy to compare performance of different algorithms. Our intelligent intersection management algorithm is implemented by the Traffic Control Interface (TraCI) in SUMO.

The simulated scenario in our simulation is the traffic of a typical isolated four way intersection with two incoming lanes and one outgoing lane on each road. $v_m=70\ km/h$ is set as the maximum allowed speed for incoming roads. We generate vehicles with random velocity within the range of $40\%*v_m$ and $v_m$ when they enter the communication region. To simulate intersection traffic as real as possible, we use different maximum allowed speeds for vehicles with different routes. Although there are not specific speed limits for vehicles who are turning left or right, people are still using some lower speed to feel comfortable and maintain safety. We choose conservative speed limits for turning based on experience from driving in daily life. We use $25\ km/h$ for right turning and $35\ km/h$ for left turning. For vehicles with through route, $65\ km/h$ is set as the speed limit.
The time step we used in simulation is $0.05\ s$. The maximum acceleration and deceleration rates are $2\ m/s^2$ and $4.5\ m/s^2$. Vehicles have a size of $5$ meters length and $1.8$ meters width. Since, in some cases, a vehicle may need to stop just before the entrance line of the intersection region to avoid collisions with other vehicles, the distance from the entrance line of the communication region to the entrance line of the intersection region should be long enough so that a vehicle can stop from its maximum speed $v_m$. Thus, from the value used for $v_m = 70\ km/h$ the maximum deceleration rate $a_{min} = 4.5\ m/s^2$, we need at least $v_m^2/(2 a_{min}) \approx 42.03\ m$. So, we use $50\ m$ for the distance from the entrance line of the communication region to the entrance line of the intersection region.
We evaluate the performance of our algorithm in situations where vehicles are spawned randomly from each direction at different probabilities. In our simulation, we consider three different traffic volumes. For each traffic volume, through different traffic generation time and randomly generated vehicles' routes, we test each case for twelve different traffic patterns. Average data of twelve different traffic patterns are used as the result for that case. Each simulation run is terminated when a certain time limit ($10$ min) has been reached.
Fig \ref{fig:screenshot} shows a screenshot of simulation in SUMO when vehicles of different routes appears within the intersection simultaneously without occurrence of collision. Inside the intersection, the straight going vehicle from East goes inside the intersection shortly after the vehicle from North to South clears the conflicting space. Vehicles whose DTOTs is not conflicting with these two can pass the intersection at the same time, for example the right-turning vehicle from South in the figure.

\begin{figure}[!t]
    \centering
    \includegraphics[width=3.3in]{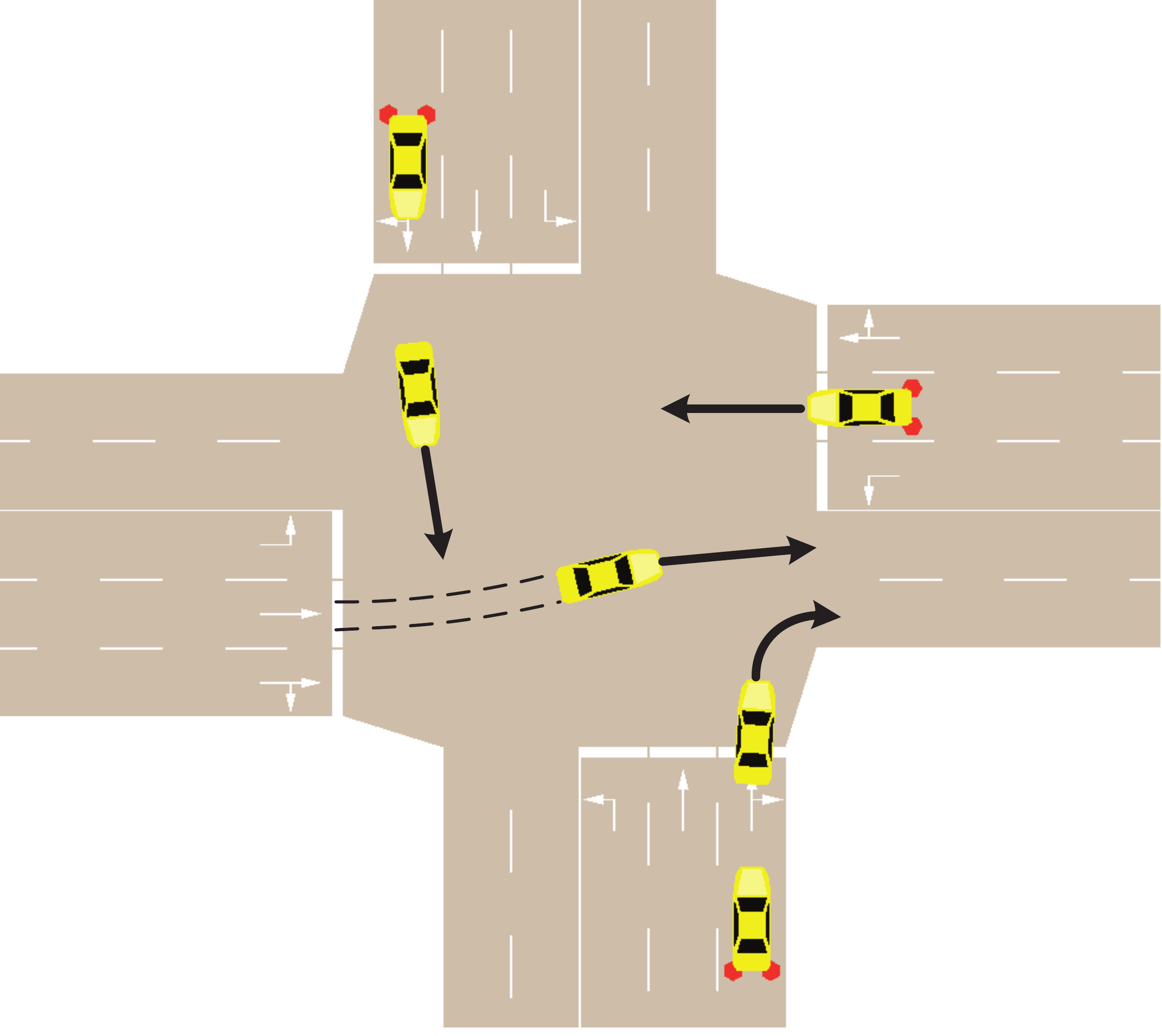}
    \caption{A screenshot of simulation which illustrates a situation when vehicles with conflicting routes cross the intersection simultaneously.}
    \label{fig:screenshot}
\end{figure}

\begin{table}[htbp] %\small%\footnotesize
    \centering
    \caption{Simulation results comparison.}
    \begin{tabular}{ccccccc}
        \toprule
        %\hline
        \multicolumn{1}{c}{\multirow{2}[2]{*}{}} & \multicolumn{1}{c}{\multirow{2}[2]{*}{Volume}}  & \multirow{2}[2]{*}{$\rho$} & \multicolumn{3}{c}{Crossed Vehicles} & \multirow{2}[2]{*}{$\bar{\tau}_e$ } \\

        \cmidrule{4-6}\multicolumn{1}{c}{} &\multicolumn{1}{c}{} & & $\bar{\tau}$ & $\sigma_\tau$ & $\eta$ \\

        \midrule
        %\hline
        \multirow{3}[1]{*}{Concurrent} & Volume 1 & 93.99\% & 13.98 & 8.58 & 51.74\% & 14.85 \\
						        & Volume 2 & 60.24\% & 89.08 & 40.91 & 95.34\% & 150.77 \\
							    & Volume 3 & 30.91\% & 125.09 & 48.33 & 97.17\% & 408.90 \\ 
		\midrule
        \multirow{3}[1]{*}{DICA} & Volume 1 & 95.11\% & 7.21  & 2.97 & 10.13 \% & 7.57 \\
 						         & Volume 2 & 91.09\% & 20.47  & 15.54 & 55.86\% & 22.53 \\
					             & Volume 3 & 57.09\% & 50.60 & 32.87 & 84.37 \% & 88.92 \\
        \bottomrule
        %\hline
    \end{tabular}%
    \label{table:dica}%
\end{table}%

\subsection{Simulation Results}
Performance improvement has been validated through extensive simulations of DICA and Concurrent Algorithm. Simulation results of different volumes are shown in Table \ref{table:dica}. To evaluate and compare the performance, we define several performance measures. \emph{Trip time ($\tau$)} is the difference between actual exit time of the intersection and the time the vehicle enters the intersection communication range. \emph{Average trip time ($\bar{\tau}$)} is the average value of trip times of all crossed vehicles. \emph{Standard deviation ($\sigma_\tau$)} is computed based on the trip times of all crossed vehicles. \emph{Throughput ($\rho$)} is defined as the percentage of crossed vehicles against total generated vehicles. \emph{Stopped rate ($\eta$)} is obtained through dividing stopped vehicles number by crossed vehicles number. 

However, note that neither the average trip time nor the throughput alone is sufficient to correctly evaluate the performance of an algorithm. In some cases, it could be possible that one algorithm shows better performance on average trip time while another algorithm performs better on throughput. Thus, both of the two measures should be considered together to correctly compare and evaluate the performances of different intersection control algorithms. We calculated the ratio of average trip time to throughput, which is called \emph{effective average trip time ($\bar{\tau}_e$)} and believe it could show the performance of an algorithm more comprehensively, i.e. $\bar{\tau}_e = \bar{\tau} / \rho$. 

Compared with Concurrent Algorithm, DICA increases the throughput and largely decreases the average trip time. As shown in Table \ref{table:dica}, both algorithms have less and less throughput with the increase of traffic volume. Also, larger decrease of throughput is shown in Concurrent Algorithm compared with DICA. DICA's smaller values of standard deviation imply that DICA is fairer than Concurrent Algorithm. Stopped rate shows that less vehicles experience a stop at the intersection enter line for our DICA algorithm thus saves energy. Effective average trip time could tell us comprehensive information about the performance of an intersection control algorithm. Efficiency and fairness of an algorithm are integrated in this measure. With more vehicles crossed and less average trip time, DICA has much less value of effective average trip time than that of Concurrent Algorithm.

\begin{figure}[!t]
    \centering
    \includegraphics[width=4in]{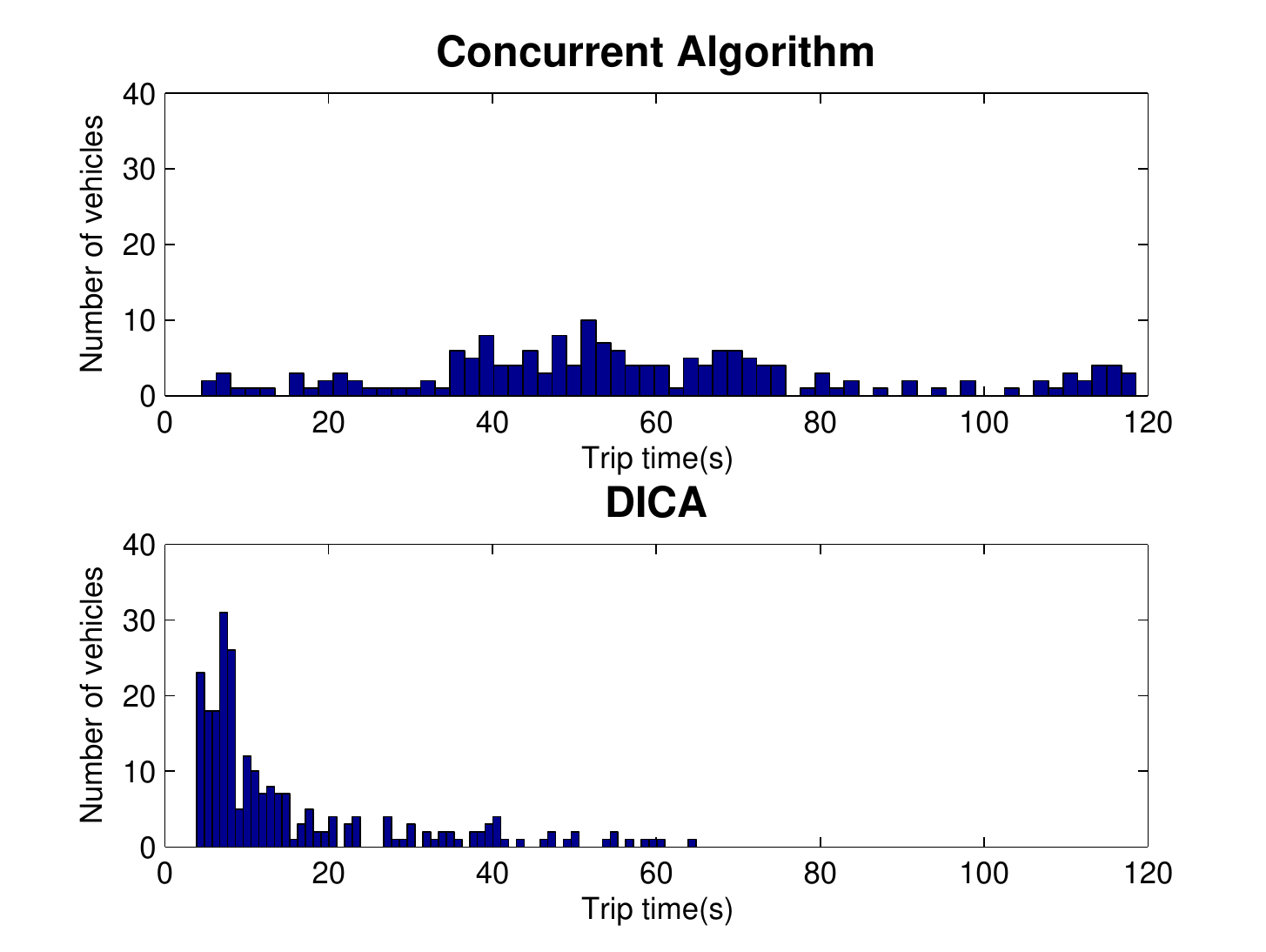}
    \caption{The histogram of Trip Times of crossed vehicles in a simulation.}
    \label{hist}
\end{figure}

The histogram of trip times of crossed vehicles in one simulation run from case 1 and one particular traffic pattern is shown in Figure \ref{hist}. From the figure we can see that under the same traffic setting, for the crossed vehicles, DICA results in less and concentrated trip times. On the other side, Concurrent Algorithm leads to much longer and wider distributed trip times. Compared with Concurrent Algorithm, DICA is fairer and much efficient for crossed vehicles.

\begin{figure}[!ht]
    \centering
    \ %includegraphics[width=3.5in,bb=0 0 403 302]{comppng}
    \includegraphics[width=4in]{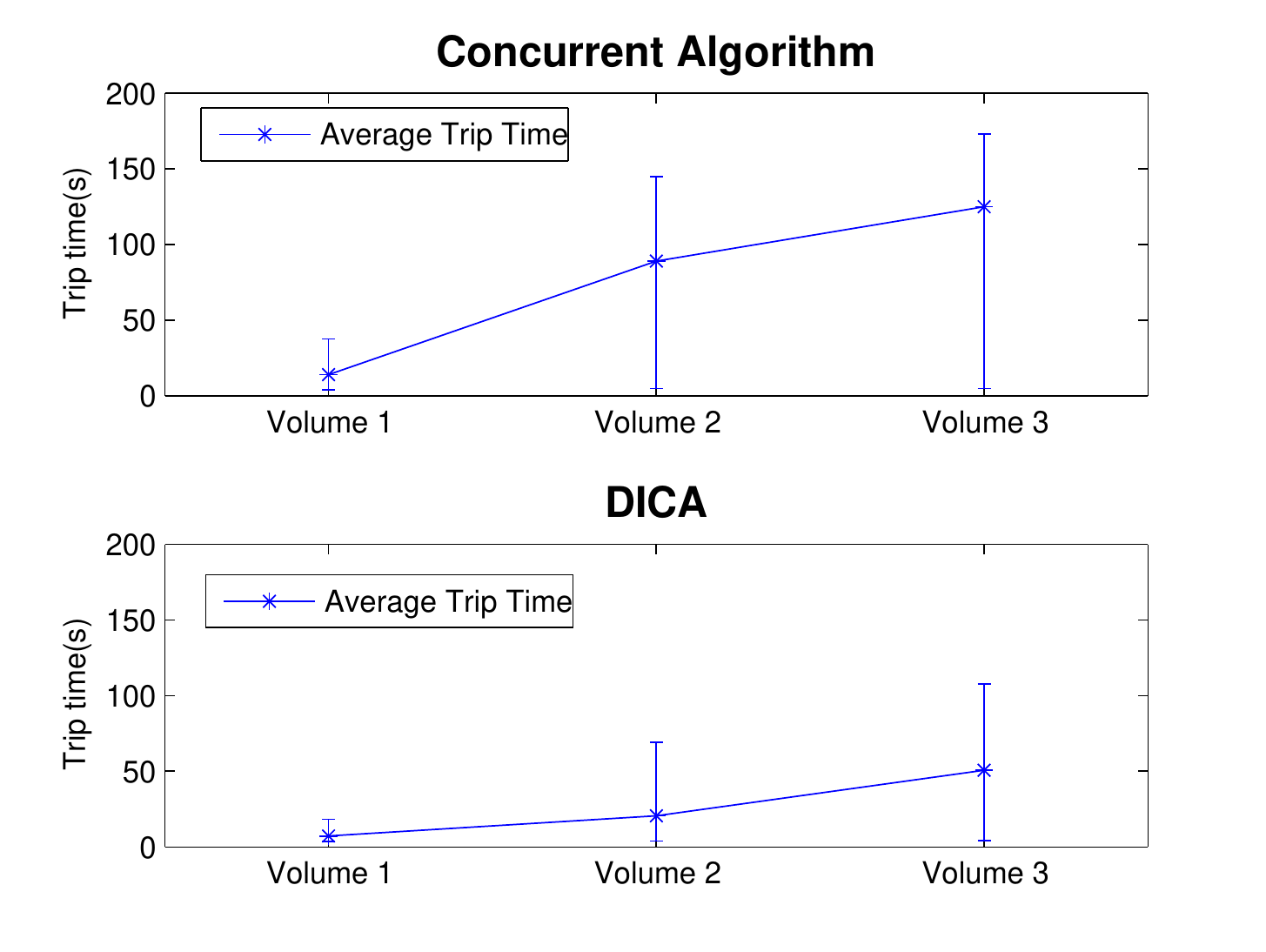}
    \caption{Comparison of Trip Times of 3 different cases between DICA and Concurrent Algorithm.}
    \label{comp}
\end{figure}
Figure \ref{comp} shows the maximum, average and minimum trip time for both algorithms under three different volumes. With heavier traffic volume, both algorithms have large maximum trip times and average trip times. However, DICA always performs better than Concurrent Algorithm in all three volumes.

\section{Summary}

In this chapter, we have developed an intelligent intersection control algorithm DICA employing the concept DTOT. V2I interaction protocol has been established for interactions between vehicles and intersection. DICA is able to manage limited intersection space at a more accurate and efficient way. Simulation results show that our algorithm achieves less effective average trip time compared with Concurrent Algorithm proposed by \cite{kim2014mpc}. %input the file 

\chapter{Computational Complexity Improvements of DICA} \label{chapter:e-dica}% Enhanced DICA
In this chapter, we analyze the overall computational complexity of DICA and improve it in several computational technical approaches. We also enhance the algorithm accordingly so that it is possible to operate the algorithm in real-time for autonomous and connected intersection traffic management.

\section{Computational Complexity Analysis} \label{sec:anal}

In this section, we analyze the computational complexity of DICA shown in Algorithm \ref{code:dica}. 
Recall that $\mathcal{S}$ is the set of vehicles within the communication region of an intersection that has been confirmed to cross.
Let us assume that there are $n$ vehicles in $\mathcal{S}$, i.e., $\vert \mathcal{S} \vert = n$. 
Then we have the following result on the computational complexity analysis of DICA.

\begin{prop} \label{prop:dica}
DICA has $\mathcal{O}(n^2 L_m^3)$ computational complexity where $L_m$ is the maximum length of intersection crossing routes in an intersection.
\end{prop}
\begin{proof}
Let $v^i$ be the vehicle which is currently being processed by ICA for intersection crossing confirmation. 
Also let $N_m := \max_{k \in \mathcal{S}'} N^k$ where $\mathcal{S}' = \mathcal{S} \cup \{ v_i \}$ and $N^k$ is the number of occupancies in the vehicle $k$'s DTOT.
Then, in line 3 (Algorithm \ref{code:dica}), it is easy to see that creating DTOT from the TSS and vehicle size information in the vehicle $v^i$'s REQUEST message involves only $\mathcal{O}(N_m)$ computational complexity.
In line 4 (Algorithm \ref{code:dica}), as explained in Section \ref{sec:dtot}, the front vehicle checking function \verb+checkFV()+ does a simple comparison with every confirmed vehicle in $\mathcal{S}$ to see if there are any vehicles which might affect the vehicle $v^i$'s DTOT and modifies the DTOT if it is necessary to ensure enough separation time and distance between the vehicle $v^i$ and other vehicles in front. And this process requires $\mathcal{O}(nN_m)$ computational complexity. 
Then, in line 5 (Algorithm \ref{code:dica}), the function \verb+getCV()+ is called to identify the set $\mathcal{C}$ of vehicles in $\mathcal{S}$ whose DTOTs might be in space-time conflict with the vehicle $v^i$'s DTOT. (Note that, as shown in Algorithm \ref{code:GetCV}, $\mathcal{C}$ is an ordered set according to time of collision and it is clearly $\mathcal{C} \subseteq \mathcal{S}$.) Thus, to return the set $\mathcal{C}$ from the set $\mathcal{S}$, this function performs $n$ times of space-time conflict checking between the vehicle $v^i$ and vehicles in $\mathcal{S}$. If a non-empty set $\mathcal{C}$ is returned in line 5 (Algorithm \ref{code:dica}), then, in lines 6 $\sim$ 10 (Algorithm \ref{code:dica}), the vehicle $v^i$'s DTOT is iteratively updated until the set $\mathcal{C}$ becomes empty within the \verb+while+ loop. (As one can see in Algorithms \ref{code:dica} and \ref{code:GetCV}, these steps are indeed the main part of the DICA algorithm and involve some computationally expensive operations. Hence, we describe the computational complexity of steps within the \verb+while+ loop separately in the next paragraph.)  
After the \verb+while+ loop, as the last steps in Algorithm \ref{code:dica} in lines from \ref{store} to \ref{send}, the space-time conflict free DTOT for the vehicle $v^i$ is stored, converted into TSS, and then sent to $v^i$ so that the vehicle can cross the intersection according to the DTOT. Clearly, these steps are fairly simple in terms of computation and in fact require $\mathcal{O}(1)$ complexity.
Next, we analyze the computational complexity of steps within the \verb+while+ loop. 

\emph{Space-time conflict checking steps}: 
As described in Section \ref{sec:dtot}, space-time conflict checking in \verb+getCV()+ is done using DTOTs of vehicles. Specifically, the two nested \verb+if+ blocks from line 6 to line 14 in Algorithm \ref{code:GetCV} perform this operation. For space conflict checking, it is checked if there exist non-empty intersections between two occupancies: one from DTOT of the vehicle $v^i$ and another from DTOT of one of the vehicles in the set $\mathcal{S}$. This is done in the outer \verb+if+ block and requires $n \cdot N_m^2$ times of iteration in the worst case. If two vehicles have a space conflict, then Algorithm \ref{code:GetCV} proceeds to check for time conflict. To check time overlapping between two space conflicting occupancies, the function needs to calculate time intervals for these occupancies during which each vehicle occupies its occupancy. This can be done easily by comparing occupancy time between occupancies within the same DTOT. As an example, for a given occupancy $O^i_k$ which is the $k$-th occupancy within the vehicle $v^i$'s DTOT, the lower and upper bounds for the occupancy time can be determined by space overlapping checking between the occupancies $O^i_k$ and $O^i_{k'}$ for $k' = \{ 1, \cdots, N_m \} \setminus k$. Thus, the two function calls to \verb+getOTI()+ within the \verb+if+ block involve the computational complexity of $\mathcal{O}(N_m)$. Once the occupancy time intervals are determined, it is a straightforward calculation to check time overlapping as shown in line 9 of Algorithm \ref{code:GetCV} and it takes $\mathcal{O}(1)$ computational complexity. After identifying all space-time conflicting vehicles from the set $\mathcal{S}$ and storing them to the set $\mathcal{C}$, Algorithm \ref{code:GetCV} then sorts the set $\mathcal{C}$ according to the ascending order of occupancy times of space-time conflicting occupancies and returns the set. Note that $\vert \mathcal{C} \vert \le n$ and $n \ll N_m$ in general. Hence, this sorting operation can be done with $\mathcal{O}(n log_2 N_m)$ computational complexity. If we consider all these calculation steps in the \verb+getCV()+ function, then one can see that the overall computational complexity for space-time conflict checking steps in \verb+getCV()+ is $\mathcal{O}(n N_m^3)$.

\emph{DTOT adjustment for collision avoidance}: 
Once the set $\mathcal{C}$ is returned from the function \verb+getCV()+, the DICA algorithm updates the vehicle $v^i$'s DTOT to avoid space-time conflict with DTOTs of the vehicles in the set $\mathcal{C}$. In line 7 (Algorithm \ref{code:dica}), it is shown that the first vehicle $v^j$ in the set $\mathcal{C}$ is considered for updating the vehicle $v^i$'s DTOT. As described in Section \ref{sec:dtot}, our update strategy to avoid space-time conflicts is to make the vehicle $v^i$ enter the intersection area a little bit late so as to give enough time for vehicle $v^j$ to cross the intersection safely. For this, the DICA algorithm first needs to compute the delay time needed to avoid the space-time conflict with the vehicle $v^j$. Since the occupancy time interval $I(O^j_k)$ for the vehicle $v^j$'s earliest space-time conflicting occupancy has already been determined from the function \verb+getCV()+, it is easy to calculate this delay time in this update process. Once the delay time is determined, then the remaining step is simply to change the occupancy times of all the occupancies in the vehicle $v^i$'s DTOT to be delayed and this results in $\mathcal{O}(N_m)$ computational complexity.

As described above, the number of vehicles in the set $\mathcal{S}$ is $n$ when the function \verb+getCV()+ is called for the first time in line 5 (Algorithm \ref{code:dica}). Then, within the \verb+while+ loop, the function \verb+updateDTOT()+ adjusts the vehicle $v^i$'s DTOT to avoid collision with the first vehicle in the set $\mathcal{C}$ and this step reduces the number of vehicles in the set $\mathcal{C}$ that can potentially collide with the vehicle $v^i$ at least by one. Thus, in the worst case, the number of vehicles in the set $\mathcal{C}$ returned by the second call of $\verb+getCV()+$ within the \verb+while+ loop is $(n-1)$. If we assume the worst case for all following iterations within the \verb+while+ loop until the set $\mathcal{C}$ becomes empty, then it is easy to see that the functions \verb+getCV()+ and \verb+updateDTOT()+ are called $n$ times within the \verb+while+ loop. This implies that since the computational complexity of the function \verb+updateDTOT()+ is significantly lower than that of the function \verb+getCV()+, the overall computational complexity of the \verb+while+ loop can be considered as $\mathcal{O}(n^2 N_m^3)$. 

Note that the maximum number of occupancies $N_m$ depends on both the time that it takes for a vehicle to cross the intersection and the discrete time step used to construct the DTOT by ICA. If we let $h$ be the discrete time step used by ICA and $T_m$ be the time it takes for a vehicle to completely cross an intersection when the vehicle starts from rest and accelerates to cross the intersection as quickly as possible, then we have $\bar{N}_m := T_m/h$ as an upper bound for $N_m$. Note that $T_m$ depends on the length of an intersection crossing route that a vehicle takes to cross an intersection. If we let $L_m$ be the maximum length out of all intersection crossing routes for an intersection, then $\bar{N}_m$ can be expressed in terms of $L_m$ instead of $T_m$. Specifically, if $L_m$ is long enough so that a vehicle can reach its maximum allowed speed $v_m$ within an intersection before it completely crosses the intersection, then it can be shown that $\bar{N}_m = {(2a_m L_m + v_m^2) / (2a_m v_m h)}$ where $a_m$ is the maximum acceleration rate of a vehicle. On the other hand, if $L_m$ is not long enough for a vehicle to reach $v_m$ while crossing an intersection, then it is also relatively straightforward to show that $\bar{N}_m = (\sqrt{2 L_m / a_m})/h$. 
(These two different cases are illustrated in Figure \ref{fig:complexity}.)  
If we fix values for $h, v_m$, and $a_m$, then one can see that $\bar{N}_m$ for the former case is proportional to $L_m$ while, for the latter case, $\bar{N}_m$ is proportional to the square root of $L_m$. Hence, if we substitute $L_m$ for $N_m$ in the computational complexity $\mathcal{O}(n^2 N_m^3)$ that we derived above, then we finally have $\mathcal{O}(n^2 L_m^3)$ as the overall computational complexity of DICA.

\begin{figure}[!bt]
\centering{\includegraphics[width=2.8in]{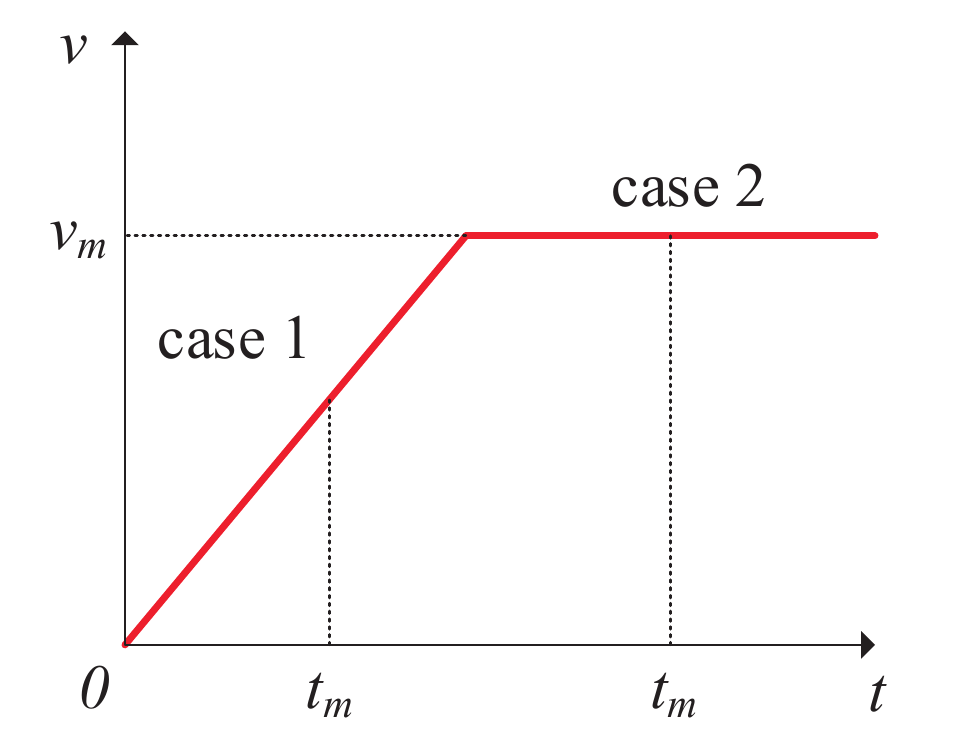}}
\caption[Two different cases for shortest intersection crossing time ($T_m$) calculation]{Two different cases for shortest intersection crossing time ($T_m$) calculation. (Case 1 is the situation when $L_m$ is too short to reach $v_m$ and case 2 is the situation when $L_m$ is long enough to reach $v_m$ while a vehicle is crossing an intersection.) \label{fig:complexity}}
\end{figure}
\end{proof}

\section{Algorithm Improvements} \label{sec:improv}
According to the computational complexity analysis result described in the previous section, it is true that the original DICA algorithm that is shown in Algorithms \ref{code:dica} and \ref{code:GetCV} is somewhat conservative in terms of computational cost to be used in practice. 
In this section, we present several approaches that can be used to improve the overall computational efficiency of the algorithm.

\subsection{Reduced Number of Vehicles for Space-Time Conflict Check} \label{sec:imp1}
As shown in Algorithm \ref{code:GetCV}, all confirmed vehicles in the set $\mathcal{S}$ are examined to obtain the set of space-time conflicting vehicles $\mathcal{C}$ for a new unconfirmed head vehicle $v^i$. However, we see that this computation process can be improved by excluding vehicles that cannot be in space-time conflict with the vehicle $v^i$ under any circumstances from the set $\mathcal{S}$. For example, a confirmed vehicle $v^j \in \mathcal{S}$ who has an intersection crossing time interval that is not overlapping with the vehicle $v^i$'s intersection crossing time interval can be excluded. Note that the intersection crossing time interval of a confirmed vehicle can be easily determined by the lower bound of the occupancy time $\tau_{lb}(O_{first})$ of the vehicle's first occupancy $O_{first}$ and the upper bound of the occupancy time $\tau_{ub}(O_{last})$ of the vehicle's last occupancy $O_{last}$ in the vehicle's confirmed DTOT. In addition to these vehicles, vehicles in the set $\mathcal{S}$ whose intersection crossing routes are compatible with that of vehicle $v^i$ can also be excluded. 
Hence, if we let $\mathcal{S}^*$ be the subset of all confirmed vehicles in set $\mathcal{S}$ that can be obtained after excluding all above-mentioned vehicles in determining the set $\mathcal{C}$, then the resulting computational complexity for the space-time conflict checking in function \verb+getCV()+ becomes $\mathcal{O}(\alpha_1 n N_m^3)$ where $\alpha_1 := \tilde{n}/n$, $\tilde{n} = \vert \mathcal{S}^*\vert$, $n = \vert \mathcal{S} \vert$, and $N_m$ is the maximum number of occupancies of all vehicles that are in the set $\mathcal{S}$ and also the vehicle that is currently under consideration for confirmation. (See the proof of Proposition \ref{prop:dica} for the precise definition of $N_m$.)

\subsection{Efficient Space Conflict Check} \label{sec:imp2}
Note that, for any two vehicles coming from different directions, they can collide with each other only within some parts of their intersection crossing routes. Thus, not all occupancies of a vehicle's DTOT needs to be checked for space conflict with another vehicle's DTOT. For example, the two vehicles $v^i$ and $v^j$ in Figure \ref{fig:Example} have very short ranges of intersection crossing routes that are space conflicting with each other. Thus, the occupancies to be checked can be reduced to {\{$O^i_2, O^i_3$\} and \{$O^j_5, O^j_6$\} from their entire DTOTs. 
Since the number of occupancies in a DTOT is very large in general, this can improve computational speed considerably. 
Note that, since the intersection crossing routes are fixed for a specific intersection, we can predetermine these space conflicting short ranges offline only one time for all pairs of incompatible intersection crossing routes. Hence, this extra preparation process does not incur an additional computational cost during the online operation of DICA.
If we use DTOT$^*$ to denote the subset of the original DTOT for a vehicle that can be obtained from this approach, then the computational complexity of the function \verb+getCV()+ in Algorithm \ref{code:GetCV} can be expressed as $\mathcal{O}(\alpha_2^3 n N_m^3)$ where $\alpha_2 := \tilde{N}_m / N_m$ and $\tilde{N}_m$ is the maximum number of occupancies of all vehicles that are in the set $\mathcal{S}^*$ and the vehicle that is currently under consideration for confirmation.

\subsection{Approximate Occupancy Time Interval Calculation} \label{sec:imp3}
 As explained in Section 3, ICA checks if an occupancy of a vehicle is conflicting in time with another vehicle's occupancy using occupancy time intervals that can be obtained from each vehicle's DTOT. However, the way to obtain an occupancy time interval presented in the proof of Proposition \ref{prop:dica} is somewhat naive in the sense of computational complexity. In fact, as analyzed in the proof, such an exhaustive search involves a computational complexity of $\mathcal{O}(N_m)$. To simplify this computation process, we propose to estimate the occupancy time interval for a certain occupancy based on the vehicle's speed, length, and acceleration rate instead of performing the exhaustive search. To clarify this idea, let us consider an example. For simplicity of explanation, we consider a case when a vehicle is moving in a straight line as shown in Figure \ref{fig:oti}. Let $O^i_k$ be the occupancy for which the DICA algorithm needs to determine the occupancy time interval $I(O^i_k) = [\tau_{lb}(O^i_k), \tau_{ub}(O^i_k)]$, $L(v^i)$ be the vehicle length of the vehicle $v^i$, $h$ be the sampling time interval, $x_k$ be the center position of the $O^i_k$ along the straight line. Then the algorithm first estimates the vehicle's speed and acceleration rate around the occupancy $O^i_k$ from $x_k$, $x_{k-1}$, $x_{k+1}$, and $h$. Occupancies at $x_{k-1}, x_{k+1}$ are very close to the occupancy $O^i_k$ and are not shown in Figure \ref{fig:oti} for simplicity. Specifically, if we let $V_{k^-}(v^i)$ and $V_{k^+}(v^i)$ be the speed of the vehicle $v^i$ from $O^i_{k-1}$ to $O^i_k$ and from $O^i_{k}$ to $O^i_{k+1}$ respectively, then these speeds can be approximated as follows:
\begin{equation} 
	V_{k^-}(v^i) \approx \frac{x_k - x_{k-1}}{h} , \quad V_{k^+}(v^i) \approx \frac{x_{k+1} - x_{k}}{h}
\end{equation}
From these speeds, we now approximate the acceleration rate of the vehicle as follows:
\begin{equation} 
	A_k(v^i) \approx \frac{V_{k^+}(v^i) - V_{k^-}(v^i)}{h}
\end{equation}
where $A_k(v^i)$ denotes the acceleration of the vehicle $v^i$ at the occupancy $O^i_k$. If we take the average of the speeds around $O^i_k$, then we can also approximate $V_k(v^i)$ which is the speed of the vehicle $v^i$ at $O^i_k$. Note that since the length of a vehicle $L(v^i)$ is just a few meters in general, the actual motion of the vehicle $v^i$ within the occupancy $O^i_k$ can be approximated fairly accurately by $V_k(v^i)$ and $A_k(v^i)$. 

Now, since it is a straightforward process to estimate $\tau_{lb}(O^i_k)$ and $\tau_{ub}(O^i_k)$ from $L(v^i)$,  $V_k(v^i)$, and $A_k(v^i)$, we omit the details of these calculations. For the case when the vehicle is moving on a curved path, we can still use the same method to approximate $V_k(v^i)$ and $A_k(v^i)$. But, in this case, we may need to add a short extra distance to the $L(v^i)$ to estimate $\tau_{lb}(O^i_k)$ and $\tau_{ub}(O^i_k)$ more accurately. Such an extra distance can be simply determined by the curvature of the path that is represented by the DTOT of a vehicle. 
Finally, if we apply this approximation method for an occupancy time interval calculation in the \verb+getOTI()+ function, then the computational complexity of the function \verb+getCV()+ improves from $\mathcal{O} (n^2 N_m^3)$ to $\mathcal{O}(n^2 N_m^2)$.  

\begin{figure}[!t]
\centering{\includegraphics[width=3.in]{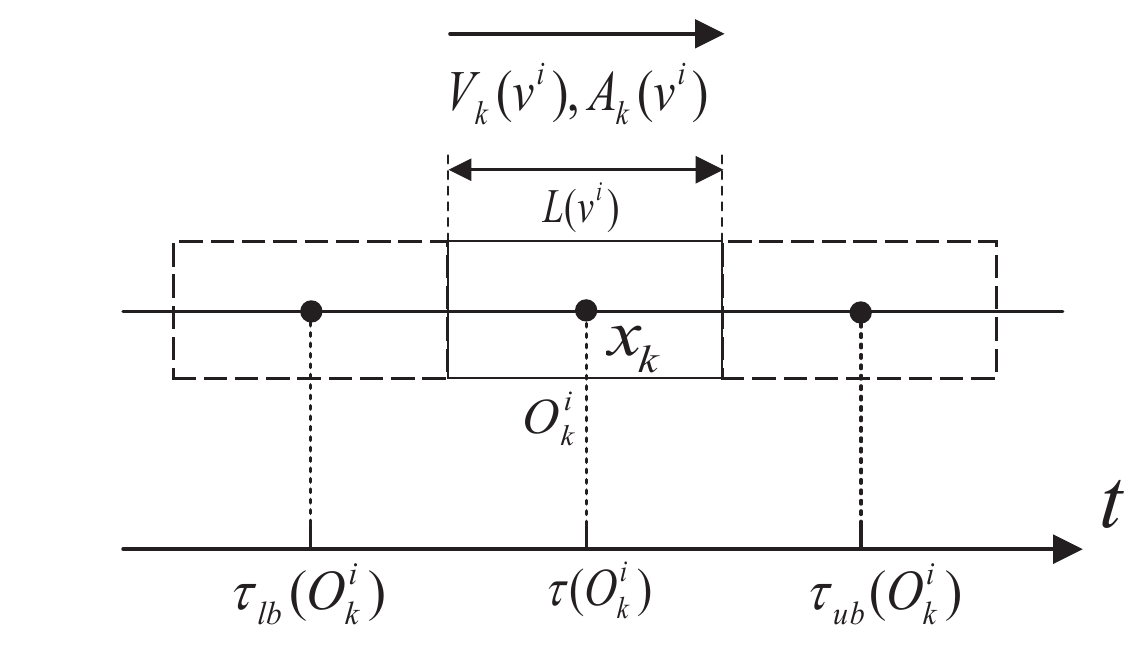}}
\caption{Approximate occupancy time interval calculation for a vehicle with through route. \label{fig:oti}}
\end{figure}

\subsection{Efficient Occupancies Comparison} \label{sec:imp4}
In addition to all the techniques described above, the overall computational complexity of Algorithm \ref{code:dica} can be improved further if we employ an efficient searching method such as the bisection method in the process of time-conflict checking between two DTOT$^*$s.
If we employ this bisection approach for time-conflict checking as shown in Algorithm \ref{code:Enhanced_GetCV}, then the computational complexity of the function \verb+getCV()+ can be improved significantly from $\mathcal{O}(n^2 N_m^3)$ to $\mathcal{O}( n^2 N_m^2 \log_2 N_m)$.

All of the improvement techniques discussed in this section are incorporated into the function \verb+getCV()+ to improve the overall computational complexity of the space-time conflict checking process. Algorithm \ref{code:Enhanced_GetCV} shows this modified \verb+getCV()+ function which is now called \verb+enhanced_getCV()+. In Algorithm \ref{code:Enhanced_GetCV}, $\mathcal{S}^*$ represents the set of already confirmed vehicles that is obtained from the process in Section \ref{sec:imp1} and DTOT$^*$ represents the subset of original DTOT for a vehicle that can be obtained from the approach in Section \ref{sec:imp2}. The function \verb+getOTI()+ within the \verb+while+ loop is now replaced by the new function \verb+getEstOTI()+ that calculates the occupancy time interval approximately as described in Section \ref{sec:imp3}. Lastly, the approach for efficient time conflict checking that is presented in Section \ref{sec:imp4} is implemented throughout the \verb+while+ loop of the DICA algorithm.

\begin{algorithm}[!t]   
	\caption{enhanced\_getCV$(\mathcal{S}^*, DTOT(v^i))$}
%	\caption{enhanced\_getCV$(DTOT_i)$}
		\begin{algorithmic}[1]
		%    	\caption{DTOT Intersection Control Algorithm}	
		%		\STATE
		\STATE Let $\mathcal{S}^*$ be the reduced set of $\mathcal{S}$.
		\STATE Let $DTOT^*$ be the reduced $DTOT$.
		\STATE $\mathcal{C} = \emptyset$
		\FOR {$v^j$ in $\mathcal{S}^*$}
		%			\STATE 
		%			\STATE
		\FOR {$O^j_{k_j}$ in $DTOT^*(v^j)$} \label{O}
%		\STATE
		\IF {$v^j$ not in $\mathcal{C}$}
		\STATE $high = \vert DTOT^*(v^i) \vert - 1$
		\STATE $low = 0$
		\WHILE {$low \ne high$}
%		\STATE
		\STATE $middle = {(high + low)/ 2}$
%		\STATE $O_i \gets Reduced\_DTOT_i$
		
		\STATE Call getEstOTI($O^j_{k_j})$ $\rightarrow I(O^j_{k_j})$
		\STATE Call getEstOTI($O^i_{middle}$) $\rightarrow I(O^i_{middle})$
		\IF {$I({O}^j_{k_j}) \cap I({O}^i_{middle}) \neq \emptyset$}
		\STATE Assign $\tau_{lb}({O}^j_{k_j}) \rightarrow v^j.firstTimeAtCollision$
		\STATE Push $v^j$ into $\mathcal{C}$
		\STATE Sort $\mathcal{C}$ in ascending order of \emph{firstTimeAtCollision}
		\ELSIF{$\tau({O}^j_{k_j}) > \tau({O}^i_{middle})$}
		\STATE $low = middle$
%		\IF {$O_v.time > O_i.time$}
%		\STATE $low = middle$
		\ELSIF{$\tau({O}^j_{k_j}) < \tau({O}^i_{middle})$}
		\STATE $high = middle$
%		\ENDIF
%		\STATE
		\ENDIF

%		\STATE
		\ENDWHILE
		%				\STATE
		\ENDIF
%		\STATE
		\ENDFOR
		%			\STATE			
		%			\STATE
		\ENDFOR
%		\STATE
%		\STATE Sort $CV_i$ based on the times of conflicting occupancies of corresponding vehicles.

	\end{algorithmic}
	
	\label{code:Enhanced_GetCV}
\end{algorithm}

\begin{prop} \label{prop:edica}
Enhanced DICA has $\mathcal{O}(\alpha n^2 L_m  \log_2 L_m)$ computational complexity where $\alpha := \alpha_1^2 \alpha_2 \ll 1$, $n$ is the number of vehicles already confirmed to cross an intersection, and $L_m$ is the maximum length of intersection crossing routes in an intersection.
\end{prop} 
\begin{proof}
First, note that the only part in Algorithm \ref{code:dica} that is affected by this proposed enhancement is that the number of confirmed vehicles to be considered for a space-time conflict check is reduced from $n = \vert \mathcal{S} \vert$ to $\tilde{n} = \vert \mathcal{S}^* \vert$ where $\tilde{n} = \alpha_1 n$ and $\alpha_1 \in (0, 1]$. Thus, in Algorithm \ref{code:dica}, the functions \verb+enhanced_getCV()+ and \verb+updateDTOT()+ are now called $\alpha_1 n$ times. 
Next, we also note that, since nothing is changed due to this improvement in the \verb+updateDTOT()+ function whose computational complexity is already significantly lower than that of the function \verb+getCV()+, it suffices to analyze the computational complexity of the function \verb+enhanced_getCV()+ presented in Algorithm \ref{code:Enhanced_GetCV} for the overall computational complexity of the enhanced DICA. 

Now, as one can see in Algorithm \ref{code:Enhanced_GetCV}, the entire block within the outer \verb+for+ loop is executed for $\alpha_1 n$ times since the number of confirmed vehicles to be checked for a space-time conflict with the vehicle $v^i$ is reduced from $n$ to $\alpha_1 n$ due to the approach discussed in Section \ref{sec:imp1}. Then, within the \verb+for+ loop, for each vehicle $v^j$ in the set $\mathcal{S}^*$, occupancies from each vehicle's DTOT are evaluated for space and time conflict which \emph{typically} requires $N_m^2$ times occupancy comparison operation where $N_m$ is the maximum number of occupancies in a vehicle's DTOT. However, in the \verb+enhanced_getCV()+ function, we first note that the maximum number of occupancies for each vehicle's DTOT to be tested for space-time conflict is reduced from $N_m$ to $\tilde{N}_m$ where $\tilde{N}_m = \alpha_2 N_m$ and $\alpha_2 \in (0, 1]$ due to the approach presented in Section \ref{sec:imp2}. Another important improvement is that the computational complexity for the occupancy time interval calculation is improved from $\mathcal{O}(N_m)$ to $\mathcal{O}(1)$ within another enhanced function \verb+getEstOTI()+ as discussed in Section \ref{sec:imp3}. Therefore, the overall computational complexity of the outer \verb+for+ loop can be estimated as $\mathcal{O}(\alpha_1 \alpha_2^2 n N_m^2)$. However, note that this is the case when we use the same occupancies comparison method as used in the original \verb+getCV()+ function. As shown in Algorithm \ref{code:Enhanced_GetCV}, the process of occupancies comparison is now performed based on the bisection search method. Roughly speaking, for given $n$ and $N_m$, this efficient search method improves the overall computational complexity of the function from $\mathcal{O}(n N_m^2)$ to $\mathcal{O}(n N_m \log_2 N_m)$ as discussed in Section \ref{sec:imp4}. If we combine this and others discussed above for the overall computational complexity of the \verb+enhanced_getCV()+ function, then we have $\mathcal{O}(\alpha_1 \alpha_2 n N_m \log_2 N_m)$. Recall that the \verb+enhanced_getCV()+ function is called at $\alpha_1 n$ times in the main \verb+while+ loop as discussed above, we have $\mathcal{O}(\alpha_1^2 \alpha_2 n^2 N_m \log_2 N_m)$ as the overall computational complexity of DICA.

As we have analyzed already in the proof of Proposition \ref{prop:dica}, $N_m$ is linearly proportional to the maximum length of intersection crossing routes $L_m$. Hence, if we substitute $L_m$ for $N_m$, then we finally have $\mathcal{O}(\alpha n^2 L_m \log_2 L_m)$ as the overall computational complexity of enhanced DICA where $\alpha := \alpha_1^2 \alpha_2 \ll 1$. 
\end{proof}

\section{Simulation} \label{sec:sim}

In this section, we present simulation results that demonstrate the improved performance of the enhanced DICA over the original algorithm. The performance of the enhanced algorithm is also compared with that of an optimized traffic light intersection control.

\subsection{Simulation Setup}
To evaluate the performance of the original DICA and the enhanced DICA, we implemented both algorithms in SUMO \cite{krajzewicz2012recent}, and performed extensive intersection traffic simulations. The simulated situation is an intersection crossing traffic on a typical isolated four-way intersection with three incoming lanes, one of which is a dedicated lane for left-turning vehicles, and two outgoing lanes on each road. 
%as shown in Figure \ref{fig:screenshot}. 
We set $70\ km/h$ as the maximum allowed speed $v_m$ for all incoming vehicles. We let vehicles approach an intersection with different speeds when they enter into the communication region of the intersection to make the simulation more realistic. Specifically, when a new vehicle is spawned outside of the communication region, we assign the initial speed of the vehicle randomly within the range from 40\% to 100\% of the maximum allowed speed $v_m$. Thus, a vehicle keeps this random initial speed until it enters the communication region and then it either follows another vehicle or is confirmed by ICA with a feasible DTOT.
The maximum acceleration ($a_{max}$) and deceleration ($a_{min}$) rates for vehicles that are used in simulations are $2\ m/s^2$ and $4.5\ m/s^2$, respectively. The size of a vehicle used in simulations is $5$ meters long and $1.8$ meters wide. The distance from the enter line of the communication region to that of the intersection region is set as $50\ m$.
The time step that is used in simulations is $0.05$ seconds. In most cases, a simulation terminates when the simulation time reaches $10$ minutes. 

In our simulations, vehicles are spawned according to several random variables in order to generate various traffic volumes as well as traffic patterns. Specifically, $p_V$ is the probability that a vehicle is spawned. $p_L, p_S, p_R$ are the probabilities for this new vehicle to have a route of left-turning, through or right-turning. Thus, by adjusting $p_V$, we can generate various traffic volumes.  
As shown in Table \ref{table:params-e-dica}, we set $p_L = 0.2$, $p_S = 0.6$, and $p_R = 0.2$ for all traffic volume cases to generate $20\%$ of all incoming vehicles for left turing, $60\%$ for going straight, and the other $20\%$ for right turning. 
We use three random seeds to generate three different intersection traffic patterns for each traffic volume. Thus, to obtain simulation data for each traffic volume, we run three simulations with different traffic patterns for each simulation and then use the averages of these simulation results as the result for each traffic volume case. 
The intersection crossing traffics generated in most of our simulations are balanced traffics in the sense that the numbers of vehicles generated in each incoming road are about the same. However, for a simulation to show the starvation free property of the proposed DICA algorithm, the intersection traffic is purposely designed to be unbalanced where the number of vehicles for minor approaching roads is roughly 30\% of the vehicles on major roads.

 \begin{table}[!t]
  \begin{center}
  \caption{Parameters used for various traffic volumes and patterns. ($*$ Expected number of vehicles per 10 minutes.)\label{table:params-e-dica}}
{\begin{tabular}{lc}\toprule
Parameter & Value \\
 		\midrule 
 		Traffic volumes$*$ & 100 / 200 / 300 / 400 / 500 \\
 		%\hline
 		$p_V$ & 0.03 / 0.06 / 0.08 / 0.11 / 0.14 \\ 
 		%\hline
 		$p_L$  & 0.20  \\ 
 		%\hline
 		$p_S$  & 0.60   \\ 
 		%\hline
 		$p_R$ & 0.20  \\ 
 		%\hline
 		Random seeds  & 12 / 21 / 66   \\ 
		\bottomrule
	\end{tabular}}{}
  \end{center}
\end{table}

In the following discussion on our simulation results, \emph{simulation time} means the simulated time used in simulation program and \emph{computation time}, which will be discussed later in Section \ref{sec:sim:time}, means the actual elapsed time that it takes for a computer to run a simulation. 
Also, in Section \ref{sec:sim:perf}, the traffic control performance of the enhanced DICA is compared with that of a traffic light algorithm with fixed cycles. To have a comparable traffic light program, we computed the optimal signal cycles for different traffic volume cases by using the exponential cycle length model $C_0 = 1.5 L e^{1.8 Y}$ from \cite{cheng2005development}. In the model, $L$ represents the total lost time within the cycle. The lost time for each phase is assumed to be $4$ seconds \cite{HCM2000}. Thus, $L = 4 \times 4\ s = 16\ s$. $Y$ is the sum of critical phase flow ratios. The duration of the yellow light of each phase is 3 seconds. 

All simulations were run on a 64bit Windows computer, and its processor is Intel(R) Core(TM) i7-4770 CPU @ 3.40 GHz with 8 GB RAM. The interface programs with SUMO were coded in Python.

\subsection{Simulation Results}
%
%\begin{figure}[!tb]
%	\centering
%	%	\epsfig{file={figures/screenshot}, width=2.4in}
%	\includegraphics[width=2.8in]{figures/screenshot}
%	\caption[A screenshot of simulation]{A screenshot of simulation which illustrates a situation when vehicles with conflicting routes cross an intersection simultaneously.}
%	\label{fig:screenshot}
%\end{figure}

Computation times and performances of three different traffic patterns for all five volume cases are recorded from simulations. 
%A screenshot of simulation in SUMO can be found in \cite{lu2016intelligent} which shows that vehicles of different routes are crossing an intersection simultaneously without the occurrence of collisions.

\subsubsection{Computation Time} \label{sec:sim:time}
 
\begin{figure}[!b]
\centering{\includegraphics[width=5.8in]{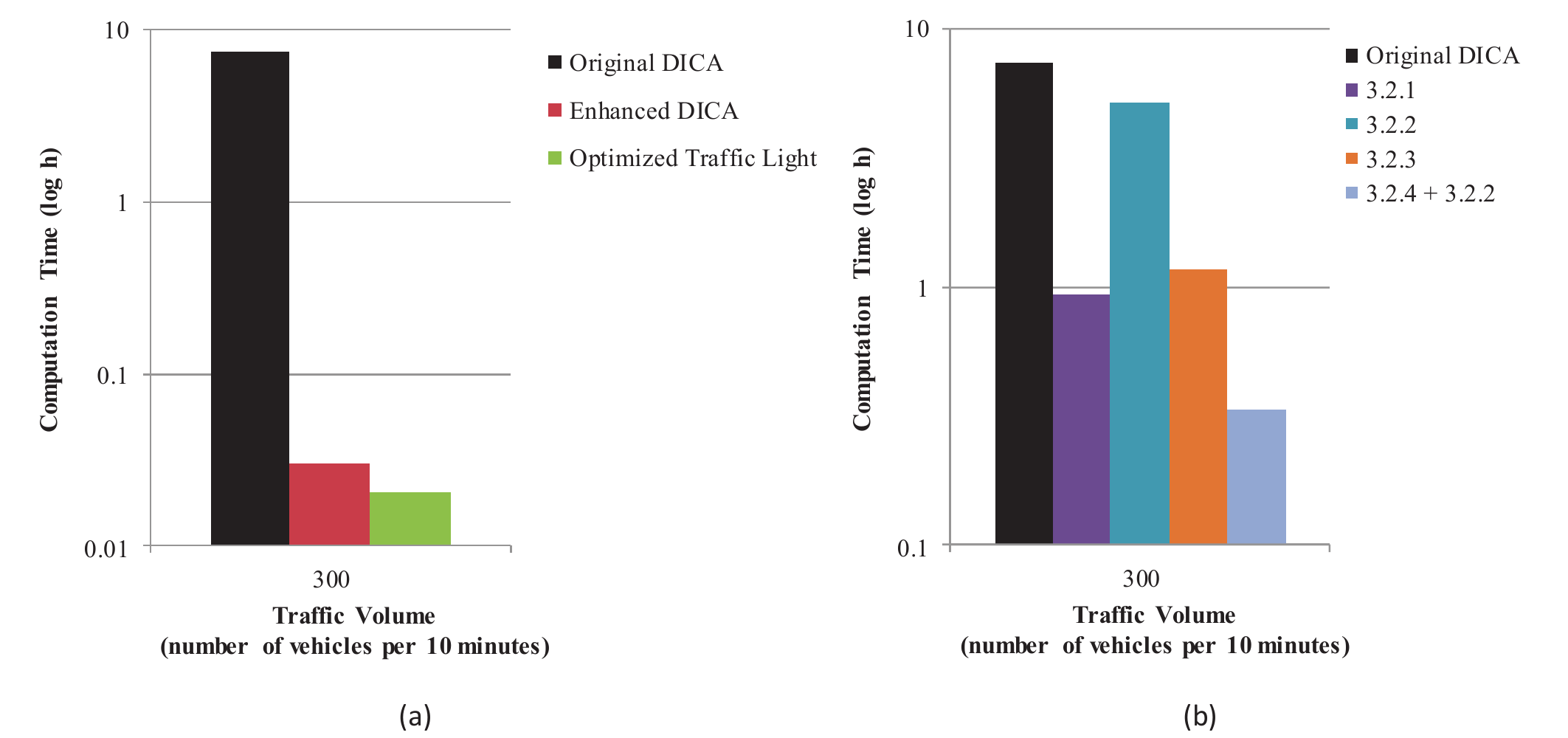}}
\caption[Computation times comparison]{Computation times comparison for traffic volume with 300 vehicles per 10 minutes. (The symbol 3.2.x represents the improvement technique in Section \ref{sec:improv}.x where x = $\{$ 1, 2, 3, 4 $\}$.)
(a) original DICA with different algorithms (b) original DICA with different improvement techniques \label{fig:SimulationTime} 
}
\end{figure}

Figure \ref{fig:SimulationTime} (a) compares the computation times of the original DICA, the enhanced DICA, and the optimized traffic light algorithm. Figure \ref{fig:SimulationTime} (b) shows how much computational improvement was made through each computational improvement technique discussed in Sections \ref{sec:imp1}, \ref{sec:imp2}, \ref{sec:imp3}, and \ref{sec:imp4}. Note that since the computational improvement technique in Section \ref{sec:imp4} is implemented based on the computational improvement technique in Section \ref{sec:imp2}, we had to combine techniques from both Sections \ref{sec:imp4} and \ref{sec:imp2} to show the improvement due to the technique in Section \ref{sec:imp4} indirectly. 
Here, we show the computation times comparison for only one traffic volume case with 300 vehicles per 10 minutes since the trends for other volume cases are similar. 
The vertical axis in Figure \ref{fig:SimulationTime} is the computation time in hour unit which is represented in logarithmic scale. As shown in Figure \ref{fig:SimulationTime} (a), the enhanced DICA that implements all improvements discussed in Section \ref{sec:improv} takes significantly less computation time, i.e. only $0.4\%$ computation time of the original algorithm. 
When we apply each computational improvement technique individually, our result shows that it takes about $11\%$ of the computation time of the original DICA with the technique in Section \ref{sec:imp1}, $59\%$ with the technique in Section \ref{sec:imp2}, $13\%$ with the technique in Section \ref{sec:imp3}, and $6\%$ with techniques in Sections \ref{sec:imp2} and Section \ref{sec:imp4} together. If we combine all of these individual improvements altogether to estimate the collective improvement, then we have about $0.45\%$ computation time of the original DICA which is similar to the computation time result with the enhanced DICA in which all these techniques are implemented. 

Table \ref{table:time} compares the computation times between the enhanced DICA and the optimized traffic light algorithm for all five traffic volume cases. 
From the results shown in the table, we note that the computation time for the optimized traffic light algorithm gradually increases as the traffic volume increases. 
However, since the optimized traffic light algorithm has $\mathcal{O}(1)$ computational complexity, its computation time cannot be affected by the number of vehicles around an intersection. Thus, roughly speaking, one can say that the computation time of the optimized traffic light for a particular traffic volume case is in fact the time required for the simulation software SUMO to run a simulation with the number of vehicles for that particular traffic volume case. Therefore, the actual computation time of the enhanced DICA for a particular traffic volume case can be roughly approximated by subtracting the computation time of the optimized traffic light for the case from the computation time of the enhanced DICA presented in the table. For example, for the traffic volume with 500 vehicles, the actual computation time for the enhanced DICA can be approximated as $0.031 (= 0.058 - 0.027)$ hours which is 1.86 minutes.    
Note that this 1.86 minutes is the computation time taken by the algorithm to handle 500 vehicles. Thus this in turn implies that it takes only 0.2232 seconds to handle each vehicle. 
An exception to this approximation is the case with 100 vehicles traffic volume case where the computation time for optimized traffic light takes longer time than that of the enhanced DICA. The reason for this result can be understood by considering the fact that, in such a low traffic volume situation, the average number of vehicles to be simulated by SUMO at each simulation time step is smaller in the enhanced DICA case since vehicles are crossing an intersection much faster without waiting at an intersection under the enhanced DICA than the optimized traffic light as shown in Section \ref{sec:sim:perf}.

 \begin{table}[tb]	
 	\centering
 	\vspace{-0.5cm}
 	\caption{Computation time comparison between enhanced DICA and optimized traffic light}
 	\smallskip
 	\label{table:time}
 	\begin{tabular}{cccccc}
 		\toprule 
 		Traffic volume &  \multirow{2}{*}{100} &  \multirow{2}{*}{200}& \multirow{2}{*}{300} & \multirow{2}{*}{400} & \multirow{2}{*}{500} \\ 
 		(Number of vehicles per 10 minutes) &  & &  &  &  \\
 		
 		\midrule 
 		
 		Optimized Traffic light (h) & 0.014 & 0.017 & 0.020 & 0.024 & 0.027 \\ 
 		Enhanced DICA (h) & 0.011 & 0.024 & 0.026 & 0.042 & 0.058 \\ 
 		
 		\bottomrule 
 		
 	\end{tabular} 
 \end{table} 
 
\subsubsection{Liveness and Safety}
\begin{figure}[!t]
	\centering
	\includegraphics[width=3.4in]{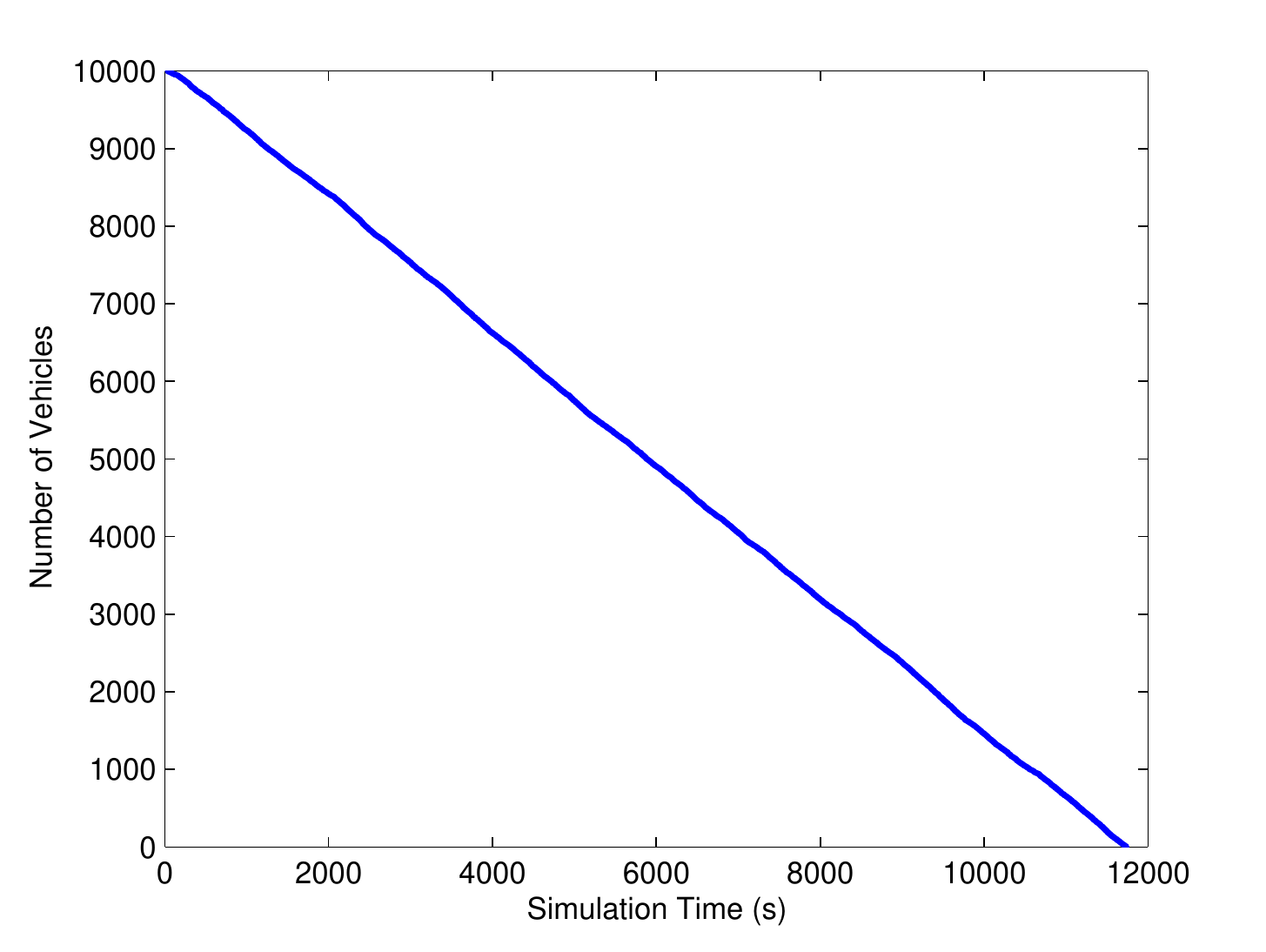}
	\caption{The number of vehicles which wait to cross the intersection over time.}
	\label{fig:waitToCross}
\end{figure}

Although we have theoretically showed the liveness of DICA, it is better to have simulation results that support the theory. Since the simulation in this section is only a verification, we run a simulation with 10, 000 vehicles instead of giving a restriction on the simulation time. The simulation ends after all 10, 000 vehicles have exited the simulation scene. We recorded the number of vehicles that are waiting to cross the intersection at each simulation time step and plot the number profile in Figure \ref{fig:waitToCross}. As shown in the figure, the number of vehicles drops to zero in almost a linear way within a finite time which demonstrates that every vehicle was able to cross the intersection eventually which proves the proposition 1 in Section \ref{sec:anal}. 
We also performed a set of simulations for the case of unbalanced traffic situation where the number of vehicles on minor roads is only 30\% of that of major roads to demonstrate the fairness of DICA.
To show the fairness of the algorithm, we recorded the average trip times for major roads and minor roads respectively for every traffic volumes. 
As shown in Table \ref{table:unbalanced}, one can find that the average trip time of the minor roads is about the same as that of the major roads. This shows that there is not a case that some vehicles cannot get confirmation or will experience a very long time to be confirmed which demonstrates that DICA is starvation free.

\begin{table}[!b]
\begin{center}
\caption{Average trip time comparison between major roads and minor roads in an unbalanced traffic\label{table:unbalanced}}
{\begin{tabular}{llllll}\toprule 
		Traffic volume &  \multirow{2}{*}{100} &  \multirow{2}{*}{200}& \multirow{2}{*}{300} & \multirow{2}{*}{400} & \multirow{2}{*}{500} \\ 
		(Number of vehicles per 10 minutes) &  & &  &  &  \\
		
		\midrule 

		Average trip time on major roads (s) & 6.17 & 6.60 & 7.38 & 8.15 & 10.15 \\ 
		Average trip time on minor roads (s) & 6.21 & 6.57 & 7.38 & 7.90 & 9.63 \\ 
		
\bottomrule
\end{tabular}}{}
\end{center}
\end{table}

 \begin{figure}[!b]
	\centering
	\includegraphics[width=3.8in]{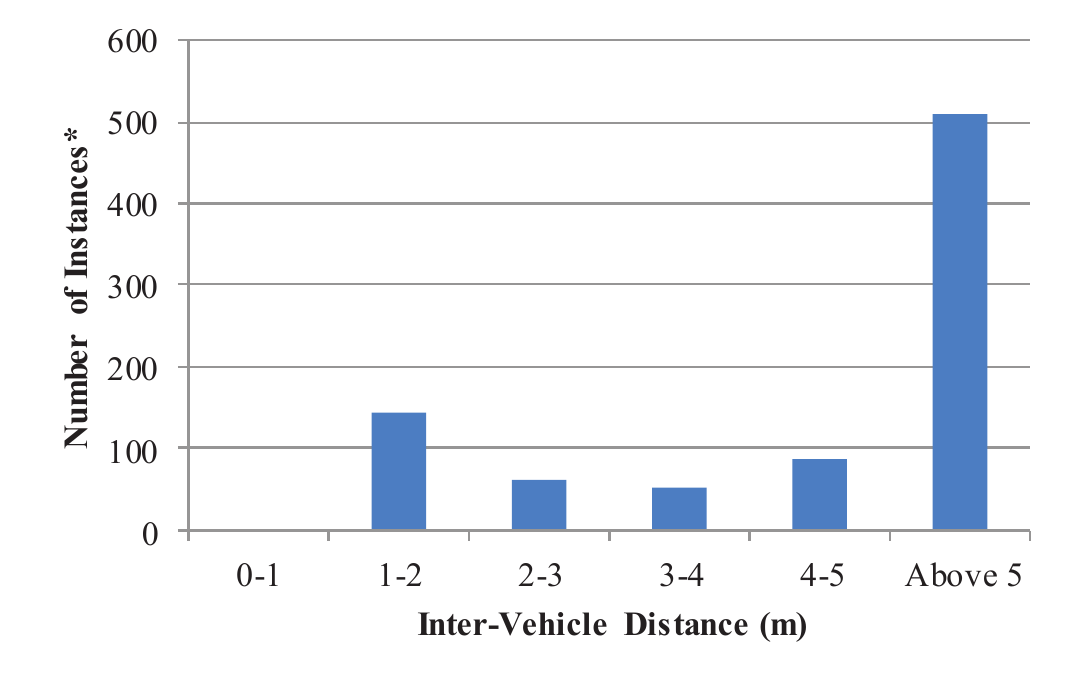}
	\caption[Histogram of the inter-vehicle distance within the intersection.]{Histogram of the inter-vehicle distance within the intersection. (* An instance means the situation when a pair of vehicles are separated by the calculated inter-vehicle distance.)}
	\label{fig:minDistance}
\end{figure}

To validate the safety property (i.e., collision freeness) of DICA through simulation, we computed the inter-vehicle distance between every pair of vehicles within an intersection at every second in simulation time. Since each vehicle is represented as a polygon, a $5\ m$ long and $1.8\ m$ wide rectangle more precisely, we obtained this data based on an algorithm of the shortest distance calculation between two polygons. 
A histogram of the recorded inter-vehicle distances is shown in Figure \ref{fig:minDistance}. 
Clearly, the inter-vehicle distance must be less than or equal to zero if two vehicles are in a collision and must be positive otherwise. 
As one can see from the figure, there is no instance observed throughout the entire simulation with less than $1 m$ inter-vehicle distance, which is a clear indication that there is no collision inside the intersection. Note that Figure \ref{fig:minDistance} is demonstrating the safety of the DICA algorithm, the safety problem that vehicles cannot follow confirmed DTOT correctly pertaining to the robustness of DICA will be studied in our future work.

\subsubsection{Control Performance} \label{sec:sim:perf}

\begin{figure*}[!tb]
\centering{\includegraphics[width=5.6in]{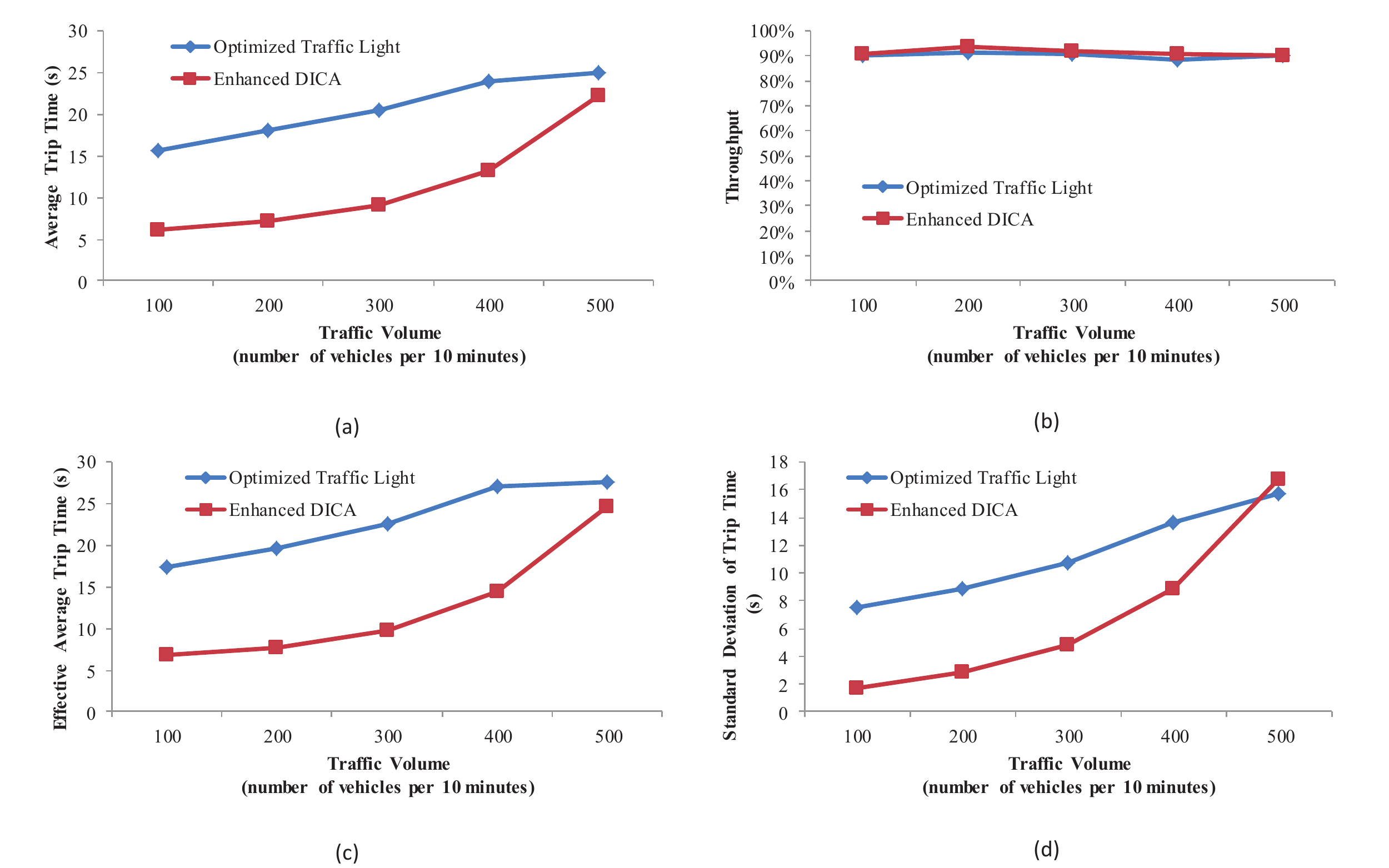}}
\caption[Performance comparison between enhanced DICA and optimized traffic light]{Performance comparison between enhanced DICA and optimized traffic light. 
(a) average trip time (b) throughput (c) effective average trip time (d) standard deviation of trip time \label{fig:performance}}
\end{figure*}

The overall traffic control performance of the enhanced DICA is also evaluated and compared with that of the optimized traffic light algorithm based on the following performance measures. For each vehicle, we recorded the \emph{trip time} that is the time taken for a vehicle from the moment when it enters into the communication region of an intersection until the vehicle completely crosses the intersection region. From the recorded trip time data for all crossed vehicles, we calculated several related statistic information which is the \emph{average trip time} and the \emph{standard deviation of trip time}. Besides these trip time related performance measures, we also calculated the percentage of all crossed vehicles' number against the total number of generated vehicles, which we call the \emph{throughput}.
However, note that neither the average trip time nor the throughput alone is sufficient to correctly evaluate the performance of an algorithm. In fact, both of these measures should be considered together to correctly compare and evaluate the performances of different intersection traffic control algorithms. For this reason, we calculated the ratio of average trip time to throughput, which we call the \emph{effective average trip time}, and believe that this could show the performance of an algorithm better.
Comparison of the performance between the enhanced DICA and the optimized traffic light control algorithm is shown in Figure \ref{fig:performance}. 
From this result, we can see that, since the throughputs of the two algorithms are always similar, the profiles of average trip time and effective average trip time also show similar trends. The enhanced DICA always performs better than the optimized traffic light for the first four traffic volume cases. In the case of the traffic volume with 500 vehicles, the average trip time performance of the enhanced DICA becomes closer to that of the optimized traffic light. Also, the enhanced DICA has a bit larger standard deviation of trip time than the optimized traffic light. In short, the enhanced DICA performs much better than the optimized traffic light from low to medium traffic volume cases while its performance becomes worse and closer to the performance of the optimized traffic light for heavy traffic volumes.  

We note that this result is mainly due to the fundamental difference between individual vehicle based traffic coordination algorithms and traffic flow based coordination algorithms. To see this, we can consider a heavy traffic situation when all incoming roads are congested. In such a situation, we know that most vehicles start to cross an intersection at rest when they are allowed to cross the intersection either by green light under traffic light algorithm or confirmation under the proposed DICA. Under a traffic light control, if a vehicle is crossing an intersection, then it is highly likely that a few more following vehicles can also cross the intersection without being stopped. However, in the case when vehicles are controlled by an individual vehicle based coordination algorithm like our enhanced DICA, it is possible to have a situation where vehicles from different roads are permitted alternatively to cross an intersection, which inevitably results in more frequent stops than the case of traffic light control. This is the reason why the enhanced DICA is performing worse and closer to the optimized traffic light in the heavy traffic volume situation.
In fact, this result reveals an important point that to achieve the best throughput performance, it is necessary to combine both strategies: an individual vehicle based coordination in normal traffic volume and a traffic flow based coordination in a congested situation.  
According to this result, we are currently developing algorithms that incorporate the advantage of traffic flow based algorithms when congested into the proposed enhanced DICA.    

\begin{figure}[!b]
	\centering{\includegraphics[width=3.6in]{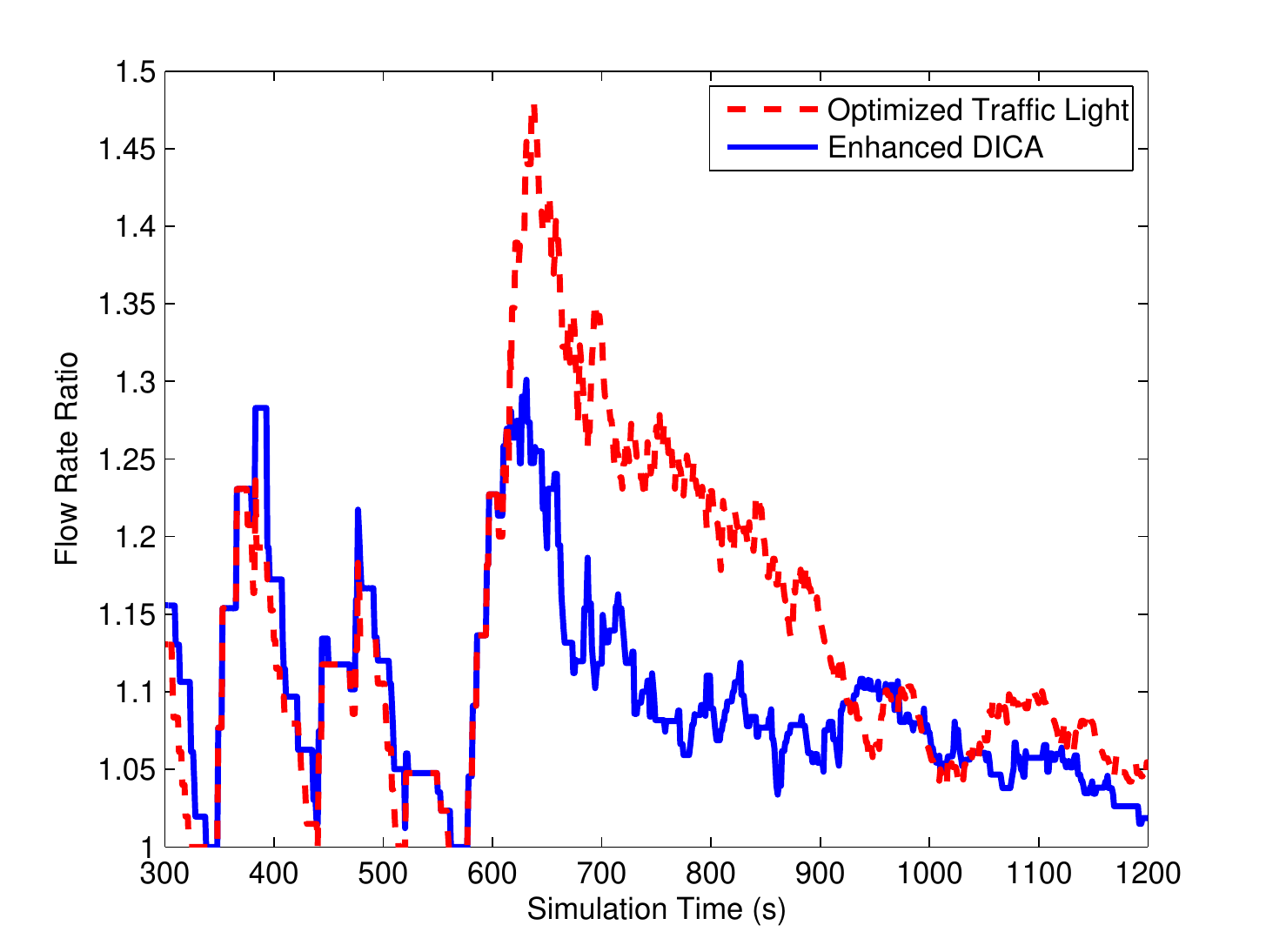}}
	\caption{Flow rate ratio when traffic volume changes from 100 to 500.	\label{fig:ChangeProb}}
\end{figure}

Another simulation was performed to validate the transient traffic control performance of DICA when the traffic volume is changing. We run a simulation with $20$ minutes long simulation time during which the traffic volume increases from the case of $100$ vehicles to $500$ vehicles per $10$ minutes. 
At each simulation time step, the ratio of the vehicle number generated to the number of vehicles that have exited the intersection, which we call the \emph{flow rate ratio}, was calculated to see how much congestion can occur and also how long it takes to address the congestion. 
The flow rate ratio measured during the simulation time is plotted in Figure \ref{fig:ChangeProb}. In this figure, if the flow rate ratio is close to $1$, then it means that all vehicles approached an intersection have already crossed the intersection and there are no vehicles waiting to cross at that time.  
The simulation time starts from $300\ s$ in the figure since the flow rate ratio needs some time to be stable. From the figure, we can also see that before the increase of the traffic volume, the flow rate ratios of the two algorithms are very similar. After $600\ s$ at when the traffic volume is changed to the $500$ vehicles case, the flow rate ratio of the optimized traffic light increased a lot. Figure \ref{fig:ChangeProb} shows that DICA is more resilient to the change of traffic volume than the optimized traffic light.

\section{Conclusion} \label{sec:conc}

In this chapter, We analyzed the computational complexity of the original DICA and enhanced the algorithm so that it can have better overall computational efficiency. Simulation results show that the computational efficiency of the algorithm is improved significantly after the enhancement and the properties of starvation free and safety are guaranteed. We also validated that the overall throughput performance of our enhanced DICA is better than that of an optimized traffic light control mechanism in the case when the traffic is not congested. 

 %input the file 

\chapter{Reactive DICA: an Approach for Expedited Crossing of Emergency Vehicles} \label{chapter:r-dica}% R-DCIA
The problem of evacuating emergency vehicles as quickly as possible through autonomous and connected intersection traffic is addressed in this chapter. DICA is augmented to allow emergency vehicles cross intersections faster and keep the influence on other vehicles' travel as minimum as possible.

\section{Reactive DICA} \label{sec:R-DICA}

The problem we want to solve is that how to let EVs which are driven autonomously cross an intersection as soon as possible under the connected and autonomous traffic environment. 
In the mean time, we aims to keep all other vehicles having similar travel times as when there are no EVs in the traffic. 
In short, our objective is to evacuate EVs through an intersection as quickly as possible while other vehicles' travel times are minimally affected. Note that for simplicity the term emergency vehicle in this dissertation means a vehicle in an emergency status (i.e. with siren and the lights on). 
Same assumptions with our previous work \cite{lu2016intelligent, ieee} are employed in this problem. Overtaking and lane-changing inside the communication region are not allowed which means that vehicles on each lane will keep its lane once it enters the communication region. 
As an approach to give preference to EVs in autonomous traffic, we give priority to EVs in an intersection crossing traffic by optimizing the sequence of crossing vehicles. 
Also, since we are augmenting the original DTOT-based intersection control algorithm, the new algorithm will only be used to coordinate vehicles when there is an EV within the communication region of an intersection while the crossing traffic is controlled the same way as before when all vehicles are normal vehicles inside the communication region. 
Thus, the entering of an EV activates the new algorithm, so we call the augmented DICA the \emph{Reactive DICA} (R-DICA).
DICA is only taking care of head vehicles which reduces computational complexity and communication load of ICA a lot.
However, unlike in DICA, more vehicles are needed to be considered in R-DICA in order to allow EVs to cross an intersection as fast as possible. Specifically, all vehicles on the lane of an EV which are ahead of the EV should be included in the set of vehicles whose intersection crossing order are to be optimized. In the sequel, we call all those vehicles as \emph{vehicles on EV's lane}.
Thus, the set of vehicles that we need to consider for vehicle ordering are all unconfirmed vehicles on EV's lane and also all confirmed vehicles which are not on EV's lane.
All these vehicles can be divided into two types: vehicles whose DTOTs cannot be modified (vehicles who have already entered the intersection or cannot make a stop at the enter line even with maximum deceleration), and vehicles whose DTOTs could be changed (vehicles who are stopping at the enter line of the intersection or are able to make a stop at the enter line, or unconfirmed vehicles who are ahead of the EV). The sequence of vehicles of the latter type is what we can optimize to expedite the crossing of EVs. We define the set of these vehicles as $\mathcal{S}^*$.

Roughly speaking, our approach for fast crossing of emergency vehicles is to assign the highest priority to them and delay confirmation for all other normal vehicles. Thus, incorporating a priority based ordering of vehicles into the basic DICA framework would achieve this goal. 
To find such an optimal vehicle ordering, we formulate an optimization problem based on the \emph{entrance time} of vehicles which is the time a vehicle enters the line of an intersection.  
Let $\mathcal{P}(\mathcal{S}^*)$ be the set of ordered vehicle sequences (or simply called a \emph{sequence} in the sequel) from the set of vehicles in $\mathcal{S}^*$.
Then, if we use $T_e^v$ to represent entrance time of vehicle $v$, a reasonable objective function for our optimization problem would be:
\begin{equation}
	\min_{\mathcal{P}(\mathcal{S}^*)} T_e^{EV}
\end{equation}
where $T_e^{EV}$ is the entrance time of an EV at an intersection.
Thus, to solve this optimization problem, we first need to introduce an approach that determines the entrance time of an EV. 

%Note that if a vehicle is in a sequence, its entrance time will be influenced by the entrance time of its precedent vehicle. 

\begin{figure}[!htb]
	
	\centering

	\includegraphics[width=5.7in]{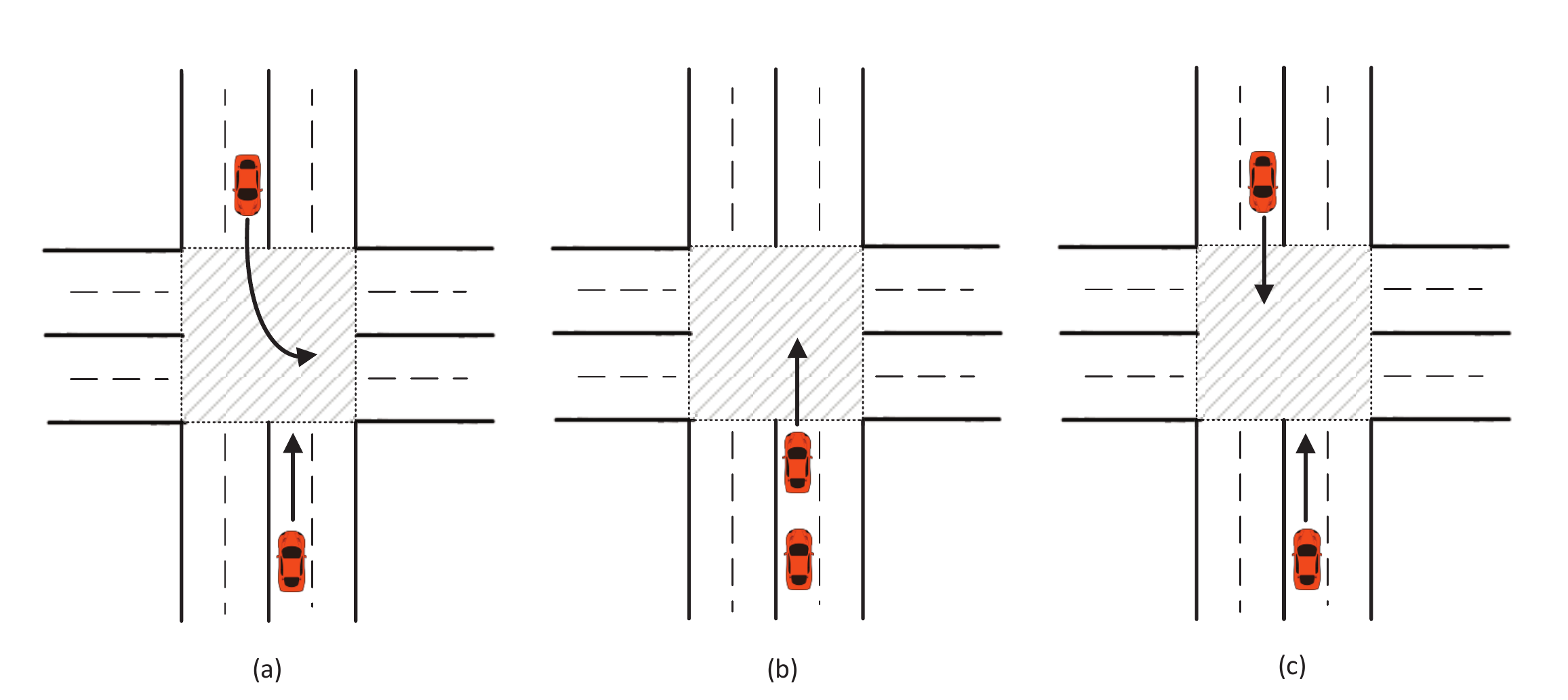}
	\caption{Three different situations for separation time.}
	
	\label{fig:separationTime}
	
\end{figure}

First, we note that some sequences in $\mathcal{P}(\mathcal{S}^*)$ can be eliminated if we impose some constraints for optimal vehicle ordering. For example, the order of vehicles on EV's lane cannot be altered and hence should be preserved. Also, since all confirmed vehicles $\mathcal{S}^*$ are able to stop before the enter line of an intersection, we can allocate higher priorities for vehicles on EV's lane than those in other lanes. 
We use $\bar{\mathcal{P}}(\mathcal{S}^*)$ to denote the set of ordered sequences of vehicles satisfying these constraints.
Now, let us consider a sequence $s$ in the set $\bar{\mathcal{P}}(\mathcal{S}^*)$. 
Then, if we consider the first vehicle $v$ in the sequence $s$, then it is easy to see that the vehicle is always a head vehicle on EV's lane and has a confirmed DTOT. Hence the entrance time of this vehicle $v$ can be determined simply by its $\tau(\mathcal{O}^v_1)$ which is the time when the vehicle $v$ occupies the first occupancy of its DTOT. 
For any other vehicles which are not the first vehicle in the sequence $s$, the way to compute its entrance times is a bit different. We need a time interval between any two successive vehicle in a sequence to ensure safety. This time interval is called \emph{separation time} $\tau_s$. In this chapter, as shown in Figure \ref{fig:separationTime}, we define three separation times for different situations between two vehicles.

\begin{equation}
\tau_s = 
\begin{cases}
\delta_c & \quad 	v_i \otimes v_j \mbox{ Figure \ref{fig:separationTime} (a), or} \\
\delta_s & \quad v_i \prec v_j \text{ or } v_j \prec v_i \mbox{ Figure \ref{fig:separationTime} (b), or} \\
0 & \quad v_i  \odot v_j \mbox{ Figure \ref{fig:separationTime} (c)}
\end{cases}
\end{equation}
where symbols $\odot$ and $\otimes$ are used to represent two vehicles' routes are compatible and conflicting respectively. $v_i \prec v_j$ represents that vehicles $v_i$ and $v_j$ are on a same lane and $v_i$ is following $v_j$.
The separation time's value depends on pavement conditions, vehicle mechanical errors and weather conditions. The focus of this chapter is proposing a coordination algorithm not the determination of these values. Thus, we just approximate the values from current empirical estimations which is widely accepted \cite{chen2010markov}.
Then the expression to compute the entrance time of $v^j$ which is not the first vehicle $v^1$ in the sequence is:

\begin{equation}
%T_e^j = max\{\tau(\mathcal{O}^v_1), T_e^i + \tau_s\}
T_e^j = max\{T_a^j, T_e^i + \tau_s\}
\end{equation}
where $v^i$ is the immediate predecessor of $v^j$ in the sequence, $T_a^j$ is the \emph{predicted arrival time} of the vehicle $v^j$ which is the shortest time for the vehicle to arrive at the enter line of an intersection under the constraints of maximum acceleration and speed without considering other vehicles in a traffic. $T_e^i$ is $v^i$'s entrance time and $\tau_s$ is the separation time between $v^i$ and $v^j$. Starting from the second vehicle in sequence, this equation is iteratively used to compute the entrance time of of each vehicle in the sequence until the entrance time of the emergency vehicle is computed.

Now the complete form of an optimization problem for optimal vehicle ordering to minimize the entrance time of the EV is formulated as follows:
Given predicted arrival times $T_a^{v_i}$ for all $v_i \in \mathcal{S}^*$, find $s^*$ such that 
\begin{eqnarray} \label{eq:opt}
	s^* &=& \min_{s \in \bar{\mathcal{P}}(\mathcal{S}^*)} T_e^{EV} \\ \nonumber
	&s.t.& |T_e^i - T_e^j| \geq 
\begin{cases}	
0 & \quad	v_i  \odot v_j \\
\delta_c & \quad	v_i \otimes v_j \\
\delta_s & \quad	v_i \prec v_j \text{ or } v_j \prec v_i
\end{cases} \\ \nonumber
	&&  T_e^v \geq T_a^v \quad \forall v \in \mathcal{S}^*
\end{eqnarray}

A naive approach to solve the optimization problem in (\ref{eq:opt}) is an exhaustive search in all possible sequences that can be generated from the set $\mathcal{S}^*$. If we suppose that there are $n$ numbers of vehicles in $\mathcal{S}^*$ (i.e. $n=|\mathcal{S}^*|$) and there are $n_{EV}$ numbers of vehicles on EV's lane, then there are $n!/n_{EV}!$ numbers of sequences in $\mathcal{P}(\mathcal{S}^*)$. 
However, if $n$ is becoming large, then the computational time and resources required to solve the optimization problem are increasing significantly. 
Hence it might not be an efficient approach to use an exhaustive search method when we want to solve the problem (\ref{eq:opt}) with many vehicles. Such computation issues of the problem present the need to seek heuristic approaches which are good at solving complex problems in a very short time compared with exhaustive search. Several heuristic optimization approaches like genetic algorithm, ant colony system, artificial neural networks exist in literature. 
\cite{murata1996} used permutation encoding scheme and solved the flowshop scheduling problem with an objective of minimizing the makespan. \cite{yan2013autonomous} proposed a genetic algorithm to optimize the groups of compatible vehicles in a very short time. \cite{hart2005} and \cite{werner2011} reviewed many researches that genetic algorithms can be used to solve job scheduling problems which can meet our requirements. Thus, we also choose to use genetic algorithm (GA) to obtain optimal sequence of vehicles. 

The high level architecture of R-DICA combining GA and DICA is shown in Figure \ref{fig:R-DICA}. 
R-DICA activates GA when ICA detects an EV. Then ICA stop accepting any confirmation of new vehicles which are detected after the EV. All vehicles who belong to $\mathcal{S}^*$ are rearranged to obtain the optimal sequence for the EV's crossing by GA. Then ICA only confirms vehicles who are already included in the set $\mathcal{S}^*$ until the EV exits the intersection. Once the EV is complete out of intersection, ICA switchs back to use DICA to manage normal intersection crossing traffic. 

\begin{figure}[!hb]
	
	\centering

	\includegraphics[width=3.5in]{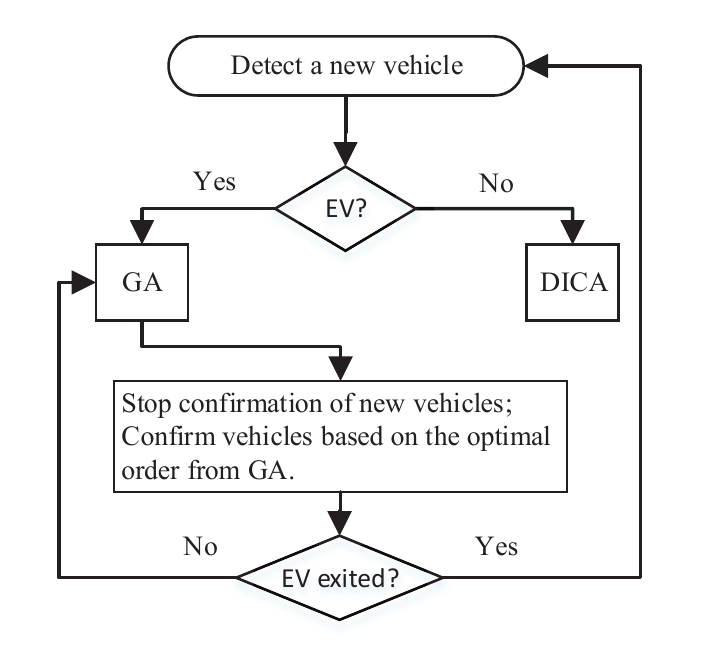}
	
	\caption{Control flow diagram of ICA in R-DICA.}
	
	\label{fig:R-DICA}
	
\end{figure}

\section{Genetic Algorithm for Vehicle Ordering} \label{sec:ga}

In this section, we discuss the details of how GA is used to find the optimal vehicle sequence in (\ref{eq:opt}). 

Genetic Algorithms, which have been widely used to solve problems in computer science, artificial intelligence, information technology and engineering, are techniques of self-organized and self-adapting artificial intelligence mimicking the evolutionary process of creatures in nature \cite{dahal2007evolutionary, hart2005}. A solution in GA is called an \emph{individual} which is encoded compactly to facilitate the processes of crossover and mutation that are essential in a genetic algorithm. A group of individuals is called a \emph{population} in which some individuals are selected as parents to generate offspring through crossover and mutation. Based on some features of each individual, some individuals survive and others die among all the original population and new individuals. Individuals who correspond or near correct solution have a better chance to survive during evolving since they have high objective values which is called \emph{fitness}. Fitness function should be defined properly to evaluate each individual. As introduced above, solutions in GA evolve to adapt the objective problem. Optimal or near-optimal solutions are expected to be obtained after a certain number of generations.
In this chapter, we propose a GA to solve the complex traffic control problem for emergency vehicles in a short time. Permutation encoding scheme is used in the algorithm. And crossover and mutation operators suitable for permutation encoding scheme are devised. 
The proposed GA for vehicle ordering is shown in Algorithm \ref{code:ga}.
Detailed discussion for permutation scheme, crossover, mutation, etc. of the proposed GA are given in the following sections.

\begin{algorithm}[!thb] 
	
	\caption{Genetic Algorithm for Vehicle Ordering}
	
	\begin{algorithmic}[1]
		
		% 	\caption{DTOT Intersection Control Algorithm}
		
		\STATE Generate $N_{pop}$ different individuals randomly for $n$ vehicles in $\mathcal{S}^* \rightarrow \mathcal{I}$
		\STATE $feasibilityCheck(\mathcal{I}) \rightarrow \mathcal{I}$
		%		\STATE 
		\STATE $k = 0$
		\STATE $j = 0$
		\STATE $fitness\_best\_last = 0$
		\STATE
		\WHILE{$k < N_{max} $ and $j < N_{noChange}$}

		%		\FOR {vehicle $i \in$ unconfirmed Head Vehicle set}
		
		%				\STATE
		
		\STATE $crossOver(P_c, \mathcal{I}) \rightarrow \mathcal{I}$
		
		\STATE $feasibilityCheck(\mathcal{I}) \rightarrow \mathcal{I}$
		
		%		\STATE 
		
		%				\STATE
		
		\STATE $mutation(P_m, \mathcal{I}) \rightarrow \mathcal{I}$
		\STATE $feasibilityCheck(\mathcal{I}) \rightarrow \mathcal{I}$
		\STATE 
		
		\STATE $fitness\_best, individual\_best = fitness(\mathcal{I})$
		\IF {$fitness\_best > fitness\_best\_last$}
		\STATE $j = 0$
		\ELSE 
		\STATE $j = j + 1$
		\ENDIF

		%		\STATE 
		\STATE top $N_{pop}$ individuals $\rightarrow \mathcal{I}$ 
		
		\STATE $k = k + 1$
		
		\ENDWHILE
		%				\STATE
		
		%		\ENDFOR
		\STATE Decode $individual\_best$
	\end{algorithmic}
	
	\label{code:ga}
	
\end{algorithm}

In the proposed GA, we first generate a random population $\mathcal{I}$ that contains $N_{pop}$ individuals which are encoded by permutation scheme. 
The function $feasibilityCheck()$ takes a set of individuals and makes modification to the infeasible individuals. Feasible individual corresponds to a sequence of vehicles that does not violate the order of vehicles on EV's lane. After proper modification, the function returns a set containing individuals which are all feasible. 
The function $crossOver()$ then perform crossover on randomly selected pairs of individuals from the population $\mathcal{I}$ with a probability $P_c$ to generate new offspring. Then the feasibility of the offspring is checked. Notice that after crossover, the number of individuals is larger than $N_{pop}$.
Mutation on the produced offspring with probability of $P_m$ is done by function $mutation()$. The mutated individuals also need to be checked for feasibility and modified if needed.
Based on given conditions, each individual in $\mathcal{I}$ is evaluated by a fitness function $fitness()$ which computes the reciprocal of the entrance time of the emergency vehicle in that individual. The highest fitness value and the corresponding individual are recorded. Notice that the fitness can also be obtained by using other metrics like the exit time of the EV, or the trip time of the EV, etc. These metrics will give us similar results. We choose the entrance time because we have the predicted arrival time for each vehicle. Thus, it is easy to implement the algorithm. 
Then we use the top $N_{pop}$ individuals from the original population and offspring to form the new population.
If any of the stopping criteria (maximum number of iterations or the best solution is not updated for a certain number of generations) are met, then the algorithm terminates. Otherwise the algorithm repeats the steps inside the \emph{while} loop.

%Now we introduce each step of the genetic algorithm \ref{code:ga} in details.

\subsection{Chromosome Encoding and Feasibility Check} \label{encoding}

\begin{figure} [!ht]
	
	\centering		
	\includegraphics[width=2.6in]{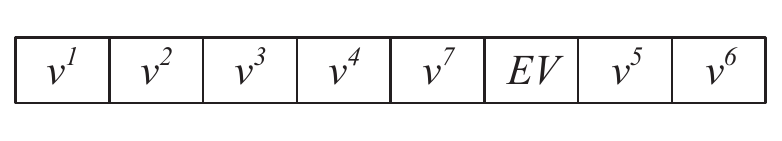} 
	\caption{An example of the permutation encoding scheme, the left-most vehicle has the highest priority while the right-most one has the lowest priority.}
	\label{fig:permutation}
	
\end{figure}

\begin{figure} [!h]
	
	\centering	
	\includegraphics[width=4.3in]{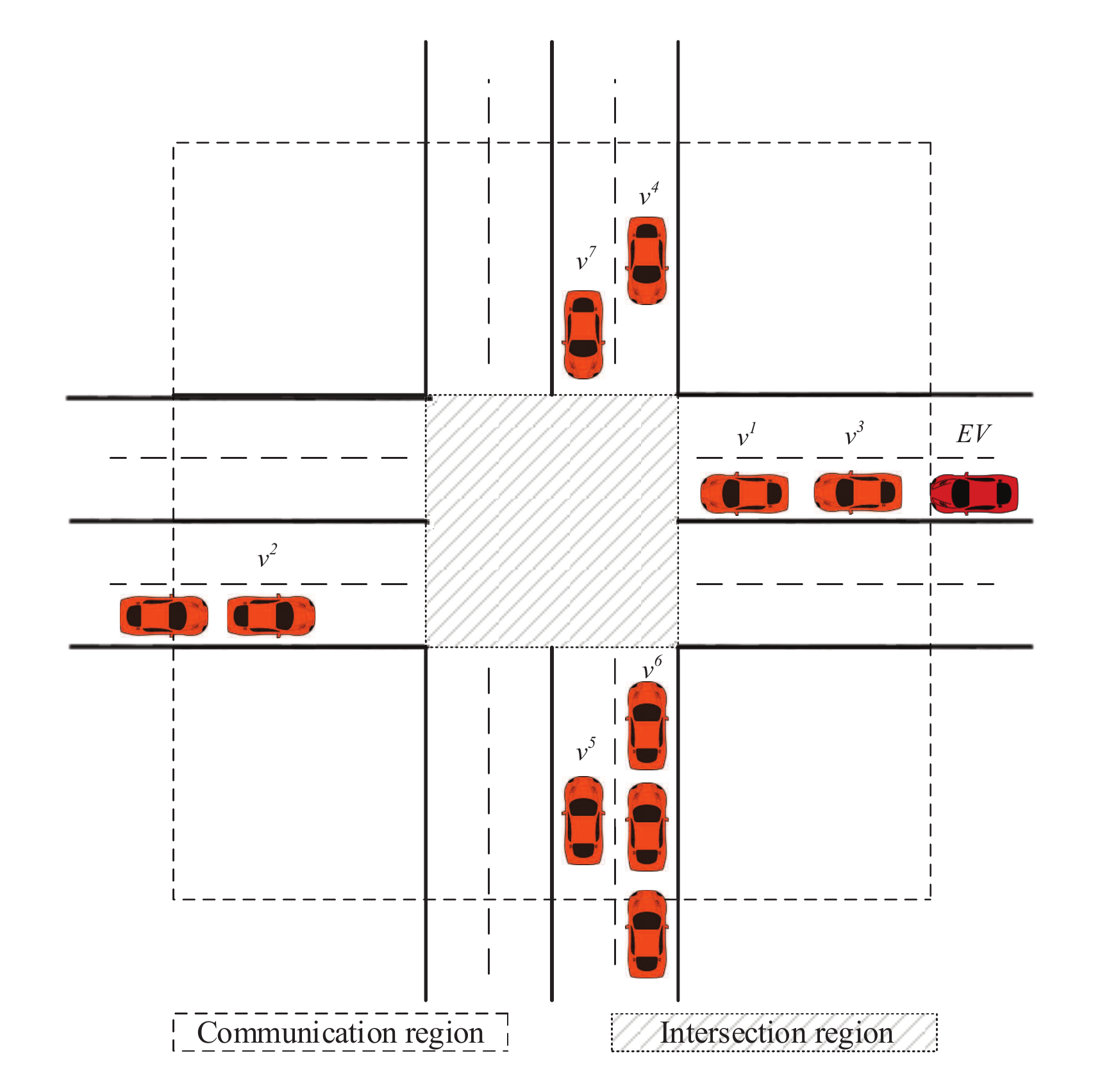}
	\caption{Example situation of vehicles whose sequence to be optimized in intersection space, note: vehicles on EV's lane are: $v^1, v^3, EV$}
	
	\label{fig:vehiclesSequence}
	
\end{figure}

Instead of using the popular binary encoding scheme for genetic algorithms, we choose to use permutation encoding scheme which is more suitable to find an optimal sequence for vehicle ordering. As shown in Figure \ref{fig:permutation}, the individual corresponds to a sequence of vehicles which is $\{v^1, v^2, v^3, v^4, v^7, EV, v^5, v^6\}$ with $v^1$ on the leftmost is the first and the rightmost vehicle $v^6$ is the last one. Different chromosomes denote different sequences of vehicles. Once an individual is created, it is not always true that the corresponding sequence is a feasible one since vehicles' order on EV's lane cannot be altered. Every new generated individual should be checked against the sub-sequence of vehicles on EV's lane for feasibility. Figure \ref{fig:vehiclesSequence} is provided to have a visual impression of the situation when vehicles' sequence needs to be optimized. In the figure, $v^1, v^3$ and $EV$ are the vehicles on EV's lane whose order could not be altered. And note that except vehicles on EV's lane, all other vehicles who are not a head vehicle are not part of the sequence. The vehicles from South and West who are not head vehicles are such vehicles that will be confirmed only after EV exits. If an individual is not feasible, the corresponding bits of vehicles on EV's lane should be changed to conform the correct relative order. The function $feasibilityCheck()$ is making the corresponding modifications on an infeasible individual. An example of adjustment according to the sequence of vehicles who are ahead of the EV on the same lane is shown in Figure \ref{fig:crossover}.

\subsection{Crossover and Mutation}

Two individuals perform crossover to generate offspring if they are selected to be parents. The offspring inherit features (i.e. gene structures) from their parents. Different encoding scheme has different crossover operator since they have different gene structures. 
For the most popular binary encoding scheme, it is easy to do crossover and mutation since a chromosome only contains binary bits. For our permutation encoding scheme, we choose to apply one-point crossover \cite{dahal2007evolutionary} which is implemented in the function $crossOver()$. 
As shown in Figure \ref{fig:crossover}, the same bits may exist in one chromosome after the parts behind the randomly chosen position are swapped. In the second step in the figure, the two children have same bits $\{v^2, v^3\}$ and $\{v^5, v^6\}$ respectively. To generate correct chromosomes, we adjust the chromosome of one child by swapping those same bits from another child's chromosome while preserving the relative ordering of parents. 
Note that the new chromosomes may also not be feasible since the order of vehicles on EV's lane in a chromosome may not be the same as the actual order. If this happens, since the order of the vehicles on EV's lane cannot be changed, we manually adjust the relative order of vehicles to be the correct order to have a feasible chromosome. For example in Figure \ref{fig:crossover}, we adjust the order of $v^1$ and $v^3$ for the second child in the last step.  Feasibility check and adjustment are done by the function $feasibilityCheck()$.

\begin{figure} [!t]

	\centering

	\includegraphics[width=5.4in]{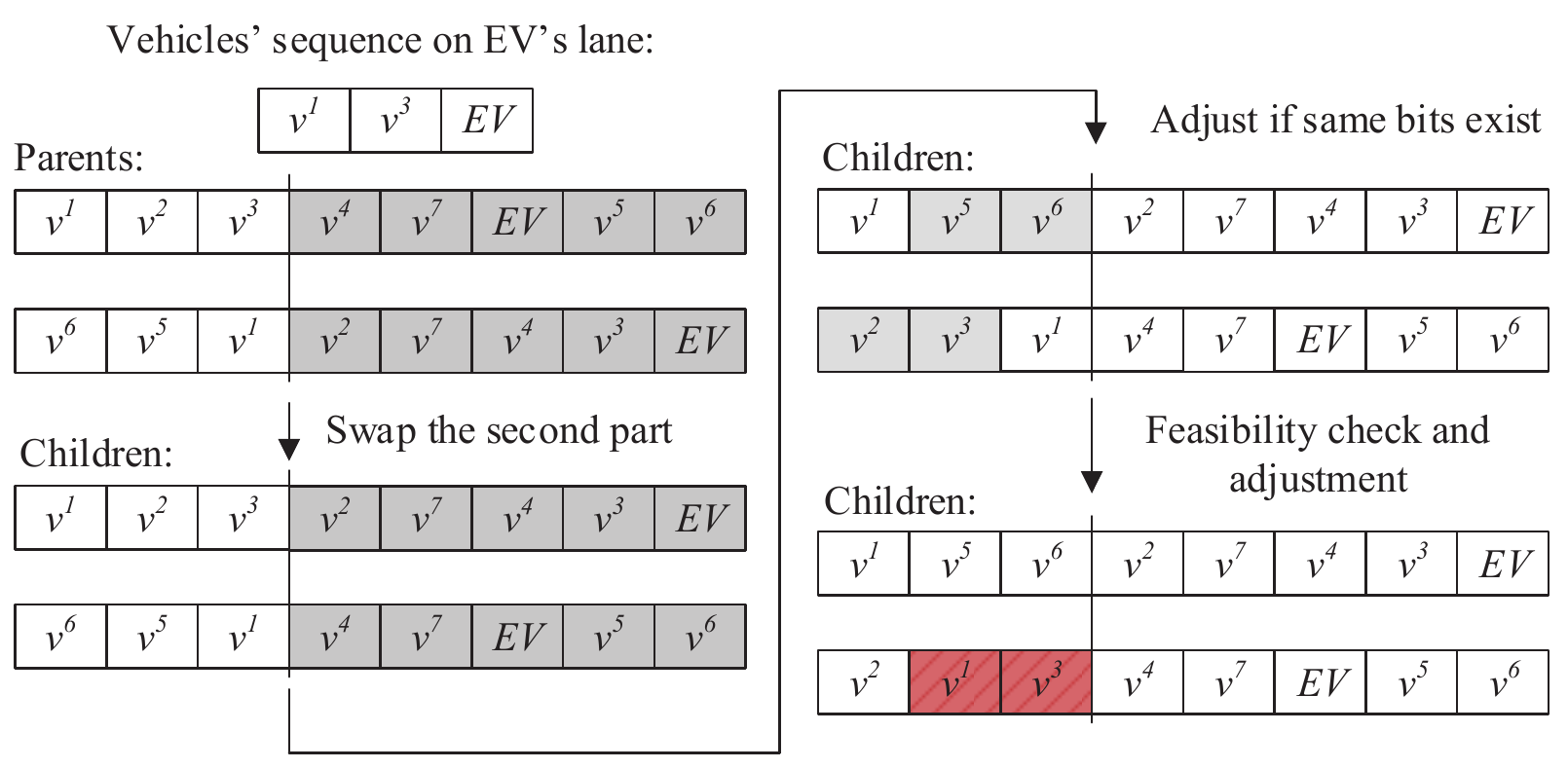} %crossover}

	\caption{An example of one-point crossover. Relative ordering of parents is preserved when the chromosomes are adjusted due to the existence of same bits. Feasibility is checked for the two children based on vehicles' sequence on EV's lane and adjustments are made.}

	\label{fig:crossover}

\end{figure}

Similar to probabilistically selecting two individuals for crossover, we apply mutation on produced chromosomes based on a given probability by the function $mutation()$. Different with binary encoding scheme's mutation which could be done by simply changing the value of a randomly selected bit from 1 to 0 or 0 to 1, our permutation encoding scheme exchanges the bits on two randomly chosen positions to obtain a new chromosome. As shown in Figure \ref{fig:mutation}, positions of $v^2$ and $EV$ are randomly chosen to exchange values and feasibility check based on vehicles' order on EV's lane is performed after mutation. 

\begin{figure} [!ht]

	\centering

	\includegraphics[width=2.6in]{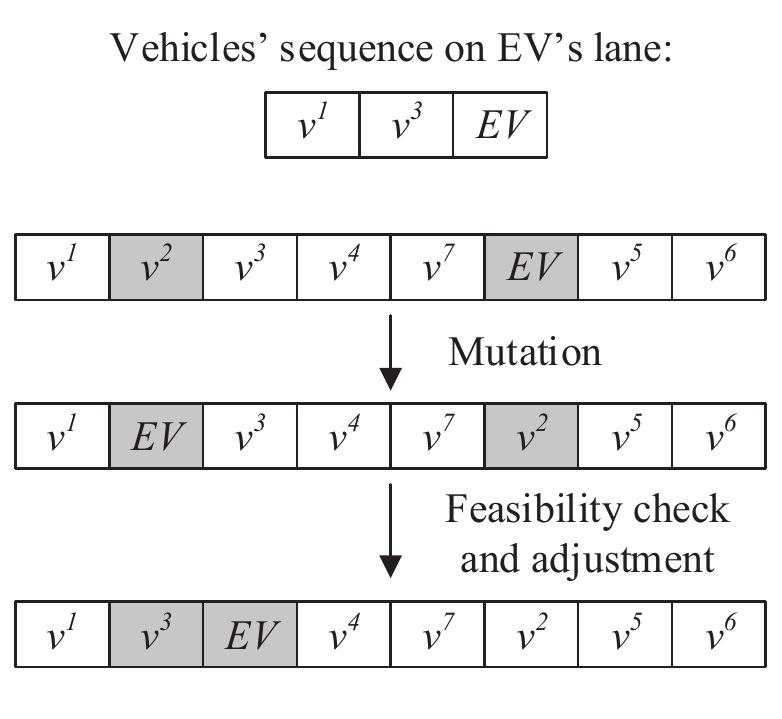}

	\caption{An example of mutation. Feasibility is checked for the mutated chromosome based on vehicles' sequence on EV's lane and adjustment is made. $v^2$ and $EV$ are randomly chosen to swap to perform mutation. $EV$ is swapped with $v^3$ to conform vehicles' sequence on EV's lane.}

	\label{fig:mutation}

\end{figure}

\subsection{Fitness and Generating New Generation}

The reciprocal of the entrance time of the emergency vehicle is defined as the fitness of an individual in our proposed GA. The entrance time of the emergency vehicle can be determined as discussed in Section \ref{sec:R-DICA}. For an individual in a population, the higher the fitness is, the closer the corresponding solution is to the optimal solution. Among all individuals in the population and the offspring produced, the top $N_{pop}$ individuals with respect to the fitness values are selected to form the next generation.

\subsection{Stopping Criterion}
The constant $N_{max}$ represents the number of maximum generations and $N_{noChange}$ represents the number of continuous generations that solutions are not changed. As shown in In Algorithm \ref{code:ga}, if the best solution is not updated after $N_{noChange}$ generations or the $N_{max}$th generation has been reached, then the algorithm terminates and stops searching a better solution.

\section{Simulation} \label{sec:sim}

The performance of the proposed optimization approach for EVs is evaluated against the DICA and a reactive traffic light algorithm which is explained below. All simulations are implemented in an open-source traffic simulator SUMO (Simulation of Urban MObility) \cite{Krajzewicz2012}. The default traffic management for intersections in SUMO is not used and the control algorithms are programmed as Python applications. The TraCI is used for the interaction between the Python applications and SUMO. Configurations for intersections in the simulation and corresponding results are described in this section, followed by discussions on obtained results. 

\subsection{Simulation Setup}

Extensive simulations of different traffic volumes are performed on an isolated perfect 4-way intersection where each approach has three incoming lanes and two exit lanes. 
Similar to real intersections in the United States, among the three incoming lanes, the left-most lane is dedicated for left-turn vehicles, and through vehicles can use the other two lanes. The right-most lane can also be used by right-turn vehicles. 
All roads have the speed limit of $v_m = 70 km/h$. The maximal acceleration ($a_{max}$) and deceleration ($a_{min}$) for all vehicles are set to be $2 m/s^2$ and $-4.5 m/s^2$. In the simulation, for simplicity we used the same size for normal vehicles and EVs that they both have $5$ meters length and $1.8$ meters width. We let vehicles approach an intersection with different speeds when they enter into the communication region of the intersection to make the simulation more realistic. In detail, when a new vehicle is spawned outside of the communication region, the speed of the vehicle is set with a random value within the range from $40\%$ to $100\%$ of the maximum allowed speed $v_m$. 
In those cases where vehicles need to stop just before the enter line of the intersection region to avoid potential collisions with other confirmed vehicles, the distance between the enter line of the communication region and the enter line of the intersection region should be long enough for a vehicle to be able to stop from maximum allowed speed $v_m$. 
Thus, it is easy to conclude that the distance should be at least $-v_m^2 / (2a_{min}) \approx 42.01 m$. So, in simulation, we set the distance between the enter lines of the communication region and the intersection region as $50 m$.  

Vehicles are generated randomly on each road with a randomly assigned intersection route. Every generated vehicle has the probability of $p_{EV}$ to be an EV, otherwise it will be a normal vehicle. 
In our simulation, an emergency vehicle is generated only when there is no such vehicle inside the communication region. 
To create variations on the traffic pattern, we use several different random seed numbers to generate different traffic patterns and make the simulations reproducible. 
Table \ref{rv_simulation} summarizes the parameters used for various traffic volumes and patterns that were employed in many of our simulations where $p_V$ corresponds to traffic volumes, $p_L, p_S$ and $p_R$ are the probabilities for a generated vehicle to take Left, Straight, and Right routes respectively. For every traffic volume, we run three simulations with different traffic patterns and then use the averages of these simulation results as the result for each traffic volume case. For the genetic algorithm, we set $N_{pop} = 100, N_{noChange} = 10, P_c = 0.85, P_m = 0.05$ and $N_{max} = 100$. 

Simulations were run by $0.05$s time step. We terminate each simulation when the simulation time reaches one hour. The \emph{simulation time} here represents the simulated time in simulation programs. And the \emph{computation time} which will be used in the following discussion is the time that a computer takes to run a simulation program. All simulations were run on a 64bit Windows computer, and its processor is Intel(R) Core(TM) i7-4770 CPU @ 3.40 GHz with 8 GB RAM.

\begin{table}[tb]
	\centering
	\vspace{-0.5cm}
	\caption{Parameters used for various traffic volumes and patterns. ($*$ Expected number of vehicles per 10 minutes.)}
	\smallskip
	\label{rv_simulation}
	\begin{tabular}{lc}
		\toprule
		Parameter & Value \\
		\midrule
		Traffic volumes$*$ & 100 / 200 / 300 / 400 / 500 \\
		
		$p_V$ & 0.03 / 0.06 / 0.08 / 0.11 / 0.14 \\ 
		
		$p_L$  & 0.20  \\ 
		
		$p_S$  & 0.60   \\ 
		
		$p_R$ & 0.20  \\ 
		
		$p_{EV}$ & 0.02  \\ 
		
		Random seeds  & 12 / 21 / 66   \\ 
		\bottomrule	
		
	\end{tabular} 
\end{table}

\begin{figure} [!h]
	
	\centering
	
	\includegraphics[width=3.5in]{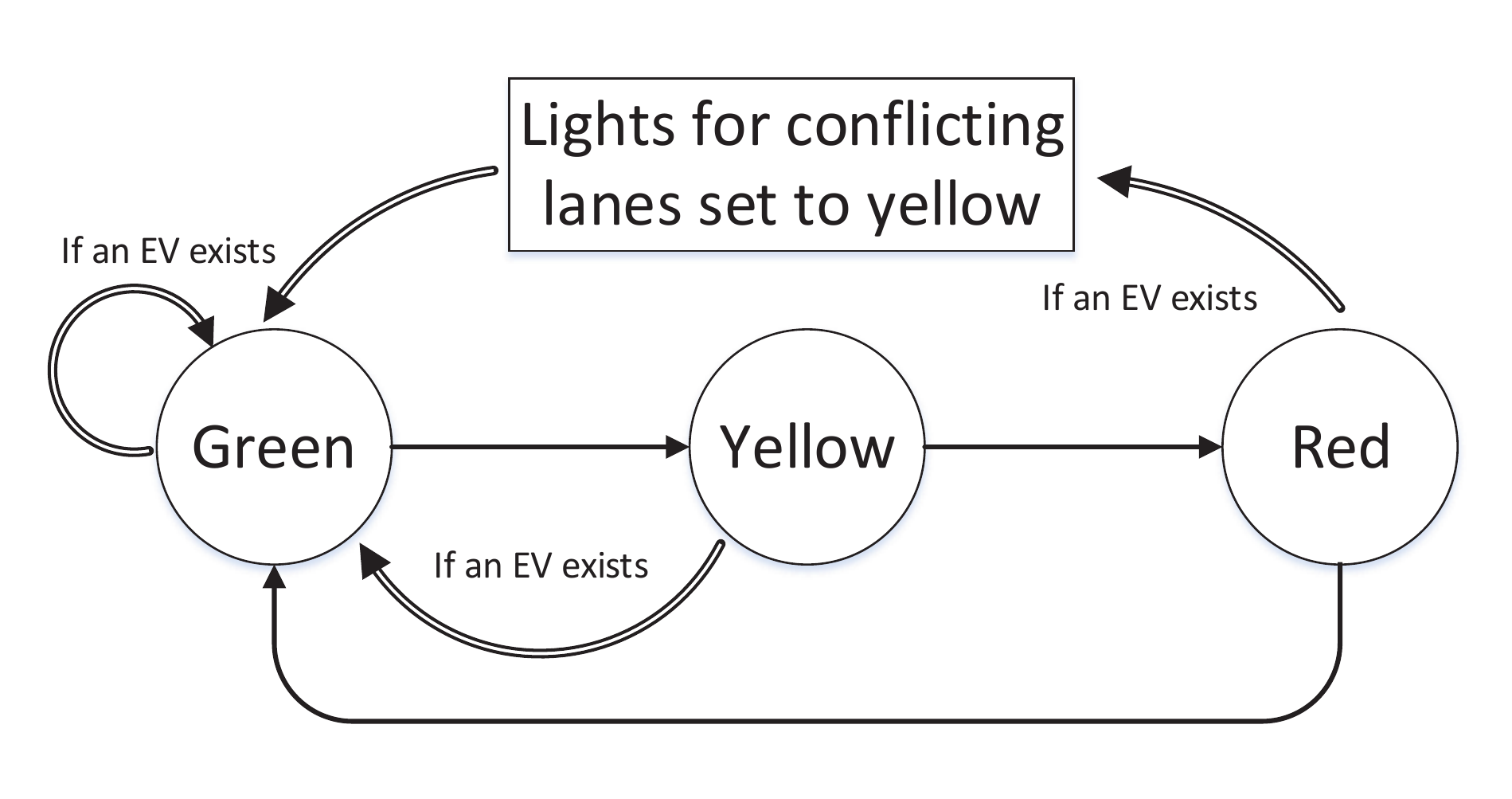}
	
	\caption{Reactive traffic light diagram.}
	
	\label{fig:RTL}
	
\end{figure}

\subsection{Reactive Traffic Light}
To show the effectiveness of the proposed R-DICA for emergency vehicles, a reactive traffic light algorithm for emergency vehicles is implemented and tested. As shown in Figure \ref{fig:RTL}, the traffic light for the lane of an EV changes to green as quickly as possible when an EV is detected on the boundary of communication region. Arrows with single line represent the state transitions (i.e. conventional traffic light algorithm) when there is no EV inside the communication region while arrows with double lines show the actions the algorithm will perform if an EV exists. The conventional traffic light algorithm we used is the default traffic light implemented in SUMO which has 31, 13 and 83 seconds durations for green, yellow and red light phases respectively. As shown in Figure \ref{fig:RTL}, if the current status is yellow or green when an EV is detected, the algorithm changes the light back to green or just extends the green light for a fixed amount of time respectively. If the current status of the lane is red when an EV enters the communication region, the algorithm immediately sets the green lights of conflicting lanes to yellow and then the lane of the EV will have green light after the yellow phase of the conflicting lanes. This augmentation of reactive mechanism in traditional traffic light system certainly help an EV to cross an intersection as quickly as possible.

\begin{figure}[!b]
\centering
\includegraphics[width=5.8in]{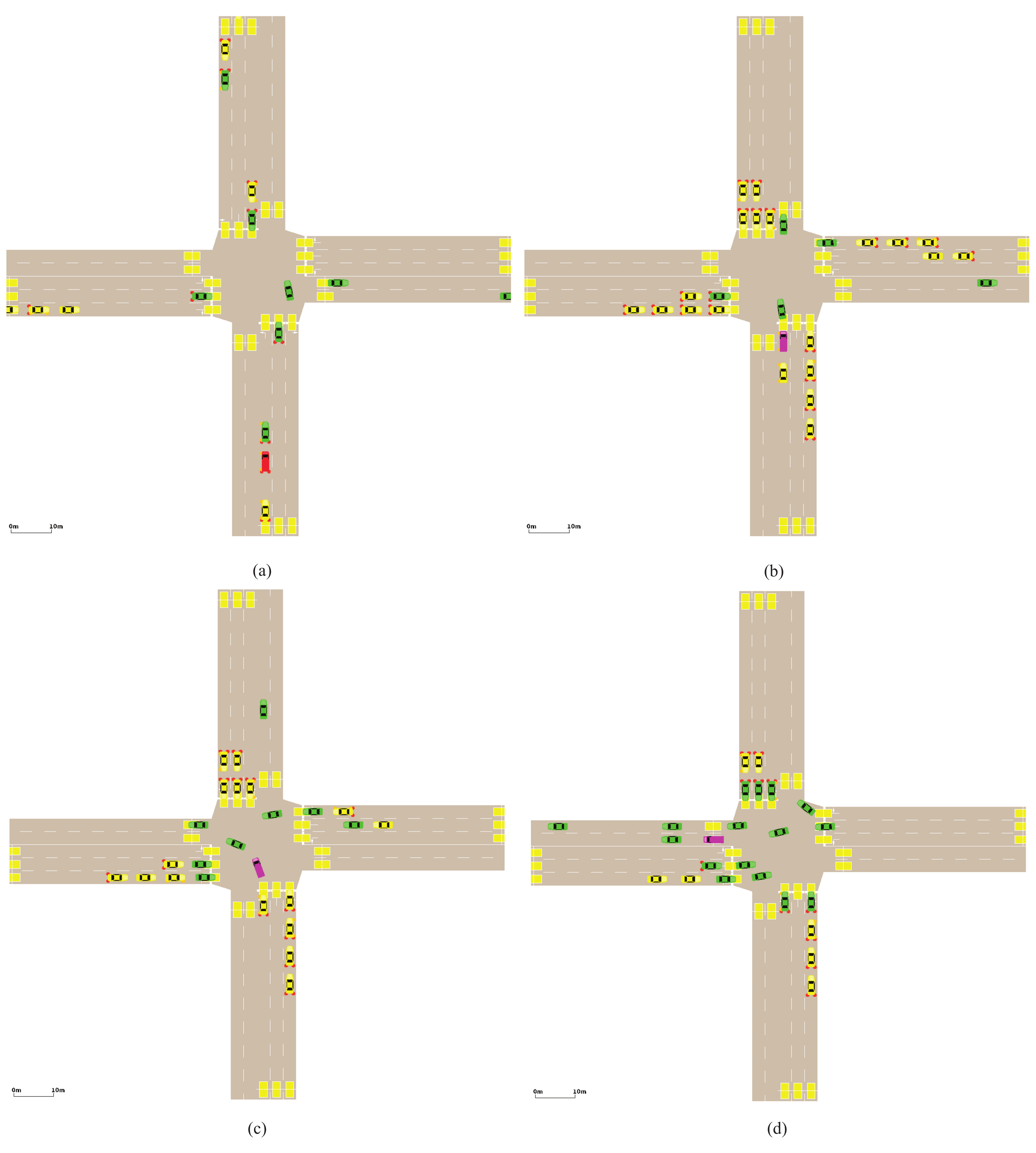}
\caption[A series of screenshots of simulation]{A series of screenshots of simulation which illustrates a situation when an EV is crossing the intersection.}
\label{fig:screenshot}
\end{figure}

\subsection{Simulation Results} \label{sec:res}
Performances of three different traffic patterns for all five volume cases are recorded from simulations. 
Figure \ref{fig:screenshot} shows a series of screenshots of simulation employing R-DICA in SUMO when an EV (the vehicle in red) is crossing the intersection from South. In the simulation, normal vehicles in yellow are not confirmed by ICA while green normal vehicles are the confirmed ones. In Figure \ref{fig:screenshot} (a), we can see that R-DICA activated GA algorithm which establishes an optimal order of vehicles to expedite the crossing of the EV. As one can note that in Figure \ref{fig:screenshot} (a) and (b) the head vehicle on the right lane of West road is not confirmed which means that this vehicle has a lower priority than the EV. All vehicles whose DTOTs cannot be modified are confirmed vehicles. And head vehicles who have a higher priority than the EV are confirmed. Figure \ref{fig:screenshot} (b) and (c) show that the EV is crossing the intersection unhindered while lower priority vehicles are waiting before the intersection. As shown in Figure \ref{fig:screenshot}, as soon as the EV exits the intersection, all head vehicles get confirmed which means R-DICA operates the same way as DICA. The optimal vehicle-passing sequence from the genetic algorithm ensures the fast crossing an intersection for the EV.

\subsubsection{Computation Time}

To show the computational efficiency of R-DICA using GA, we implemented R-DICA in two different versions: One with GA and the other one with the Exhaustive Search (ES) method to solve the optimization problem. Computation times of different volume cases are recorded for both methods. Simulation results are shown in Table \ref{computationTime} where `N/A' means that the computer was not able to complete the simulation due to memory errors.

\begin{table}[h]
	\centering
	\vspace{-0.5cm}
	\caption{Computation time comparison between ES and GA}
	\smallskip
	\label{computationTime}
	\begin{tabular}{cccccc}
		\toprule 
		Traffic volume &  \multirow{2}{*}{100} &  \multirow{2}{*}{200}& \multirow{2}{*}{300} & \multirow{2}{*}{400} & \multirow{2}{*}{500} \\ 
		(Number of vehicles per 10 minutes) &  & &  &  &  \\
		
		\midrule

		Computation time of ES (h) & 0.05 & 0.52 & N/A & N/A & N/A \\ 
		Computation time of GA (h) & 0.05 & 0.15 & 0.35 & 0.65 & 0.72 \\ 
		
		\bottomrule 
		
	\end{tabular} 
\end{table}

From the result, we can see that R-DICA with GA has definite advantages over R-DICA with ES in terms of computational efficiency. As shown in the table, the exhaustive search method only works for light traffic volumes while it has memory issues for traffics of higher volumes. 
%The maximum computation time of 0.72 hour for one hour simulation time implies that R-DICA with GA is possible to be applied in real-time.

\subsubsection{Performance of EVs}

The following performance measures are obtained to compare the performance of R-DICA with DICA and the reactive traffic light: a vehicle's \emph{trip time} ($\tau$) is defined as the time taken for a vehicle from the moment when it enters into the communication region of an intersection until the vehicle exits the intersection. Based on the measurement of $\tau$ for all crossed vehicles, we obtained the \emph{average trip time} ($\bar{\tau}$) and the \emph{maximum trip time} ($\tau_m$) which show the performance of the crossed vehicles. Besides these performance measures, we also calculated \emph{throughput} ($\rho$) which is the percentage of all crossed vehicles against total number of generated vehicles. 
We calculated the rate of average trip time to throughput, which we call the \emph{effective average trip time} ($\bar{\tau}_e$). Detailed explanation for this metric can be found in \cite{lu2016intelligent}.

\begin{figure}[!b]
	\centering
	\includegraphics[width=5.7in]{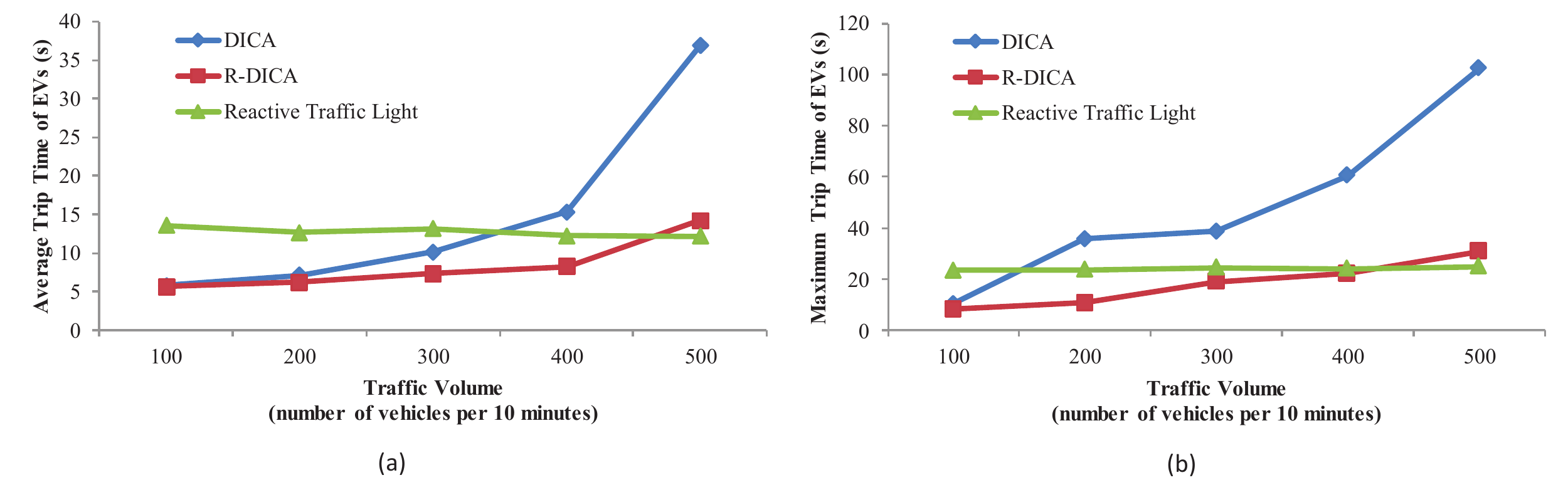}
	\caption[Performance comparison of EVs]{Performance comparison of EVs for DICA, R-DICA, and the reactive traffic light: (a) average trip time, (b) maximum trip time}
	\label{fig:EV}
\end{figure}

As shown in Figure \ref{fig:EV}, the average trip times of EVs in all three algorithms: DICA, R-DICA and the reactive traffic light are compared. 
For traffic volumes from 100 to 400, R-DICA has least average trip time of EVs than the other two algorithms. Especially in light traffic volumes, R-DICA reduces the EVs' average travel time by more than 50\% from the reactive traffic light. However, the algorithm has a bit longer average trip time for EVs than that of reactive traffic light in 500 traffic volume. 
The worst case for EVs' travels is illustrated by the maximum trip time of EVs in (b) of Figure \ref{fig:EV}. The maximum trip time of R-DICA increases and becomes greater than that of the reactive traffic light. 
Both average trip time and maximum trip time of EVs for DICA are increasing along the volumes. One may note that the average trip time and the maximum trip time of the reactive traffic light keep almost the same with the increase of the traffic volume. Through observation of the simulations, part of this is because too many vehicles accumulate before the intersection when the lane of the EV is under red light. At this situation, the EV is not detected and is stopping outside the communication region. When the light for the lane of the EV turns green, the EV accelerates from rest to enter the communication region which results in a higher speed for the EV. Thus, for the heavier traffic volumes, EVs always have a higher speed when detected and are expedited to cross by preference. The trip time within the communication region is then reduced compared with R-DICA. 

\subsubsection{Performance of Normal Vehicles}

\begin{figure}[!t]
\centering
\includegraphics[width=5.7in]{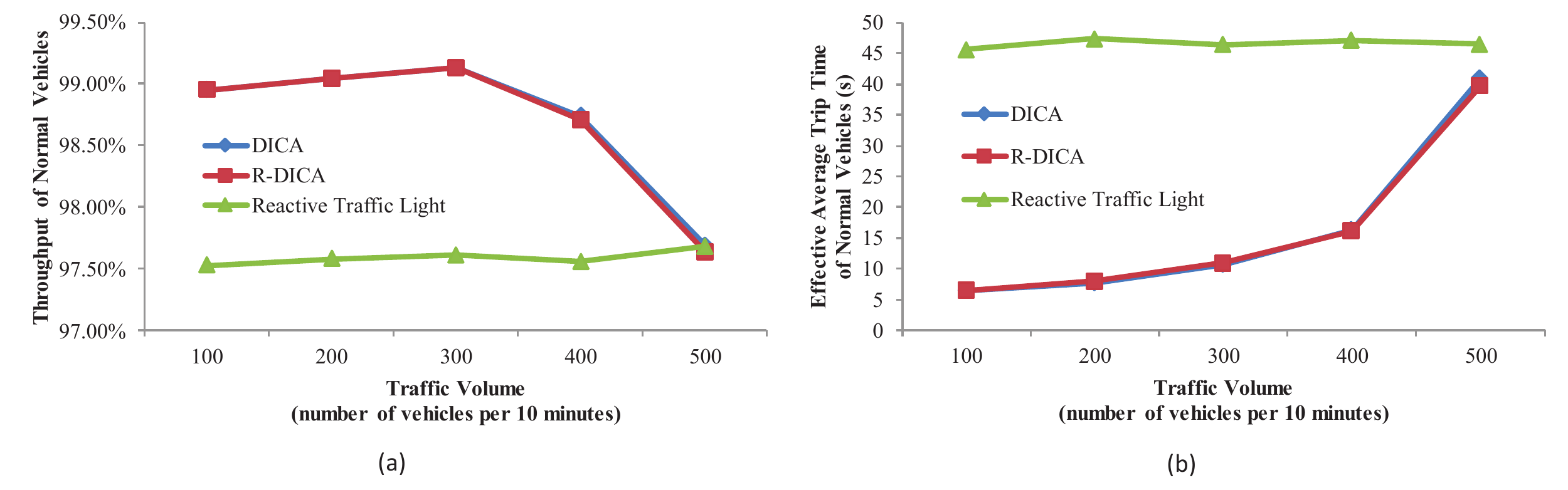}
\caption[Performance comparison of normal vehicles]{Performance comparison of normal vehicles for DICA, R-DICA, and the reactive traffic light: (a) throughput, (b) effective average trip time}
\label{fig:normal}
\end{figure}

Comparison of the performance for normal vehicles for all three algorithms is shown in Figure \ref{fig:normal}. 
From this result, we can see that the throughput and effective average trip time of R-DICA are nearly the same as those of DICA which shows that the performance of normal vehicles is minimally affected by EVs. Both the throughput and effective average trip time of normal vehicles become worse with the increase of traffic volumes. This is consistent with the result in our previous work \cite{lu2016intelligent, ieee}. Also, one can see from Figure \ref{fig:normal} that the reactive traffic light has steady and worse performance for normal vehicles than the other two algorithms.

To investigate more about the negative effect of prioritizing EVs on other normal vehicles, we compare the maximum trip time of normal vehicles for DICA and R-DICA in Table \ref{table}. The maximum trip time of R-DICA is very close to that of DICA and their difference increases with traffic volumes. This shows that it will bring more negative effect on normal vehicles to evacuate an EV in congested traffic.

\begin{table}[h]
	\centering
	\vspace{-0.5cm}
	\caption{Comparison of maximum trip times of normal vehicles between DICA and R-DICA}
	\smallskip
	\label{table}
	\begin{tabular}{cccccc}
		\toprule 
		Traffic volume &  \multirow{2}{*}{100} &  \multirow{2}{*}{200}& \multirow{2}{*}{300} & \multirow{2}{*}{400} & \multirow{2}{*}{500} \\ 
		(Number of vehicles per 10 minutes) &  & &  &  &  \\
		
		\midrule

		Maximum trip time of  & \multirow{2}[1]{*}{17.88}& \multirow{2}[1]{*}{35.80} & \multirow{2}[1]{*}{55.37} & \multirow{2}[1]{*}{89.47} & \multirow{2}[1]{*}{132.90} \\
		normal vehicles: DICA (s) \\ 
		Maximum trip time of & \multirow{2}[1]{*}{18.62} & \multirow{2}[1]{*}{38.52} & \multirow{2}[1]{*}{61.68} & \multirow{2}[1]{*}{99.15} & \multirow{2}[1]{*}{150.00} \\ 
		normal vehicles: R-DICA (s)\\
		\bottomrule 
		
	\end{tabular} 
\end{table}

\section{Summary} 

In this chapter, we have shown that the DICA algorithm can be augmented to allow emergency vehicles to cross intersections faster. A genetic algorithm based approach is proposed as part of the augmented algorithm, called R-DICA, to optimize the sequence of vehicles which gives the emergency vehicle the highest priority and keeps the influence on other vehicles' travel times as minimum as possible. The R-DICA operates the same way as DICA if there is no EV inside the communication region and optimizes vehicle-passing sequence if an EV enters the communication region. A reactive traffic light and DICA algorithms are also implemented for simulation and their results are compared with R-DICA to evaluate the performance of R-DICA. Simulation results show that R-DICA is effective to reduce travel times of EVs and has better performance than the reactive traffic light for normal vehicles.
We conclude that the performance of normal vehicles is not noticeably affected based on the simulation results of DICA and R-DICA.

 %input the file

\chapter{Conclusion and Future Work} \label{chapter:conclusion} % Conclusion
This dissertation has presented several control algorithms for connected and autonomous intersection traffic. 
%These algorithms ranging from basic DICA, reactive DICA for emergency vehicles' expedited crossings, to DICA-MIP which gives every vehicle an optimal trajectory. 
All algorithms are validated through extensive simulations. This chapter summarizes the dissertation briefly and gives potential future work.

\section{Conclusion} \label{sec:con}
In Chapter \ref{chapter:dica}, we developed an intelligent intersection control algorithm DICA employing the concept DTOT. V2I interaction protocol has been established for interactions between vehicles and intersection. The chapter introduced the concept of DTOT by which ICA is able to manage limited intersection space at a more accurate and efficient way. Theoretical analysis shows that DICA is free from deadlocks and starvation problems. Simulation results show that our algorithm achieves less Effective Average Trip Time compared with concurrent intersection control algorithms.

In Chapter \ref{chapter:e-dica}, we analyzed the computational complexity of the original DICA and enhanced the algorithm so that it can have better overall computational efficiency. The enhancement was done through several computational techniques like determining conflicting spaces offline, employing bisection method in time-conflict checking, etc. Simulation results show that the computational efficiency of the algorithm is improved significantly after the enhancement and the properties of starvation free and safety are guaranteed. We also validated that the overall throughput performance of our enhanced DICA is better than that of an optimized traffic light control mechanism in the case when the traffic is not congested.

In Chapter \ref{chapter:r-dica}, we have shown that the DICA algorithm can be augmented to allow emergency vehicles to cross intersections faster. A genetic algorithm based approach is proposed as part of the augmented algorithm, called R-DICA, to optimize the sequence of vehicles which gives the emergency vehicle the highest priority and keeps the influence on other vehicles' travel times as minimum as possible. The R-DICA operates the same way as DICA if there is no EV inside the communication region and optimizes vehicle-passing sequence if an EV enters the communication region. A reactive traffic light and DICA algorithms are also implemented for simulation and their results are compared with R-DICA to evaluate the performance of R-DICA. Simulation results show that R-DICA is effective to reduce travel times of EVs and has better performance than the reactive traffic light for normal vehicles.We conclude that the performance of normal vehicles is not noticeably affected based on the simulation results of DICA and R-DICA.

%In Chapter \ref{chapter:dica-mip}, we use MIP to formulate a new head vehicle's optimal control problem. Besides the constraints from vehicles and local laws, constraints from conflicting vehicles and front vehicles are also properly modeled. Great improvements from DICA are shown in simulation results which implies that the new algorithm DICA-MIP is able to reduce intersection congestions effectively.

\section{Future Work}

In the future, assumptions like perfect communication, accurate prediction of DTOT can be relaxed and methods to deal with car failures will be studied to make the algorithm more applicable to real situations.

In addition to giving priority to special vehicles e.g. emergency vehicles by forming optimal sequence like R-DICA, R-DICA can be enhanced to allow special vehicles' faster crossings through efficient usage of intersection space. For example, ICA may modify both occupancies’ positions and times of a vehicle’s DTOT in order to form a passage for special vehicles.

We will study algorithms based on DICA on a network of intersections to have a global optimal performance on a city level. Also, we will work to integrate the grouping strategy used in traffic flow based intersection control algorithms into our DICA to achieve better performances in more congested situations. DICA can be generalized to work with mixed traffic where autonomous vehicles and human-driven vehicles coexist. More interesting and difficult problem could be including pedestrians in the intersection traffic.

%We will study algorithms based on DICA-MIP on a network of intersections to have a global optimal performance on a city level. Also, we will work to integrate the grouping strategy used in traffic flow based intersection control algorithms into our DICA-MIP to achieve better performances in more congested situations. DICA-MIP can be generalized to work with mixed traffic where autonomous vehicles and human-driven vehicles coexist. More interesting and difficult problem could be including pedestrians in the intersection traffic.				
\newpage
%\appendix
\addcontentsline{toc}{chapter}{References}
\renewcommand{\bibname}{References}
\bibliographystyle{IEEEtran} % this can go anywhere in document and sets referral style and bib style
\bibliography{comprehensive}

%\appendix
%\input{Appendices}

\end{document}